\documentclass[11pt]{article}
\usepackage{mainsty}
\usepackage{amsmath}
\usepackage[utf8]{inputenc}
\usepackage{enumitem}
\usepackage{soul}
\usepackage{ulem}
\usepackage{booktabs,tabularx,array,pifont}
\newcommand{\cmark}{\ding{51}}
\newcommand{\xmark}{\ding{55}}
\newcommand{\pmark}{\(\sim\)}
\RequirePackage[capitalise, nameinlink, noabbrev]{cleveref}
\crefname{as}{Assumption}{Assumptions}
\RequirePackage{thm-restate}

\usepackage{setspace}

\newcommand{\ib}{{\color{green}\operatorname{ib}}}
\newcommand{\is}{{\color{blue}\operatorname{is}}}
\newcommand{\cc}{{\color{red}\operatorname{cc}}}

\let\tr\trmat
\newcommand{\tr}{{\color{black}\operatorname{tr}}}

\newcommand{\tB}{\textsf{B}}
\newcommand{\tS}{\textsf{S}}
\newcommand{\tBS}{\textsf{BS}}
\newcommand{\edavg}[1]{\hat{\dbar{#1}}}
\newcommand{\eavg}[1]{\hat{\bar{#1}}}
\DeclareMathOperator{\dir}{dir}
\DeclareMathOperator{\tot}{tot}

\usepackage{authblk}

\title{Regression Adjustments for Double Randomization in Two-Sided Marketplaces}

\author[1]{Timothy Sudijono\thanks{Email: \texttt{tsudijon@stanford.edu}}}
\author[2]{Lihua Lei}
\author[3]{Lorenzo Masoero}
\author[4]{Suhas Vijaykumar}
\author[2]{Guido Imbens}
\author[3]{James McQueen}

\affil[1]{\textit{Stanford Statistics}}
\affil[2]{\textit{Stanford GSB}}
\affil[3]{\textit{Amazon}}
\affil[4]{\textit{UCSD Economics}}
\date{\today}

\begin{document}
\maketitle

\begin{abstract}
Multiple randomization designs (MRDs) are a class of experimental designs used to handle interference in two-sided marketplaces. We investigate regression adjustment strategies for estimating total, spillover, and direct effects in MRDs. We derive minimum asymptotic variance estimators among a broad class of linearly adjusted estimators, without assuming a linear model on the potential outcomes. Surprisingly, the optimal regression adjustments are estimable from data and are generally different from regression adjustments in classical randomized experiments. For example, one such optimal estimator for the direct effect corresponds to a weighted regression with interacted two-way fixed effects. We establish model-robustness properties, central limit theorems, and inferential methods for our estimators, relying on improved theoretical results for MRD experiments. Our results provide the analog of classical regression adjustments for marketplace experiments. Numerical simulations demonstrate a considerable increase in efficiency over simpler approaches, enabling better inference when running MRDs. 
\end{abstract}


\section{Introduction}

Two-sided marketplaces are ubiquitous in the digital economy. In these settings, standard experimental designs (A/B tests) aimed at measuring the impact of an intervention can fall short due to \textit{interference}: interventions can lead to strategic interactions (or ``spillovers") within or across groups. Multiple Randomization Designs (MRDs) are a rich class of experimental designs \cite{bajari2021multiple, johari2022experimental} designed to handle interference in two-sided marketplaces. In MRDs, the outcome of interest is indexed by a buyer-seller pair and treatments can be implemented at this level. Under a \textit{local interference} assumption, \cite{bajari2021multiple, masoero2023leveraging} use MRDs to produce estimands that capture cross-unit interference and obtain corresponding estimates of spillover effects.  However, the existing MRD literature does not address how covariate adjustments can be leveraged to improve estimator precision, as is common in standard A/B tests \cite{deng2013improving}. This gap has practical consequences: real-world MRDs can suffer from low power due to explicit randomization over both buyers and sellers in the experiment, especially if one side of the marketplace is small. Incorporating covariates, such as the outcome $Y_{ij}$ from a previous time period, can help address this limitation by enabling more powerful inference. Our work fills the gap by developing regression adjustment theory and methods for MRDs, resolving a significant barrier to their use in practice.

\subsection{Background}

Fix a two-sided marketplace with $I$ buyers and $J$ sellers. For each buyer-seller pair $(i,j)$ we observe an outcome $y_{ij}$ of interest. Further, associated to each buyer-seller pair $(i,j)$ is a vector of covariates $X_{ij} \in \bR^d$, which may represent past observed outcomes, or individual buyer / seller level characteristics.

We will consider \textit{simple multiple randomization designs} (SMRDs). The SMRD assigns variables $W_i^{\tB} \in \set{0,1}$ to each buyer, as in a completely randomized experiment, with $\sum_{i=1}^I W_i^{\tB} = I_T$, the total number of treated buyers. Similarly each seller is assigned a binary variable $W_j^{\tS} \in \set{0,1}$, with $J_T$ total treated sellers. All buyer-seller pairs $(i,j)$ for which $W_i^{\tB}W_j^{\tS} = 1$ are given a binary treatment $W_{ij} = \text{T}$, while all other pairs serve as control (if $W_i^{\tB}W_j^{\tS} = 0$, then $W_{ij} = \text{C})$. An example of such a treatment matrix $\mathbf{W} = (W_{ij})_{ij}$ is given as follows. Here, we consider a marketplace of $4$ buyers and $8$ sellers, where the first two buyers and first four sellers are given treatment:
\begin{equation} \label{eq:assignment_matrix}
\mathbf{W} = 
\begin{bmatrix}
\textcolor{black}{\text{T}} & \textcolor{black}{\text{T}} & \textcolor{black}{\text{T}} & \textcolor{black}{\text{T}} & \textcolor{green}{\text{C}} & \textcolor{green}{\text{C}} & \textcolor{green}{\text{C}} & \textcolor{green}{\text{C}} \\
\textcolor{black}{\text{T}} & \textcolor{black}{\text{T}} & \textcolor{black}{\text{T}} & \textcolor{black}{\text{T}} & \textcolor{green}{\text{C}} & \textcolor{green}{\text{C}} & \textcolor{green}{\text{C}} & \textcolor{green}{\text{C}} \\
\textcolor{blue}{\text{C}} & \textcolor{blue}{\text{C}} & \textcolor{blue}{\text{C}} & \textcolor{blue}{\text{C}} & \textcolor{red}{\text{C}} & \textcolor{red}{\text{C}} & \textcolor{red}{\text{C}} & \textcolor{red}{\text{C}} \\
\textcolor{blue}{\text{C}} & \textcolor{blue}{\text{C}} & \textcolor{blue}{\text{C}} & \textcolor{blue}{\text{C}} & \textcolor{red}{\text{C}} & \textcolor{red}{\text{C}} & \textcolor{red}{\text{C}} & \textcolor{red}{\text{C}}
\end{bmatrix}
\end{equation}

The \textit{local interference} assumption of \cite{bajari2021multiple} is used to identify several estimands of interest. 
\begin{definition}[Local Interference Assumption]

Potential outcomes satisfy the \textit{local interference} assumption if $y_{ij}(\mathbf{W}) = y_{ij}(\mathbf{W}')$ for any pair $(i,j)$ such that 1) $W_{ij} = W'_{ij}$, 2) the fraction of treated sellers for buyer $i$ is the same for $\mathbf{W}, \mathbf{W}'$, and 3) the fraction of treated buyers for seller $j$ is the same for $\mathbf{W}, \mathbf{W}'$.
\end{definition}
Under the local interference assumption, each buyer-seller pair has four potential outcomes indexed by $\gamma \in \Gamma :=  \set{\tr,\ib,\is,\cc}$, defined as follows:

\[
\gamma_{ij}
=
\begin{cases}
    \tr & \text{if } W_i^\tB = 1, W_j^\tS = 1 \\
    \ib & \text{if } W_i^\tB = 1, W_j^\tS = 0 \\
    \is & \text{if } W_i^\tB = 0, W_j^\tS = 1 \\
    \cc & \text{if } W_i^\tB = 0, W_j^\tS = 0.
\end{cases}
\]
The $\tr$ group consists of directly treated pairs, while the $\cc$ group serves as a clean, homogeneous control in which both the buyer and seller are unexposed to treatment. The $\ib$ and $\is$ groups are control pairs ($W_{ij} = \text{C}$) whose buyer (resp.\ seller) is assigned treatment and therefore interacts with treated units on the other side of the marketplace; these three control groups ($\cc$, $\ib$, $\is$) are comparable prior to treatment, yet can differ post-treatment due to spillovers. Their comparison is what enables SMRDs to detect and measure interference.
Similar forms of interference were previously considered in \cite{manski2013identification} (cf.\ ``anonymous interactions'') and in \cite{hudgens2008toward} (cf.\ ``stratified interference'').

For a detailed example of a two-sided marketplace in which local interference is at play, we refer to \cite[Example 2.5]{masoero2024multiple}.
In that example, the authors present a two-sided platform model where content creators and advertisers interact. Each creator-advertiser pair generates revenue based on their respective quality scores and a compatibility factor. The platform compensates both parties through revenue-sharing contracts, while creators and advertisers incur costs to maintain quality. The model includes a binary intervention in the form of additional incentives that affect revenue generation.
The key idea is that agents' strategic responses to these incentives create local interference: the outcome of any creator-advertiser interaction depends not only on their direct treatment status, but also on the average treatment status of all their other interactions. This leads to equilibrium outcomes (both revenue and profit) that exhibit local interference structure. Generalizing beyond this example, local interference provides a natural model for potential outcomes in two-sided marketplaces with spillover effects.

\paragraph{Estimands.} We introduce the notation 
$
\dbar{y}_\gamma := \frac{1}{IJ}\sum_{i=1}^I \sum_{j=1}^J y_{ij}(\gamma), \gamma \in \Gamma
$
to denote the average of the potential outcomes for a given type $\gamma$. We are primarily interested in the causal estimands
\[
\tau_{c} = \sum_{\gamma \in \Gamma} c_\gamma\dbar{y}_{\gamma},
\]
for any vector $c = [c_{\tr},c_{\ib},c_{\is},c_{\cc}]$. Taking $c = [1,0,0,-1], [1,-1,-1,1]$ recover the \textit{total effect} and \textit{direct effect} respectively, while $c = [0,1,0,-1],[0,0,1,-1]$ measure \textit{buyer} and \textit{seller spillover effects} defined in \cite{bajari2021multiple, masoero2024multiple}. The buyer (or seller) spillover captures the effect of having a treated buyer (resp. seller) on a given interaction. Intuitively, the direct effect captures the average impact on a pair $(i,j)$ solely from the intervention. The total effect may be thought of as the unit-level average effect of the randomized treatment policy on the whole marketplace. It may be decomposed into the direct effect of the treatment plus the effect of buyer and seller spillovers, under the local interference assumption \cite{masoero2024multiple}. 

\subsection{Motivation: the direct effect}

To showcase our results and the importance of thinking carefully about regression adjustment, consider the direct effect estimator proposed in \cite{masoero2024multiple}:
\begin{equation}
\label{eq:direct_effect_estimator}
\hat{\tau}_{\dir} := \frac{1}{I_\tr J_\tr} \sum_{(i,j) \in \tr} y_{ij} - \frac{1}{I_\ib J_\ib} \sum_{(i,j) \in \ib} y_{ij}- \frac{1}{I_\is J_\is} \sum_{(i,j) \in \is} y_{ij}  + \frac{1}{I_\cc J_\cc} \sum_{(i,j) \in \cc} y_{ij}.
\end{equation}
This estimator can equivalently be written as the coefficient $\tau$ in the regression 
\begin{equation}
\label{eq:dir_unadj}
y_{ij} \sim \mu + \alpha W_i^\tB + \delta W_j^\tS + \tau W_i^\tB W_j^\tS.
\end{equation}
Thus, a seemingly natural way to incorporate covariates into the analysis is to consider the ordinary least squares (OLS) regression
\begin{equation}
\label{eq:dir_ols}
y_{ij} \sim \mu + \alpha W_i^\tB + \delta W_j^\tS + \tau W_i^\tB W_j^\tS + X_{ij}^\top \beta,
\end{equation}
which can be seen as an analog of Fisher's ANCOVA discussed in \cite{freedman2008regressionA, ding2024first} or the analog of CUPED from A/B Testing \cite{deng2013improving}.
The estimated coefficient $\hat\tau$ serves as a regression-adjusted estimate of the direct effect. The method serves as a good baseline; under mild assumptions, $\hat\tau$ is consistent for the direct effect. 

However, when the regression specification \eqref{eq:dir_ols} does not correspond to the true model for $y_{ij}$, there exists an estimator which is provably more efficient than ANCOVA, in large samples. The intuition is the following. Let $e_{ij} = y_{ij} - X_{ij}^\top \beta$ and define
\begin{equation}
\label{eq:direct_effect_general_imputation_estimator}
\hat\tau_{\dir}(\beta) := \frac{1}{I_\tr J_\tr} \sum_{(i,j) \in \tr} e_{ij} - \frac{1}{I_\ib J_\ib} \sum_{(i,j) \in \ib} e_{ij}- \frac{1}{I_\is J_\is} \sum_{(i,j) \in \is} e_{ij}  + \frac{1}{I_\cc J_\cc} \sum_{(i,j) \in \cc} e_{ij}.
\end{equation}

By inspecting the first order condition of the OLS, $\tau$ in  \eqref{eq:dir_ols} is equal to $\hat{\tau}_{\dir}(\hat \beta)$ for $\hat{\beta}$ minimizing the OLS condition. The form of  \eqref{eq:direct_effect_general_imputation_estimator} also suggests a more efficient estimator, by choosing $\beta$ to minimize the variance of  \eqref{eq:direct_effect_general_imputation_estimator}. A priori, it would seem that $\beta$ depends on unobservable potential outcomes, but we show that there exists a consistent estimator $\hat{\beta}$ of this optimal adjustment. This insight is implicit in \cite{li2017general, lu2022tyranny} in the cross-sectional case and it is also closely related to control variates from the Monte Carlo literature; we discuss such connections further in \cref{sec:literature_review}. This suggests the simple plug-in $\hat{\tau}_{\dir}(\hat{\beta})$. 

Surprisingly, the plug-in estimator for the direct effect is interpretable and equal to the coefficient $\tau$ in the \textit{weighted} least squares
\begin{align}
y_{ij} & \stackrel{w_{ij}}{\sim} (\mu + \mu^{\tB}_i + \mu^{\tS}_j) + (\alpha + \alpha^{\tB}_i + \alpha^{\tS}_j)W_i^{\tB}  + (\delta + \delta^{\tB}_i + \delta^{\tS}_j)W_j^{\tS} + (\tau + \tau^{\tB}_i + \tau^{\tS}_j)W_i^{\tB}W_j^{\tS} + X_{ij}^\top \beta\label{eq:twfe_regression} \\
w_{ij} & = \frac{1}{I_\gamma^2 J_\gamma^2}, (i,j) \in \gamma,\nonumber
\end{align}
under suitable identifiability constraints. Specifically, the buyer and seller fixed effects are constrained to sum to zero within each treatment group:
\begin{align*}
& \textstyle\sum_{i \in \cl{I}_\cc} \mu_i^\tB = \sum_{j \in \cl{J}_\cc} \mu_j^\tS = 0, \qquad
  \sum_{i \in \cl{I}_\cc} \delta_i^\tB = \sum_{j \in \cl{J}_\cc} \alpha_j^\tS = 0, \\
& \textstyle\sum_{i \in \cl{I}_\tr} (\mu_i^\tB + \alpha_i^\tB) = \sum_{j \in \cl{J}_\tr} (\mu_j^\tS + \delta_j^\tS) = 0, \\
& \textstyle\sum_{i \in \cl{I}_\tr} (\delta_i^\tB + \tau_i^\tB) = \sum_{j \in \cl{J}_\tr} (\alpha_j^\tS + \tau_j^\tS) = 0.
\end{align*}
These constraints ensure that $\tau$ captures the global average direct effect. This is a weighted difference-in-differences panel regression with interacted two-way fixed effects. Importantly, we do not assume the regression specification \eqref{eq:twfe_regression} gives the correct model. In fact, this weighted least squares estimator has the minimal asymptotic variance among the class \eqref{eq:direct_effect_general_imputation_estimator}, without any structural assumptions on the outcome distribution conditional on covariates.

The weights are closely related to the tyranny-of-minority regression \cite{lu2022tyranny}, which minimizes the asymptotic variance with the same coefficient $\beta$ shared across all groups in cross-sectional experiments. 
Empirically, using the optimal regression adjustment $\hat \tau_{\dir}(\hat \beta)$ can significantly improve estimation and inference. 
\cref{fig:intro_plots} shows a situation in which naive regression adjustment according to ANCOVA \eqref{eq:dir_ols} can actually hurt asymptotic variance compared to no adjustment at all. Similar results were demonstrated by Freedman \cite{freedman2008regressionA} in the cross-sectional case, for imbalanced trials. At the same time, the figure demonstrates that the optimal adjustment can be much more precise than naive adjustment. Further details on this simulation are discussed in \cref{sec:simulations}, with inferential results discussed therein.

\begin{figure}[H]
    \centering
    \includegraphics[width=0.7\linewidth]{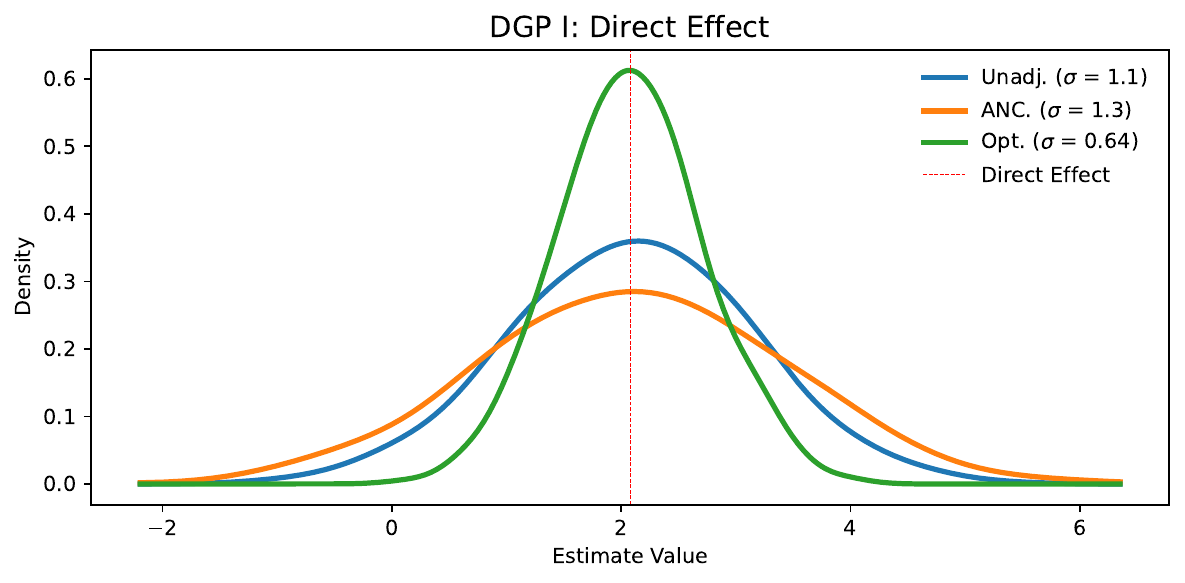}
    \caption{KDE estimate of sampling distributions for the direct effect estimators i) without adjustment, ii) with the ANCOVA adjustment, and iii) with the plug-in optimal adjustment.}
    \label{fig:intro_plots}
\end{figure}

\subsection{Contributions}

In this work, we primarily analyze regression adjustments of the form 
\begin{equation}
\label{eq:imputation_estimators}
\hat{\tau}_{c}(\beta) := \sum_{\gamma \in \Gamma} c_\gamma \frac{1}{I_\gamma J_\gamma} \sum_{(i,j) \in \gamma} \left(y_{ij} - X_{ij}^\top \beta\right),
\end{equation}
where $I_\gamma, J_\gamma$ denote the number of buyers and sellers respectively in a block $\gamma$ (cf.\ \eqref{eq:assignment_matrix}). This is equal to $I_T,J_T$ or $I_C,J_C$ depending on $\gamma.$ We will refer to these as \textit{non-interacted imputation estimators}. These imputation estimators generalize both the unadjusted estimators and naive methods of regression adjustment. Our contributions are the following.

\begin{enumerate}
    \item \textit{Optimal Regression Adjustments.} For general effects of interest, we derive the choice of $\beta$ which minimizes the variances of the imputation estimators, absent any structural assumptions on the potential outcomes. In all cases, the optimal choices of $\beta$ can be estimated from observed data; they do not depend on unidentifiable contrasts of potential outcomes. This leads to natural plug-in estimators which provably outperform the unadjusted estimators in terms of asymptotic variance. Furthermore, the optimal adjustments are generally different from ordinary least squares.
    \item \textit{Theory and Methods for Inference.} We show asymptotic normality of the regression adjustment strategy and accompanying conservative confidence intervals for the regression adjusted estimator.
    \item \textit{Empirical Performance.} In several simulation settings, we demonstrate efficiency gains of the optimal regression adjustment over other methods. The accompanying confidence intervals are also shorter while maintaining desired coverage.
    \item \textit{New Theoretical Tools for MRDs.} Our theoretical results for the plug-in estimators rely on improvements to the theory of doubly-randomized sums, which gives inference for vanilla MRDs. This is another key contribution of this paper. In particular, we give a new CLT for doubly-randomized sums in $1$-Wasserstein distance which shows asymptotic normality in settings where Theorem 4.6 of \cite{masoero2024multiple} does not apply. The CLT uses a novel approach based on Stein's method and concentration of measure tools for sampling without replacement. Moreover, we show the asymptotic consistency of the conservative variance of \cite{masoero2024multiple} under weaker assumptions, which match those imposed by our CLT.

\end{enumerate}

The results are proved in a purely finite population framework where the potential outcomes are treated as fixed and the only randomness comes from the explicit randomization of $W_i^\tB, W_j^\tS$. Previous work in the finite population framework \cite{lin2013agnostic, li2017general, lei2021regression, lu2022tyranny} typically assumes a sequence of finite populations of potential outcomes in which various quantities such as variances and population regression adjustments converge to limits. For example, consider the difference-in-means estimator for a finite population $\set{Y_i(1), Y_i(0)}_{i=1}^n$, when sampling $n_0$ units without replacement. Its variance is well known and given by
\begin{equation}
\frac{1}{n_0} S_0^2 + \frac{1}{n_1} S_1^2 - \frac{1}{n} S_{\tau}^2
\end{equation}
with $n_1 = n - n_0,$ $S_i^2 = \frac{1}{n-1} \sum_{j=1}^n (Y_j(i) - \bar{Y}(i))^2$ for $i \in \set{0,1}$, and similarly $S_\tau^2$ denotes the finite-population variance of the individual treatment effects. By way of a super-population setup and the law of large numbers, the quantities $S_0^2, S_1^2, S_\tau^2$ are presumed to have non-zero constant limits. Therefore, the variance of the difference-in-means estimator scales like $\Theta(1/n)$.

The analogous situation in the MRD setting is less straightforward, comprising one challenge in this work. The variance of the direct effect estimator \eqref{eq:direct_effect_estimator} is more complicated and is recorded in \eqref{eq:direct_effect_variance_formula}. Roughly the formula is of the form
\begin{equation}
\label{eq:intro_var_formula}
\frac{1}{I}\nu^\tB + \frac{1}{J}\nu^\tS + \frac{1}{IJ}\nu^\tBS,    
\end{equation}
where $\nu^\tB, \nu^\tS, \nu^\tBS$ are linear combinations of finite population variance quantities. Assuming that these quantities have constant limits raises a problem here, as then $\Var(\hat{\tau}_{\dir})$ can either scale like $\Theta(I^{-1} + J^{-1})$, $\Theta((IJ)^{-1})$, or in-between as $\Theta(I^{-\alpha}), \alpha \in [1,2].$ Example \ref{ex:iid_normal_superpop} details one such situation. This challenge has not been highlighted before in previous work; \cite{masoero2024multiple} primarily requires that $\Var(\hat{\tau}_{c}) = \Omega(I^{-1} + J^{-1})$ for their results to hold. 

A comprehensive analysis of MRDs in all possible scaling regimes is a significant research program; we take the first step in this direction. For our inferential results on the plug-in regression adjustment procedure $\hat{\tau}_c(\hat{\beta}_c)$, we make use of the previous assumption that $\Var(\hat{\tau}_{c}) = \Theta(I^{-1} + J^{-1})$ and in some important examples, show how we can go beyond this.  These inferential results rely on improved theoretical tools for MRDs without covariates which apply to the regime where $\Var(\hat{\tau}_{c}) = \Theta((IJ)^{-1})$. We detail theoretical questions for future research in \cref{sec:discussion} and note some avenues for improvement. Regardless, our novel estimators have clear practical utility and our analysis covers a suite of interesting examples.

Another natural class to consider is the class of \textit{interacted} imputation estimators:
\begin{equation}
\label{eq:general_interacted_estimator}
\hat{\tau}_{c, \text{int}}(\bm{\beta}) := \sum_{\gamma \in \Gamma} c_\gamma \frac{1}{I_\gamma J_\gamma} \sum_{(i,j) \in \gamma} \left(y_{ij} - X_{ij}^\top \beta_\gamma\right)
\end{equation}
where a separate $\beta_\gamma$ is used for each group and $c = [c_{\tr},c_{\ib},c_{\is},c_{\cc}] = [1,-1,-1,1].$ These strictly nest the estimators of \eqref{eq:imputation_estimators}. The same insight is true here: one can estimate the $\beta_\gamma$'s which minimize the $\Var(\hat{\tau}_{\dir}(\bm{\beta}))$ and define a plug-in estimator. As the discussion of the interacted estimators requires more complicated notation, we will only discuss this class further in \cref{sec:interacted_estimators}. Our techniques for the non-interacted case extend naturally to the interacted case. Moreover, the non-interacted estimators \eqref{eq:direct_effect_general_imputation_estimator} are worth studying in their own right, especially if one of the four MRD groups is small relative to the number of covariates. In such cases, \eqref{eq:general_interacted_estimator} may be unstable compared to \eqref{eq:imputation_estimators}. \cite{lu2022tyranny} details several practical and theoretical benefits of non-interacted versus interacted adjustments in cross-sectional experiments.

\subsection{Related Work}
\label{sec:literature_review}

\paragraph{Online Experimentation.} Online randomized experiments (A/B tests) have become a standard tool to drive innovation in the technology industry \cite{bojinov2022online, quin2024b}.
In the data analysis phase, experimenters often perform regression adjustment with pre-experimental covariates. This yields estimates that can be significantly more reliable than those obtained without adjustment \cite{masoero2023leveraging}.

Standard practice in online experimentation is to expose a completely random subset of units to treatment.  However, this standard design falls short when interventions (i) differentially impact multiple populations of units (e.g., an intervention on an online service that influences interactions between buyers and sellers) and/or (ii) lead to strategic behavior that affects the measured outcomes within or across these groups (e.g., the treatment causes buyers and sellers to adjust their behavior in a manner that influences their other, untreated interactions). Acknowledging the novel challenges posed by online experiments and the inability of standard experiments to address them, a number of novel estimators and experimental designs have been proposed over the last few years. 
In particular, a number of authors have focused on proposing experimental designs that are robust to settings in which there is \emph{interference} -- of the sort outlined above -- across units in the experiment \cite{bajari2023experimental, savje2021average, savje2023causal, bojinov2022online}.

\paragraph{Multiple Randomization Designs.} Among these proposals, Multiple Randomization Designs (MRDs) have emerged as a flexible and rich class of designs \cite{johari2022experimental, bajari2023experimental} for two-sided marketplaces of buyers and sellers, where the outcome of interest is indexed by a buyer-seller pair. \cite{johari2022experimental} motivates the design as a way to reduce bias, under a Markov Decision Process model. \cite{bajari2021multiple} adopt the finite-population potential outcome framework to show how --- under a general notion of interference across units, referred to as \textit{local interference} --- it is possible to use MRDs to define estimands and corresponding ``unadjusted'' estimators to capture cross-unit interference and estimate spillover effects. 
Exact variances and corresponding variance estimators are also provided.
While flexible, a missing piece and open question in the existing MRD framework is how covariate adjustments may be leveraged to enhance precision, similarly to what is done for ``standard'' experiments \cite{deng2013improving, lin2013agnostic}.

Beyond the standard buyer-seller marketplace, MRDs have found applications in a variety of experimental settings. When the second dimension corresponds to time, the design reduces to a \emph{switchback} (or crossover) experiment, where units alternate between treatment and control across time periods; such designs have been studied and deployed at scale in industry \cite{bojinov2020design, missault2025robust}. MRDs can also be combined with cluster randomization, where one side of the marketplace is randomized at the level of groups of units rather than individuals \cite{hut2024measuring}.

\paragraph{Regression Adjustments in Experiments.}

For standard randomized control trials with two treatment groups, regression adjustment has a long history. Most relevant to this work are \cite{freedman2008regressionA,  lin2013agnostic, li2017general, lu2022tyranny}. To explain the story and how it relates to our results, a little notation is helpful. Define $Y_i(0),Y_i(1), i=1,\dots,n$ potential outcomes, $Z_i$ the treatment indicator, and $X_i \in \bR^d$ associated covariates. $Z_i$ are indicators of sampling without replacement such that $\sum_i Z_i = n_1.$

The seminal work \cite{lin2013agnostic} addresses previous complaints of regression adjustment by Freedman \cite{freedman2008regressionA, freedman2008regressionB}. Freedman argued that regression adjustments of the form $Y_i \sim Z_i + X_i$ can hurt asymptotic variance compared to regular difference-in-means $Y_i \sim Z_i$ especially when the treatment and control groups are imbalanced. In \cite{lin2013agnostic}, it is argued that using the interacted regression $Y_i \sim X_i + Z_i + X_i Z_i$ addresses these concerns. Lin's interacted regression solution is particularly satisfying. Firstly, it is model-robust, not requiring the linearity specification to hold. Secondly, it always improves on the asymptotic variance of both difference-in-means and naive regression adjustment. Thirdly, conservative inference is easily done using robust standard errors in the interacted regression. Moreover, it is shown that a certain weighted regression, the \textit{tyranny-of-minority} (ToM) regression, also has the same asymptotic properties as Lin's interacted regression. These results are all justified in a finite population framework, where the potential outcomes are treated as fixed and the only randomness is from the randomization $Z_i$.

The results of \cite{lin2013agnostic} appear mystical; fortunately, several works \cite{li2017general, ding2024first, guo2023generalized} elucidate the interacted regression. The idea is to consider the imputation estimator
\begin{equation}
\label{eq:cross_sectional_imputation}
\frac{1}{n_1} \sum_{i=1}^n Z_i (Y_i - X_i^\top \beta_1) - \frac{1}{n - n_1} \sum_{i=1}^n (1-Z_i) (Y_i - X_i^\top \beta_0)
\end{equation}
akin to \eqref{eq:general_interacted_estimator}. When the covariates are centered, the form of \eqref{eq:cross_sectional_imputation} suggests that it should be consistent for the direct effect. Therefore, $\beta_1,\beta_0$ should be chosen in order to minimize the asymptotic variance of \eqref{eq:cross_sectional_imputation}. In Example 9 of \cite{li2017general}, it is shown via properties of least-squares that the optimal choices are the population OLS coefficients in the treatment and control groups separately. Using the plug-in estimates of these OLS coefficients, one recovers Lin's interacted linear regression. One can also ask the question: what if the coefficients $\beta_0 = \beta_1 =: \beta$ in \eqref{eq:cross_sectional_imputation} are equal? Such estimators are natural especially when $n_0,n - n_0$ are small relative to $d.$ As \cite{lu2022tyranny} shows, the optimal $\beta$ can be estimated and the plug-in returns the ToM regression. Because this class of imputation estimators nests the unadjusted estimator and the regression adjustment $Y_i \sim Z_i + X_i$, both Lin's regression and ToM have better asymptotic variance.

Our results echo these cross-sectional findings and also highlight unexpected differences. Returning to the MRD setting, the analog of ANCOVA \eqref{eq:dir_ols} can hurt asymptotic variance especially in imbalanced experiments, echoing Freedman's finding. Consider again the imputation estimator in \eqref{eq:imputation_estimators}. We also show that optimal regression adjustment $\beta$ is still estimable in this complex setting, even though a priori it depends on unobservable potential outcomes. A further inspection of the first order condition for the optimal $\beta$ shows that this should be a generic phenomenon. Using the resulting plug-in gives an interpretable regression in some cases, but not others. For the direct effect, one obtains the weighted interacted two-way fixed effect regression in \eqref{eq:twfe_regression} whose weights resemble those of ToM. The story is similar for the interacted regression \eqref{eq:general_interacted_estimator}, revealing optimal regression adjustments that are different than running separate OLS's in each group and cannot generally be written in a clean regression form. \cref{tab:mrd-vs-cross-sectional} summarizes this discussion and presages some further phenomena discussed throughout the text. 

Similar investigations on model-robustness of various regressions have been carried out for other designs. For example, \cite{su2021model, wang2026mixed} analyze the cluster randomized design. \cite{cytrynbaum2024covariate, zhao2022reconciling,arkhangelsky2024design} respectively analyze stratified experiments,  split-plot designs, and panel data settings.

\begin{table}[t]
\centering
\caption{Parallels and differences between MRDs and classical regression adjustments (A/B).}
\label{tab:mrd-vs-cross-sectional}
\renewcommand{\arraystretch}{1.15}
\begin{tabularx}{\textwidth}{
    >{\raggedright\arraybackslash}p{0.56\textwidth}
    >{\centering\arraybackslash}p{0.18\textwidth}
    >{\centering\arraybackslash}p{0.18\textwidth}
}
\toprule
\textbf{Phenomenon} & \textbf{MRD} & \textbf{A/B} \\
\midrule

No modeling assumptions on potential outcomes
& \cmark
& \cmark \\
\addlinespace[0.4em]

ANCOVA can harm asymptotic precision
& \cmark
& \cmark \\
\addlinespace[0.4em]

Optimal adjustments are estimable
& \cmark
& \cmark \\
\addlinespace[0.4em]

Optimal interacted adjustment has regression form
& \xmark
& \cmark\ (Lin) \\
\addlinespace[0.4em]

Optimal non-interacted adjustment has regression form
& \pmark\ (Direct Effect)
& \cmark\ (ToM) \\
\addlinespace[0.4em]

Optimal interacted and non-interacted adjustments are asymptotically equivalent
& \xmark
& \cmark \\
\addlinespace[0.4em]

Robust standard errors from the corresponding regression directly yield conservative inference
& \xmark
& \cmark \\

\bottomrule
\end{tabularx}

\vspace{0.35em}
\begin{minipage}{0.96\textwidth}
\footnotesize
\textit{Notes.}
\pmark\ indicates a partial result. For MRDs, the direct effect admits a weighted TWFE regression interpretation, while the total and spillover effects generally do not admit an equally simple regression representation.
\end{minipage}
\end{table}

\subsection{Further Setup \& Notation}

Before describing our results, let us establish some notation and assumptions. Let $\cl{S} = \sigma\set{(W_j^\tS)_j}$ denote the sigma algebra generated by the seller treatment indicators, and $\cl{B} = \sigma\set{(W_i^\tB)_i}$ similarly. Recall $\Gamma :=  \set{\tr,\ib,\is,\cc}$.  We will occasionally use $W^\tB, W^\tS$ to denote the vector of the treatment indicator variables. For each group $\gamma \in \Gamma$, we will write $\cl{I}_\gamma$ to denote the buyer indices in group $\gamma,$ and similarly with $\cl{J}_\gamma$. For example, $\cl{I}_\tr = \set{i : W_i^\tB = 1}, \cl{J}_\tr = \set{j : W_j^\tS = 1}$. We write $(i,j) \in \gamma$ to denote $i \in \cl{I}_\gamma, j \in \cl{J}_\gamma$ and occasionally we will use the shorthand $i \in \gamma$ and $j \in \gamma$ to mean $i \in \cl{I}_\gamma, j \in \cl{J}_\gamma$ respectively. Throughout the paper we will use $A \lesssim B$ to denote $A \leq CB$ for an absolute constant $C$ not depending on any asymptotic scaling. Throughout this paper, we will assume that 
\begin{equation}\label{eq:asymptotic_regime}
\lim_{n\rightarrow \infty} \frac{I_T}{I} \in (0, 1), \quad \lim_{n\rightarrow \infty} \frac{J_T}{J}\in (0, 1), \quad \lim_{n\rightarrow \infty} \frac{I}{J}\in (0, \infty).
\end{equation}
Implicitly, each of $I_T, J_T, I,J$ is indexed by $n \uparrow \infty$ which enumerates a sequence of finite populations.

\paragraph{Matrix notation.} 

For each $\gamma \in \Gamma$, we define the following quantities:
\begin{align*}
\bar{y}_{i,\bullet}(\gamma) & := \frac{1}{J}\sum_{j=1}^J y_{ij}(\gamma), \quad \bar{y}_{i,\bullet}^{(d)}(\gamma) = \bar{y}_{i,\bullet}(\gamma) - \dbar{y}_\gamma \\
\bar{y}_{\bullet,j}(\gamma) & := \frac{1}{I}\sum_{i=1}^I y_{ij}(\gamma) ,
\quad \bar{y}_{\bullet,j}^{(d)}(\gamma) = \bar{y}_{\bullet,j}(\gamma) - \dbar{y}_\gamma \\
y_{ij}^{(dcr)}(\gamma) & := y_{ij}(\gamma) - \bar{y}_{i,\bullet}(\gamma) - \bar{y}_{\bullet,j}(\gamma) + \dbar{y}_\gamma.
\end{align*}
We will call the operation in the last line \textit{double-decentering} a matrix, forcing row and column means to be zero. Define the following variance quantities for the potential outcomes:
\begin{align}
\label{eq:sigma_terms}
\nu_{\gamma}^{\tB} & := \frac{1}{I-1}\sum_{i=1}^I \bar{y}^{(d)}_{i,\bullet}(\gamma)^2, \ \nu_{\gamma}^{\tS} := \frac{1}{J-1}\sum_{j=1}^J \bar{y}^{(d)}_{\bullet,j}(\gamma)^2, \ \nu_{\gamma}^{\tBS} := \frac{1}{(I-1)(J-1)}\sum_{j=1}^J\sum_{i=1}^Iy_{ij}^{(dcr)}(\gamma)^2.
\end{align}
\begin{align}
\xi_{\gamma,\gamma'}^{\tB} & := \frac{1}{I-1}\sum_{i=1}^I \left( \bar{y}^{(d)}_{i,\bullet}(\gamma) 
 - \bar{y}^{(d)}_{i,\bullet}(\gamma') \right)^2, \ \xi_{\gamma,\gamma'}^{\tS} := \frac{1}{J-1}\sum_{j=1}^J \left(\bar{y}^{(d)}_{\bullet,j}(\gamma) - \bar{y}^{(d)}_{\bullet,j}(\gamma')\right)^2, \\\xi_{\gamma,\gamma'}^{\tBS} & := \frac{1}{(I-1)(J-1)}\sum_{j=1}^J\sum_{i=1}^I \left[y_{ij}^{(dcr)}(\gamma) -y_{ij}^{(dcr)}(\gamma')\right]^2.
\end{align}

Given a matrix $(A_{ij})_{ij}$, we will use the notation $\nu^\theta(A)$ for $\theta \in \set{\tB,\tS,\tBS}$ to denote the analogues of the above equations with $A$ replacing $y$.

\section{Optimal Non-Interacted Regression Adjustments}
\label{sec:methodology}
Returning to regression adjustments for MRDs, we present the general methodology for generic estimands of interest
\[
\tau_{c} = \sum_{\gamma \in \Gamma} c_\gamma\dbar{y}_{\gamma},
\]
given any vector $c = [c_{\tr},c_{\ib},c_{\is},c_{\cc}]$. We consider the non-interacted imputation estimator in \eqref{eq:imputation_estimators}:
\begin{equation}
    \tau_{c}(\beta) = \sum_{\gamma \in \Gamma}  c_\gamma\frac{1}{I_\gamma J_\gamma} \sum_{(i,j) \in \gamma}  \left(y_{ij} - X_{ij}^\top \beta\right)
\end{equation}

To describe the form of the optimal regression adjustment $\tilde{\beta}_c$ that minimizes $\Var(\tau_{c}(\beta)),$ we first claim that the variance can be written as a positive definite quadratic form given by 
\[
\Var(\tau_{c}(\beta)) = \beta^\top Z \beta - 2u^\top \beta + \Var(\tau_{c}),
\]
where $Z$ is an appropriate linear combination of $Z^\tB,Z^\tS,Z^\tBS$ and $u$ is an appropriate linear combination of $u_\gamma^\tB, u_\gamma^\tS, u_\gamma^\tBS$, to be described. By inspecting the variance and covariance formulas in \cref{prop:generic_variance_formula,prop:generic_covariance_formula}, $\Var(\tau_{c}(\beta))$ can be decomposed as the linear combination of terms $\tilde\nu^\theta_{\gamma}$ for $\theta \in \set{\tB,\tS,\tBS}$ and $\gamma \in \Gamma$, analogous to those in  \eqref{eq:sigma_terms}, but with residualized outcomes $y_{ij}(\gamma) - X_{ij}^\top \beta$ instead of $y_{ij}(\gamma)$ alongside terms of the form $\xi_\gamma^\theta$, which do not depend on $\beta$. Define the inner product terms 
\begin{equation}
\label{eq:inner_prod_terms}
u_\gamma^{\tB} := \frac{1}{I-1}\sum_{i=1}^I \bar{X}_{i,\bullet}^{(d)}\bar{y}_{i,\bullet}^{(d)}(\gamma), \quad u_\gamma^{\tS} = \frac{1}{J-1}\sum_{j=1}^J \bar{X}_{\bullet,j}^{(d)}\bar{y}_{\bullet,j}^{(d)}(\gamma), \quad u_\gamma^{\tBS} = \frac{1}{(I-1)(J-1)}\sum_{i=1}^I\sum_{j=1}^J X_{ij}^{(dcr)} y_{ij}^{(dcr)}(\gamma).
\end{equation}
for each $\gamma \in \Gamma$, and the gram-matrix terms
\begin{equation}
\label{eq:gram_matrix_terms}
Z^{\tB} := \frac{1}{I-1}\sum_{i=1}^I \bar{X}_{i,\bullet}^{(d)}\bar{X}_{i,\bullet}^{(d)\top}, \quad Z^{\tS} := \frac{1}{J-1}\sum_{j=1}^J \bar{X}_{\bullet,j}^{(d)}\bar{X}_{\bullet,j}^{(d)\top}, \quad
Z^{\tBS} := \frac{1}{(I-1)(J-1)}\sum_{i=1}^I\sum_{j=1}^J X_{ij}^{(dcr)}X_{ij}^{(dcr)\top},
\end{equation}

Expanding out the square, we find that
\[
\tilde{\nu}^{\theta}_{\gamma}  = \beta^\top Z^\theta \beta - 2\beta^\top u_\gamma^{\theta} + \nu^\theta_\gamma
\]
for all $\theta \in \set{\tB,\tS,\tBS}$ and $\gamma \in \Gamma$. Summing these terms, we have that $\Var(\tau_{c}(\beta))$ is of the form $\beta^\top \tilde{Z}_c\beta - 2\tilde{u}_c^\top \beta + C$, where 
\begin{align}
\label{eq:def_quadratic_form_components}
 \tilde{Z}_c & := \sum_{\gamma \in \Gamma} a_{\gamma, \tB} Z^{\tB} + a_{\gamma, \tS} Z^{\tS} + a_{\gamma, \tBS} Z^{\tBS}\\
\label{eq:def_inner_product_beta_components}
 \tilde{u}_c & := \sum_{\gamma \in \Gamma} a_{\gamma, \tB} u_\gamma^{\tB} + a_{\gamma, \tS} u_\gamma^{\tS} + a_{\gamma, \tBS} u_\gamma^{\tBS}\\
 C & = \Var(\tau_{c})
\end{align}
The variance formula in \cref{prop:generic_variance_formula} implies that the coefficients satisfy $a_{\gamma, \tB} = O(1/I), a_{\gamma, \tS} = O(1/J), a_{\gamma, \tBS} = O(1/IJ)$. Finally, the quadratic form is positive semi-definite because its value is nonnegative. The variance-minimizing $\beta$ is given by solving the first order condition. When $\tilde{Z}_c$ is positive definite, $\tilde{\beta}_c = \tilde{Z}_c^{-1} \tilde{u}_c$.

\begin{remark}[Well-posedness of the optimal adjustment]
Null directions of $\tilde{Z}_c$ can arise when covariates carry no variation along certain components of the variance decomposition. For instance, in the direct effect case (\cref{ex:optimal_adjustment_direct_effect}), the optimization depends only on $Z^{\tBS}$, so any covariate that is constant across all buyers or all sellers lies in the null space after double-decentering. In practice, such covariates should be excluded, or one may use the Moore--Penrose pseudoinverse. The asymptotic invertibility required by \cref{assmp:non_vanishing_buyer_seller_variation} ensures that $\tilde{Z}_c$ is eventually positive definite along the relevant directions.
\end{remark}

To estimate the optimal adjustment, we will use
\begin{align}
\label{eq:optimal_beta_estimation}
 \hat{Z}_c & := \sum_{\gamma \in \Gamma} a_{\gamma, \tB} \hat{Z}_\gamma^{\tB} + a_{\gamma, \tS} \hat{Z}_\gamma^{\tS} + a_{\gamma, \tBS} \hat{Z}_\gamma^{\tBS}\\
 \hat{u}_c & := \sum_{\gamma \in \Gamma} a_{\gamma, \tB} \hat{u}_\gamma^{\tB} + a_{\gamma, \tS} \hat{u}_\gamma^{\tS} + a_{\gamma, \tBS} \hat{u}_\gamma^{\tBS},
\end{align}
with the estimators
\begin{equation}
\label{eq:inner_prod_term_estimators}
\begin{split}
& \hat{u}_\gamma^{\tB} := \frac{1}{I_\gamma}\sum_{i \in \cl{I}_\gamma} \eavg{X}_{i,\bullet}^{(d)}(\gamma)\eavg{y}_{i,\bullet}^{(d)}(\gamma), \quad \hat{u}_\gamma^{\tS} = \frac{1}{J_\gamma}\sum_{j \in \cl{J}_\gamma} \eavg{X}_{\bullet,j}^{(d)}(\gamma)\eavg{y}_{\bullet,j}^{(d)}(\gamma)\\
& \hat{u}_\gamma^{\tBS} = \frac{1}{I_\gamma J_\gamma}\sum_{(i,j) \in \gamma}\hat{X}_{ij}^{(dcr)}(\gamma) \hat{y}_{ij}^{(dcr)}(\gamma).
\end{split}
\end{equation}

\begin{equation}
\label{eq:gram_matrix_terms_estimators}
\begin{split}
& \hat{Z}_\gamma^{\tB} := \frac{1}{I_\gamma}\sum_{i \in \cl{I}_\gamma} \eavg{X}_{i,\bullet}^{(d)}(\gamma)\eavg{X}_{i,\bullet}^{(d)\top}(\gamma), \quad \hat{Z}_\gamma^{\tS} := \frac{1}{J_\gamma}\sum_{j \in \cl{J}_\gamma} \eavg{X}_{\bullet,j}^{(d)}(\gamma)\eavg{X}_{\bullet,j}^{(d)\top}(\gamma) \\
& \hat{Z}_\gamma^{\tBS} := \frac{1}{I_\gamma J_\gamma}\sum_{(i,j) \in \gamma} \hat{X}_{ij}^{(dcr)}(\gamma)\hat{X}_{ij}^{(dcr)\top}(\gamma).
\end{split}
\end{equation}
Crucially for each of these estimators, the decentering and double decentering is done within each group $\gamma$. For example,  $\eavg{X}_{i,\bullet}^{(d)}(\gamma) := \frac{1}{J_\gamma} \sum_{j\in \cl{J}_\gamma} X_{ij} - \frac{1}{I_\gamma J_\gamma} \sum_{(i,j) \in \gamma} X_{ij}$. Although the covariates for each unit are known, we use the empirical averages to ensure $\beta$ has a projection property detailed in \cref{lemma:beta_decomposition}.

\begin{example}[Direct Effect]
\label{ex:optimal_adjustment_direct_effect}

We first derive optimal regression adjustment strategies for the direct effect. As derived in \cref{sec:optimal_direct_effect_derivation}, the variance of $\hat \tau_{\dir}(\beta)$ is given by:
\begin{equation}
\label{eq:direct_effect_variance_formula}
\begin{split}
& \frac{1}{I_T J_T }\nu^{\tBS}_\tr + \frac{1}{I_T J_C }\nu^{\tBS}_\ib + \frac{1}{I_C J_T }\nu^{\tBS}_\is + \frac{1}{I_C J_C }\nu^{\tBS}_\cc \\   
& + \frac{I_C}{I_{T}I}\xi^{\tB}_{\tr,\ib}
  - \frac{1}{I}\xi^{\tB}_{\tr,\is} 
  + \frac{1}{I}\xi^{\tB}_{\tr,\cc}
 + \frac{1}{I} \xi^{\tB}_{\ib,\is}
 - \frac{1}{I} \xi^{\tB}_{\ib,\cc}
 + \frac{I_T}{I_C}\frac{1}{I}\xi^{\tB}_{\is,\cc} \\
 &  - \frac{1}{J}   \xi^{\tS}_{\tr,\ib}  
+ \frac{J_C}{J_T}\frac{1}{J} \xi^{\tS}_{\tr,\is}
+ \frac{1}{J} \xi^{\tS}_{\tr,\cc} 
+ \frac{1}{J} \xi^{\tS}_{\ib,\is} 
+ \frac{J_T}{J_C}\frac{1}{J} \xi^{\tS}_{\ib,\cc} 
- \frac{1}{J} \xi^{\tS}_{\is,\cc} \\
&  - \frac{I_C}{I_T}\frac{1}{IJ}  \xi^{\tBS}_{\tr,\ib} 
  - \frac{J_C}{J_T}\frac{1}{IJ}  \xi^{\tBS}_{\tr,\is} 
   - \frac{1}{IJ}  \xi^{\tBS}_{\tr,\cc}
  - \frac{1}{IJ}  \xi^{\tBS}_{\ib,\is}
 - \frac{J_T}{J_C}\frac{1}{IJ}  \xi^{\tBS}_{\ib,\cc}
  - \frac{I_T}{I_C}\frac{1}{IJ} \xi^{\tBS}_{\is,\cc}.
\end{split}
\end{equation}

The potential outcome contrasts do not depend on $\beta$, so the objective reduces to minimizing 
\begin{equation}
\label{eq:population_objective_direct_effect_simplified}
\frac{1}{I_T J_T }\nu^{\tBS}_\tr + \frac{1}{I_T J_C }\nu^{\tBS}_\ib + \frac{1}{I_C J_T }\nu^{\tBS}_\is + \frac{1}{I_C J_C }\nu^{\tBS}_\cc.
\end{equation}

Solving for the optimal coefficient $\beta$ using the previous section, we see that $
\tilde{\beta}_{\dir} := \tilde{Z}_{\dir}^{-1} \tilde{u}_{\dir}$ with
\begin{align*}
    \tilde{Z}_{\dir} := \sum_{\gamma} \frac{1}{I_\gamma J_\gamma}Z^{\tBS}, \quad \tilde{u}_{\dir} = \sum_{\gamma \in \Gamma}\frac{1}{I_\gamma J_\gamma} u_\gamma^{\tBS}.
\end{align*}
Thus, we will use the estimator $\hat{\beta}_{\dir} := \hat{Z}_{\dir}^{-1} \hat{u}_{\dir}$ with
\begin{align*}
    \hat Z_{\dir} := \sum_{\gamma \in \Gamma} \frac{1}{I_\gamma J_\gamma}\hat{Z}_\gamma^{\tBS}, \quad \hat u_{\dir} = \sum_{\gamma \in \Gamma}\frac{1}{I_\gamma J_\gamma} \hat{u}_\gamma^{\tBS}.
\end{align*}
There is an interpretation for the regression adjustment in this setting. The resulting plug-in $\hat{\tau}_{\dir}(\hat{\beta}_{\dir})$ turns out to be  the coefficient $\tau$ in the weighted least squares regression
\begin{align*}
y_{ij} & \stackrel{w_{ij}}{\sim} (\mu + \mu^{\tB}_i + \mu^{\tS}_j) + (\alpha + \alpha^{\tB}_i + \alpha^{\tS}_j)W_i^{\tB}  + (\delta + \delta^{\tB}_i + \delta^{\tS}_j)W_j^{\tS} + (\tau + \tau^{\tB}_i + \tau^{\tS}_j)W_i^{\tB}W_j^{\tS} + X_{ij}^\top \beta,
\end{align*}
with weights $w_{ij} = \frac{1}{I_\gamma^2 J_\gamma^2}$ for $\gamma \ni (i,j).$ To make the coefficient $\tau$ identifiable, we enforce several constraints among the variables. The appendix gathers the relevant details. 

The result is interesting for a few reasons. Firstly, the form of the regression is evocative of cross-sectional ToM, as the groups $\gamma$ of smaller size are upweighted in this regression, by a similar factor as in ToM. Moreover, because regression corresponds exactly to an imputation form estimator, our results in \cref{sec:asymptotic_theory} show a \textit{no-harm principle} -- the estimate is consistent and asymptotically normal around the true direct effect without assuming the linear covariate model. Finally, the weighted TWFE regression asymptotically outperforms no adjustment and the ANCOVA adjustment \eqref{eq:dir_ols}, since they are all contained in this class of imputation estimators. Taking $\beta = 0$ and the least squares $\beta$ from the ANCOVA objective in \eqref{eq:imputation_estimators} yields exactly the latter two procedures. Since the TWFE regression corresponds to the optimal choice of $\beta$, we should always expect an improvement in variance reduction for each of these methods.
\end{example}

\begin{example}[Total Effect]

Conducting the same exercise for the total effect, we see that the variance formula derived in \cref{sec:total_effect_variance_formulas} simplifies slightly, yielding
\begin{align}
\label{eq:total_effect_variance_simplified}
\begin{split}
\Var(\esttr - \estcc) & = \frac{1}{I_{\tr}}\nu_{\tr}^{\tB} + \frac{1}{J_{\tr}}\nu_{\tr}^{\tS} + \frac{1}{IJ}\left( \frac{I_{\cc}J_{\cc}}{I_{\tr}J_{\tr}} - 1\right)\nu_{\tr}^{\tBS}\\
& + \frac{1}{I_{\cc}}\nu_{\cc}^{\tB} + \frac{1}{J_{\cc}}\nu_{\cc}^{\tS} + \frac{1}{IJ}\left(\frac{I_{\tr}J_{\tr}}{I_{\cc}J_{\cc}} - 1\right)\nu_{\cc}^{\tBS},\\
\end{split}
\end{align}
where we again replace the potential outcomes by $y_{ij}(\gamma) - X_{ij}^\top \beta$. The resulting expression is a quadratic form in $\beta$, which is positive definite since it is the variance of a random variable. Solving the first order condition gives the optimal $\tilde\beta_{\tot} = \tilde{Z}_{\tot}^{-1} \tilde{u}_{\tot}$ with
\begin{align}
\begin{split}
    \tilde{Z}_{\tot} & := \left(\frac{1}{I_T} +  \frac{1}{I_C} \right) Z^\tB + \left(\frac{1}{J_T} +  \frac{1}{J_C} \right) Z^{\tS} + \frac{1}{IJ}\left( \frac{I_{C}J_{C}}{I_{T}J_{T}} + \frac{I_{T}J_{T}}{I_{C}J_{C}}- 2\right) Z^{\tBS}
\end{split}
\end{align}
and
\begin{align}
\begin{split}
    \tilde{u}_{\tot} & := \frac{1}{I_T} u^{\tB}_\tr + \frac{1}{J_T} u^{\tS}_\tr + \frac{1}{IJ}\left( \frac{I_{C}J_{C}}{I_{T}J_{T}} - 1\right) u^{\tBS}_\tr \\
    & + \frac{1}{I_{C}} u^{\tB}_\cc + \frac{1}{J_{C}} u^{\tS}_\cc + \frac{1}{IJ}\left(\frac{I_{T}J_{T}}{I_{C}J_{C}} - 1\right) u^{\tBS}_\cc
\end{split}
\end{align}
As an estimator for $\beta$, we will use

\begin{align}
\begin{split}
    \hat{Z}_{\tot} & := \frac{1}{I_T} \hat{Z}^{\tB}_\tr + \frac{1}{J_T} \hat{Z}^{\tS}_\tr + \frac{1}{IJ}\left( \frac{I_{C}J_{C}}{I_{T}J_{T}} - 1\right) \hat{Z}^{\tBS}_\tr \\
    & + \frac{1}{I_{C}} \hat{Z}^{\tB}_\cc + \frac{1}{J_{C}} \hat{Z}^{\tS}_\cc + \frac{1}{IJ}\left(\frac{I_{T}J_{T}}{I_{C}J_{C}} - 1\right) \hat{Z}^{\tBS}_\cc \\
    \hat{u}_{\tot} & := \frac{1}{I_T} \hat u^{\tB}_\tr + \frac{1}{J_T}\hat  u^{\tS}_\tr + \frac{1}{IJ}\left( \frac{I_{C}J_{C}}{I_{T}J_{T}} - 1\right) \hat u^{\tBS}_\tr \\
    & + \frac{1}{I_{C}}\hat u^{\tB}_\cc + \frac{1}{J_{C}} \hat  u^{\tS}_\cc + \frac{1}{IJ}\left(\frac{I_{T}J_{T}}{I_{C}J_{C}} - 1\right) \hat  u^{\tBS}_\cc.
\end{split}
\end{align}

In contrast to the direct effect, it appears difficult to write this out as a regression, unless one also includes the row and column means of the matrix as data. This is slightly evocative of Mundlak regression from the panel data literature, where the regression can be decomposed into within-unit and across unit-terms. The optimal MRD adjustments solve a mean-square error objective that can be decomposed into mean-square error terms for the row means, column means, and the doubly decentered matrix. The weights placed on each of these terms make it difficult to write the adjustment in a clean Mundlak form.

\end{example}

\begin{example}[Spillover Effects]
Computations in \cref{sec:variance_formulas} show that for the buyer spillover, it suffices to choose $\beta$ to optimize
\begin{align*}
    \frac{1}{I_T}\nu^{\tB}_\ib  +  \frac{1}{I_C}\nu^{\tB}_\cc + \frac{J_T}{I_TJ_C J}\nu^{\tBS}_\ib + \frac{J_T}{I_CJ_C J}\nu^{\tBS}_\cc,
\end{align*}
so that 
\begin{align}
\begin{split}
    \tilde{Z}_{\text{bs}} & := \left(\frac{1}{I_T} +  \frac{1}{I_C} \right) Z^\tB + \frac{J_T}{J_CJ}\left(\frac{1}{I_T} + \frac{1}{I_C} \right) Z^{\tBS} \\
    \tilde{u}_{\text{bs}} & := \frac{1}{I_T} u^{\tB}_\ib + \frac{1}{I_{C}} u^{\tB}_\cc + \frac{J_T}{I_TJ_C J}u^{\tBS}_\ib + \frac{J_T}{I_CJ_C J}u^{\tBS}_\cc.
\end{split}
\end{align}

For the seller spillover effect,
\begin{align*}
    \frac{1}{J_T}\nu^{\tS}_\is  +  \frac{1}{J_C}\nu^{\tS}_\cc + \frac{I_T}{J_TI_C I}\nu^{\tBS}_\is + \frac{I_T}{J_CI_C I}\nu^{\tBS}_\cc.
\end{align*}

The optimal adjustments are given by
\begin{align}
\begin{split}
    \tilde{Z}_{\text{ss}} & := \left(\frac{1}{J_T} +  \frac{1}{J_C} \right) Z^\tS + \frac{I_T}{I_CI}\left(\frac{1}{J_T} + \frac{1}{J_C} \right) Z^{\tBS} \\
    \tilde{u}_{\text{ss}} & := \frac{1}{J_T} u^{\tS}_\is + \frac{1}{J_{C}} u^{\tS}_\cc + \frac{I_T}{J_TI_C I}u^{\tBS}_\is + \frac{I_T}{J_CI_C I}u^{\tBS}_\cc.
\end{split}
\end{align}

Regression-adjusted estimators may be derived in a similar fashion to the total and direct effects; these adjustments do not seem to have particularly nice interpretations as weighted least squares. 
\end{example}

\section{Asymptotic Theory}
\label{sec:asymptotic_theory}

To state our asymptotic guarantees on the regression adjustment procedure, we apply an improved CLT for MRDs discussed in \cref{sec:improvements} to obtain a central limit theorem for the oracle $\hat\tau_c(\tilde{\beta}_c)$. We then show that the plug-in $\hat\tau_c(\hat{\beta}_c)$ has the same asymptotic distribution as the oracle estimator and derive inference procedures. Throughout this section, fix a coefficient vector $c := (c_\tr,c_\ib,c_\is,c_\cc)$ and the associated adjustments $\tilde{\beta}_c,\hat \beta_c$ from \cref{sec:methodology}. All proofs can be found in \cref{sec:proofs_for_asymptotic_theory}. The main assumptions are stated below. Let $e_{ij}(\gamma) := y_{ij} - X_{ij}^\top \tilde{\beta}_c.$ 

\begin{as}[Finite Population Limit]
\label{assmp:finite_population_limits}
Assume the quantities $\nu^\theta_\gamma, u_\gamma^{\theta}, Z^\theta$ have limits for $\theta \in \set{\tB,\tS,\tBS}$ and $\gamma \in \set{\tr,\ib,\is,\cc}$. 
\end{as}

\begin{as}[Small Residuals]
\label{assmp:main_assumption}
We suppose that 
        \begin{equation}
            \begin{split}
                \max\bigg(& \frac{1}{I} \max_{i=1}^I |\bar{e}^{(d)}_{i,\bullet}(\gamma)|, \frac{1}{J} \max_{j=1}^J |\bar{e}_{\bullet,j}^{(d)}(\gamma)|, \\
                & \frac{\sqrt{\log I}}{I^{3/2}} \max_{ij} |e_{ij}(\gamma) - \bar{e}_{i,\bullet}(\gamma)|, \frac{\sqrt{\log J}}{J^{3/2}} \max_{ij} |e_{ij}(\gamma) - \bar{e}_{\bullet,j}(\gamma)|,\\
                & \frac{1}{I^2} \norm{e_{ij}(\gamma) - \dbar{e}_\gamma}_{\text{op}}  \bigg)
                 = o(\Var(\widehat{\tau}_c(\tilde{\beta}_c))^{1/2}).
            \end{split}
        \end{equation}
\end{as}

\begin{as}[Covariate Control]
 \label{assmp:covariate_assumption}
For each coordinate $k = 1,\dots,d$,
        \begin{equation}
            \label{eq:covariate_assumption}
            \frac{1}{I} \max_{i,j} | X_{ij}^{(k)} - \dbar{X}^{(k)} |^2 = o(1), \quad \max_i |X_{i,\bullet}^{(k)} - \dbar{X}^{(k)}| = O(1), \quad \max_j |X_{j,\bullet}^{(k)} - \dbar{X}^{(k)}| = O(1)
        \end{equation}
Finally, $\frac{1}{I}\norm{X_{ij}^{(k)(dcr)}}_{op} = O(1)$.
\end{as}

\begin{as}[Variance Lower Bound]
\label{assmp:non_vanishing_buyer_seller_variation} 
We will assume the variance lower bound
    \begin{equation}
        \label{eq:variance_scaling_assmp}
        \Var(\hat{\tau}_c(\tilde{\beta}_c)) = \omega(I^{-2})
    \end{equation} 
and $I \sum_\gamma a_{\gamma, \theta} Z^\theta$ for $\theta \in \set{\tB,\tS}$ converges to an invertible matrix.
\end{as}

Papers in the finite population literature \cite{lin2013agnostic, li2017general, lei2021regression,lu2022tyranny} often require similar conditions as \cref{assmp:finite_population_limits,assmp:main_assumption,assmp:covariate_assumption}. \cref{assmp:non_vanishing_buyer_seller_variation} does not have direct precedent in the literature on cross-sectional adjustments. The new challenge for multiple randomization designs, as mentioned previously, is the possibility of multiple scaling regimes for the variance as well as $\tilde{Z}_c, \tilde{u}_c$.  For example, it is possible that $Z^\tB, Z^{\tS} \rightarrow 0$ while $u^\tB_\gamma, u^{\tS}_\gamma$ do not, so that $\tilde{\beta}_c$ blows up. \cref{assmp:non_vanishing_buyer_seller_variation} is a convenient starting point for our analysis. Equation \eqref{eq:variance_scaling_assmp} covers a broader range of examples than the theory of \cite{masoero2024multiple} allows, and the second condition is sufficient for $\tilde{\beta}_c = O(1)$.  We show later in this section how to go beyond this assumption in an important example.

\begin{theorem}
\label{thm:clt_population_regression}
Suppose \cref{assmp:finite_population_limits}, \cref{assmp:main_assumption} and $\sum_{\gamma \in \Gamma} c_\gamma = 0$. Then 
\[
\frac{\hat{\tau}_c(\tilde{\beta}_{c}) - \tau_c}{\sqrt{\Var(\hat{\tau}_c(\tilde{\beta}_{c}))}} \Rightarrow \textsf{N}(0,1).
\]
\end{theorem}

This follows as an immediate consequence of \cref{cor:clt_mrd_estimators} and the balancing condition on $c$. As discussed in the introduction,
\cref{thm:clt_population_regression} shows a form of model robustness: we do not need to assume that $y_{ij}(\gamma)$ follows a linear model in the covariates. These results are true of Lin's estimator and the ToM regression \cite{lin2013agnostic, lu2022tyranny}. \cite{cohen2024no} develops another approach to these so-called no-harm principles using calibration for nonlinear models.

\begin{restatable}{proposition}{consistentbeta}
\label{prop:beta_consistency}
Under \cref{assmp:finite_population_limits,assmp:main_assumption,assmp:covariate_assumption,assmp:non_vanishing_buyer_seller_variation}, the estimator $\hat{\beta}_c$ in \cref{sec:methodology} is consistent for the optimal coefficient $\tilde{\beta}_c$: $\hat \beta_c - \tilde{\beta}_c = O_p(1/\sqrt{I}).$
\end{restatable}

\begin{proposition}
\label{prop:same_asymptotic_distribution}
Suppose \cref{assmp:finite_population_limits,assmp:main_assumption,assmp:covariate_assumption,assmp:non_vanishing_buyer_seller_variation} and $\sum_{\gamma \in \Gamma} c_\gamma = 0$. Then the oracle adjusted estimator $\hat{\tau}_{c}(\tilde{\beta}_{c})$ and the regression adjusted estimator $\hat{\tau}_{c}(\hat{\beta}_{c})$ have the same asymptotic distribution. 
\end{proposition}

\subsection{Conservative Variance Estimation and Inference}

In order to conduct inference on an effect $\tau_c$ of interest, it suffices to consider conservative variance estimators for $\Var(\hat{\tau}_{c}(\tilde{\beta}_{c}))$. The most straightforward idea is to take the residuals $\hat{e}_{ij} := y_{ij} - X_{ij}^\top \hat{\beta}_c$ and plug them into the conservative variance estimator of \cite{masoero2024multiple}. Section 4.4 of \cite{masoero2024multiple} supplies said conservative variance estimator $\hat{V}_{c}$. To review the construction, we define
\begin{equation}
    \hat{\nu}_\gamma^\tB := \frac{1}{I_\gamma} \sum_{i \in \cl{I}_\gamma} \left[ \hat{\bar{y}}_{i,\bullet}(\gamma) - \hat{\dbar{y}}_\gamma \right]^2, \quad     \hat{\nu}_\gamma^\tS := \frac{1}{J_\gamma} \sum_{j \in \cl{J}_\gamma} \left[ \hat{\bar{y}}_{\bullet,j}(\gamma) - \hat{\dbar{y}}_\gamma \right]^2
\end{equation}
and the estimator counterpart for $\nu_\gamma^\tBS$ by
\begin{equation}
\hat{\nu}_\gamma^\tBS := \frac{1}{I_\gamma J_\gamma} \sum_{(i,j) \in \gamma} \left[y_{ij}(\gamma) - \hat{\bar{y}}_{i,\bullet}(\gamma) - \hat{\bar{y}}_{\bullet,j}(\gamma) + \hat{\dbar{y}}_\gamma \right]^2.
\end{equation}

Let
\begin{equation}\label{eq:alpha_gamma}
\alpha^\tB_\gamma = \frac{1}{2}\frac{I - I_\gamma}{I I_\gamma},\quad  \alpha^\tS_\gamma = \frac{1}{2}\frac{J - J_\gamma}{J J_\gamma}.
\end{equation}
Then Theorem 4.4 of \cite{masoero2024multiple} demonstrates that an unbiased estimator for $\Var(\hat{\dbar{y}}_\gamma)$ is given by
\begin{equation}
\label{eq:conservative_var_est}
\begin{split}
\hat{\Sigma}_\gamma & := \frac{\alpha^\tB_\gamma \hat{\nu}^\tB_\gamma + \alpha^\tS_\gamma \hat{\nu}^\tS_\gamma  + \alpha^\tB_\gamma \alpha^\tS_\gamma \hat{\nu}^\tBS_\gamma }{1 - \alpha^\tB_\gamma - \alpha^\tS_\gamma + \alpha^\tB_\gamma\alpha^\tS_\gamma}    \\
& - \frac{\alpha^\tB_\gamma}{1- \alpha^\tB_\gamma} \frac{J - J_\gamma}{J(J_\gamma - 1)} \frac{1}{I_\gamma J_\gamma} \sum_{(i,j) \in \gamma} \left( y_{ij}(\gamma) - \hat{\bar{y}}_{i,\bullet}(\gamma) \right)^2 \\
& - \frac{\alpha^\tS_\gamma}{1- \alpha^\tS_\gamma} \frac{I - I_\gamma}{I(I_\gamma - 1)} \frac{1}{I_\gamma J_\gamma} \sum_{(i,j) \in \gamma} \left( y_{ij}(\gamma) - \hat{\bar{y}}_{\bullet,j}(\gamma) \right)^2.    
\end{split}
\end{equation}

As a result, one can estimate the conservative upper bound
\begin{equation}
\label{eq:conservative_upper_bound_on_variance}
V_c := \left(\sum_{\gamma \in \Gamma} |c_\gamma| \sqrt{\Var(\edavg{y}_\gamma)} \right)^2    
\end{equation}
using the conservative variance estimator
\begin{equation}
\label{eq:conservative_variance_estimator}
\hat{V}_c = \left(\sum_{\gamma \in \Gamma} |c_\gamma| \sqrt{\hat{\Sigma}_\gamma^+} \right)^2.
\end{equation}
where $x^+ = \max(x,0)$. \\  

Plugging in the residuals $y_{ij} - X_{ij}^\top \hat{\beta}_c$ into \eqref{eq:conservative_variance_estimator} yields the plug-in estimate $\hat{V}_c(\hat{\beta}_c)$. Using this we can create an asymptotically conservative confidence interval given by
\begin{equation}
\label{eq:plug_in_conservative_confidence_interval}
\hat{\tau}_c(\hat{\beta}_c)  \pm z_{1-\alpha/2} \hat{V}_c(\hat{\beta}_c)^{1/2}.
\end{equation}


The next theorem gives assumptions under which the confidence interval is valid, leveraging an improvement in conservative variance estimation for MRDs detailed in the next subsection. The main assumption is that the group average variances are of comparable or larger order than $\Var(\hat \tau_c (\tilde{\beta}_c))$ over the  localized ball $B_I(M) := \set{b : \norm{b - \tilde{\beta}_c}_2 \leq M I^{-1/2}}$. Define $e_{ij}(b,\gamma) = y_{ij}(\gamma) - X_{ij}^\top b.$

\begin{as}[Per Group Variance Bound]
\label{assmp:per_gamma_variance_lb} 
For each $M > 0$, 
\begin{equation}
\inf_{b \in B_I(M)} \Var(\edavg{e}(b,\gamma) ) = \Omega(\Var(\hat \tau_c (\tilde{\beta}_c))).
\end{equation}
\end{as}

\begin{theorem}[Valid Conservative Inference]
\label{thm:valid_inference}
Under \cref{assmp:finite_population_limits,assmp:main_assumption,assmp:covariate_assumption,assmp:non_vanishing_buyer_seller_variation,assmp:per_gamma_variance_lb}, the conservative confidence interval in \eqref{eq:plug_in_conservative_confidence_interval} has asymptotic coverage at least $1-\alpha$. 
\end{theorem}

It is possible to tighten these results under certain regimes. If we posit the variance lower bound $\Var(\hat{\tau}_c(\tilde{\beta}_c)) = \Omega(I^{-1})$ of \cite{masoero2024multiple}, we can dispense with the last three bounds of \cref{assmp:covariate_assumption} and \cref{assmp:per_gamma_variance_lb}. In some cases, we can also go beyond \cref{assmp:non_vanishing_buyer_seller_variation}, an important example being inference for the direct effect. 

\begin{example}[Inference for the Direct Effect]
\label{ex:inference_for_direct_effect}
If one inspects the case of the direct effect in \cref{ex:optimal_adjustment_direct_effect}, the regression adjustment minimizes \eqref{eq:population_objective_direct_effect_simplified}. But notice that this term is of order $O((IJ)^{-1})$, suggesting that direct effect adjustment is meaningful asymptotically only when buyer and seller mean variation is of lower order. Our tools can theoretically show regression adjustment works in this case. To illustrate, suppose that 
\[
\nu^\tB(y_{ij}(\gamma)) = O(1/I), \quad  \nu^\tS(y_{ij}(\gamma)) = O(1/I)
\]
for each $\gamma$ and similarly for each coordinate of the covariates. To avoid degenerate cases, suppose further that $\nu^\tBS(e_{ij}(\gamma)) \not \rightarrow 0$ and $\lim Z^{\tBS}$ is invertible. In particular, this implies that the variance of the direct effect estimator is $\Theta(1/IJ)$.

Under \cref{assmp:finite_population_limits,assmp:main_assumption,assmp:covariate_assumption}, one can establish the conclusions of \cref{prop:same_asymptotic_distribution,thm:valid_inference}. The idea is to apply \cref{thm:clt_population_regression} to $\hat\tau_{\dir}(\tilde{\beta}_{\dir})$ and note that $\hat\tau_{\dir}(\hat{\beta}_{\dir}) - \hat\tau_{\dir}(\tilde{\beta}_{\dir})$ is given by
\[
\left(\edavg{X}_\tr - \edavg{X}_\ib - \edavg{X}_\is + \edavg{X}_\cc\right)^\top(\hat \beta_{\dir} - \tilde \beta_{\dir}).
\]
The covariate difference term $\edavg{X}_\tr - \edavg{X}_\ib - \edavg{X}_\is + \edavg{X}_\cc$ is small: using the variance formula for the direct effect in \cref{sec:optimal_direct_effect_derivation}, the variance of each coordinate is of order $O(1/IJ).$ Exploiting this gives the results. Details can be found in the appendix.
\end{example}


\subsection{Improved MRD Theory}
\label{sec:improvements}

Our results on the plug-in optimal adjustment utilize improvements to the theory of MRDs, strengthening the conclusions and relaxing some assumptions of \cite{masoero2024multiple}. The first improvement is a quantitative CLT for the doubly randomized sum $\dbar{y}_\gamma.$ The approach uses finite CLTs and tensorization properties in the $1$-Wasserstein distance and somewhat simplifies the argument in \cite{masoero2024multiple} by avoiding multivariate quantitative CLTs. Although the bounds are not directly comparable since the CLT of \cite{masoero2024multiple} uses Kolmogorov distance, our new result shows asymptotic normality 
without two previous assumptions: 1) boundedness and 2) the variance lower bound $\Var(\edavg{y}) = \Omega(I^{-1} + J^{-1})$. The proof and examples are given in \cref{sec:appendix_improved_theory}. 

\begin{restatable}{theorem}{thmimprovedclt}
\label{thm:improved_clt}
Consider an $I \times J$ matrix $Y$ with entries $y_{ij}$. Let $W_i^\tB, W_j^\tS$ be indicator variables for sampling $I_T$ rows and $J_T$ columns respectively. Let $\edavg{y} = \frac{1}{I_TJ_T} \sum_{i,j} W_i^\tB W_j^\tS y_{ij}$, $\dbar{y}$ be the grand mean $\frac{1}{I J} \sum_{i,j} y_{ij}$, and define $\sigma_{\textsf{Tot}}^2 := \Var(\edavg{y})$. We have the bound
\begin{align}
d_W\left( \cl{L}\left(\frac{\edavg{y} - \dbar{y}}{\sigma_{\textsf{Tot}}} \right), \textsf{N}\left(0, 1 \right)\right) & \lesssim \frac{1}{\sqrt{I}} + \frac{1}{I} \frac{\max_{i=1}^I \left| \bar{y}_{i,\bullet} - \dbar{y} \right|}{\sigma_{\textsf{Tot}}} + \frac{1}{J} \frac{\max_{j=1}^J \left| \bar{y}_{\bullet,j} - \dbar{y} \right|}{\sigma_{\textsf{Tot}}} \\
& \quad + \frac{\sqrt{\log I}}{I^{3/2}}  
 \frac{\max_{i,j} |y_{ij} - \bar{y}_{i,\bullet}|}{\sigma_{\textsf{Tot}}} + \left(\frac{1}{I^2} \frac{\norm{y_{ij} - \dbar{y}}_{\text{op}}}{\sigma_{\textsf{Tot}}} \right)^{1/2}.
\end{align} 
\end{restatable}

Roughly, the theorem is proved by conditioning on the seller variables $\cl{S}$ and applying a finite population central limit theorem in the $L^1$-Wasserstein distance \cite{goldstein20071}. The sketch of course suggests that we may condition first on the buyer variables, obtaining an analogous result to \cref{thm:improved_clt} where the right-hand-side applies to $Y^\top$. This can yield better guarantees, especially when the matrix is imbalanced.  \cref{thm:improved_clt} can be used to establish a CLT for various MRD estimators.

\begin{restatable}{corollary}{cormrdclt}
\label{cor:clt_mrd_estimators}
Consider any MRD estimator $\widehat{\tau}_c = \sum_{\gamma\in\Gamma} c_\gamma \edavg{y}_\gamma,$
and suppose that for each set of potential outcomes $(y_{ij}(\gamma))_{ij}, \gamma \in \Gamma$, we have the analogue of Eq. \eqref{eq:CLT_assumptions}:  
\begin{equation}
\label{eq:CLT_assumptions_MRD_estimands}
\begin{split}
& \max\bigg(\frac{1}{I} \max_{i=1}^I |\bar{y}_{i,\bullet}(\gamma)- \dbar{y}_{\gamma}|, \frac{1}{J} \max_{j=1}^J |\bar{y}_{\bullet,j}(\gamma) - \dbar{y}_{\gamma}|, \\
& \max_{i,j}\frac{\sqrt{\log I}}{I^{3/2}} |y_{ij}(\gamma) - \bar{y}_{\bullet,j}(\gamma)|, \frac{1}{I^2} \norm{y_{ij}(\gamma) - \dbar{y}_\gamma}_{\text{op}}  \bigg) = o(\Var(\widehat{\tau}_c)^{1/2}).
\end{split}
\end{equation}
Then
\[
\frac{\widehat{\tau}_{c} - \tau_{c}}{\Var(\widehat{\tau}_{c})^{1/2}} \Rightarrow \textsf{N}(0,1).
\]
\end{restatable}

In the appendix, we show \cref{thm:improved_clt} recovers the central limit theorem of \cite{masoero2024multiple} with bounded outcomes and a lower bound on the variance. \cref{ex:iid_normal_superpop,ex:sparse_uniform_subgraph,ex:sparse_heavy_tailed} in the appendix give several examples where this \cref{cor:clt_mrd_estimators} may apply but the CLT of \cite{masoero2024multiple} does not. The idea behind these examples is vanishing variation in the row and column means, which often appears due to \textit{sparsity}. It is often the case in real two-sided marketplaces that only a small fraction of the $I \times J$ potential interactions between buyers and sellers actually take place. \cref{thm:improved_clt} can establish asymptotic normality in these settings. \\

Our next result shows that the conservative variance estimator is consistent in Theorem 4.4 of \cite{masoero2024multiple}, again relaxing the need for a $\Theta(I^{-1})$ lower bound on the variance. 

\begin{restatable}[Consistent Variance Estimation]{theorem}{conservativevariance}
\label{thm:consistent_variance_estimator}

Suppose \eqref{eq:CLT_assumptions_MRD_estimands} with the additional condition $\frac{\sqrt{\log I}}{I^{3/2}} |y_{ij}(\gamma) - \bar{y}_{i,\bullet}(\gamma)| = o(\Var(\widehat{\tau}_c)^{1/2}).$ Then
\[
\frac{\hat{V}_{c}}{V_c}  \xrightarrow{p} 1.
\]
\end{restatable}

The assumption can be slightly relaxed; in the third assumption of \cref{cor:clt_mrd_estimators}, we only need $I^{-3/2} \max_{i,j}|y_{ij}(\gamma) - \bar{y}_{\bullet,j}(\gamma)| = o(\Var(\widehat{\tau}_c)^{1/2}).$

\section{Simulation Results}
\label{sec:simulations}

We present two simulation settings comparing three methods for various target estimands: the unadjusted MRD estimator, ANCOVA, and the plug-in optimal non-interacted estimator. For simplicity, we will refer to the latter estimator as Opt. We plot sampling distributions of each estimator along with averaged coverage and lengths of confidence intervals obtained from the plug-in procedure. The simulations randomize the assignment variables $W_i^\tB,W_j^\tS$ over 5000 Monte Carlo runs. Overall, we corroborate the \textit{no-harm} property of the optimal adjustment -- in large samples, the method is guaranteed to be at least as efficient as the other methods. In some settings, the improvement can be substantial. Throughout this section, we will roughly refer to an experiment as \textit{imbalanced} if the ratios $I_T/I, J_T/J$ are small and \textit{balanced} if $I_T/I \approx J_T/J \approx 0.5$.

We first consider a toy synthetic example, which is also the setting of \cref{fig:intro_plots} in the introduction. The example considers $Y_{ij}(\gamma) \sim \textsf{N}(\mu_\gamma,1)$ from \cref{ex:iid_normal_superpop}. We take $(\mu_\tr,\mu_\ib,\mu_\is,\mu_\cc) = (5,2,2,1)$ and three specifications for the number of treated buyers and sellers, where $\set{10\%, 20\%, 50\%}$ of the buyers and sellers are treated. In this semisynthetic simulation, we provide a four-dimensional covariate where the coordinates are noisy versions of $Y_{ij}(\gamma), \gamma \in \Gamma$ and the target is the direct effect estimand. Although unrealistic, this example showcases the setting of \cref{ex:inference_for_direct_effect} where the theory of \cite{masoero2024multiple} does not apply. The results are displayed in \cref{fig:normal_data_direct_effect_sampling_distribution,fig:normal_data_direct_effect_ci}. 

Opt outperforms both methods in all cases, delivering substantial gains in efficiency for imbalanced experiments. For a balanced experiment, Opt and ANCOVA perform very similarly. For imbalanced experiments, we also find that the ANCOVA procedure has \textit{larger variance} than the unadjusted procedure, extending the findings of \cite{freedman2008regressionA, freedman2008regressionB} where the naive regression adjustment might worsen asymptotic precision in cross-sectional experiments. All methods overcover the target, with Opt continuing to deliver the shortest confidence intervals.

\begin{figure}
    \centering
    \includegraphics[width=\linewidth]{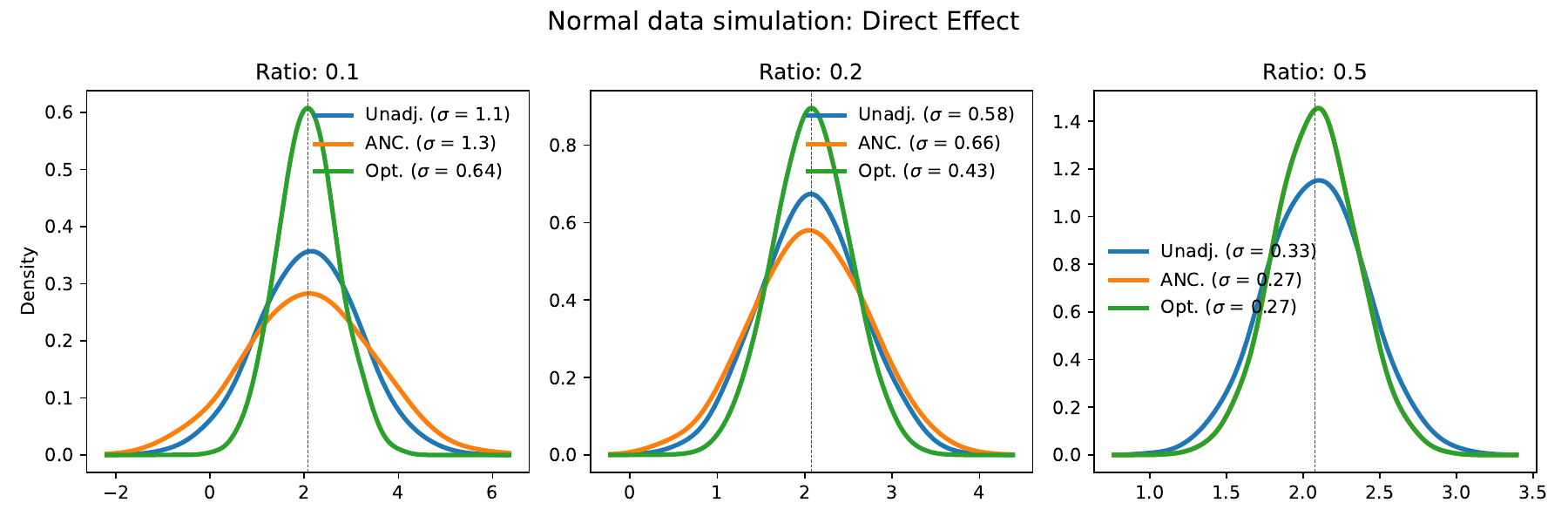}
    \caption{Sampling distributions for Unadjusted, ANCOVA, and Optimal Non-Interacted adjustments over 5000 Monte Carlo samples, on a fixed realization of the i.i.d. normal data generating process of \cref{ex:iid_normal_superpop} with parameters described in the text. The direct effect is the estimand. Monte Carlo randomizes over $W_i^\tB, W_j^\tS$. The three panels correspond to three scenarios for ratio $I_T/I = J_T/J$, $\set{0.1,0.2,0.5}.$  Legends show estimated standard deviations of each method.}
    \label{fig:normal_data_direct_effect_sampling_distribution}
\end{figure}

\begin{figure}
    \centering
    \includegraphics[width=\linewidth]{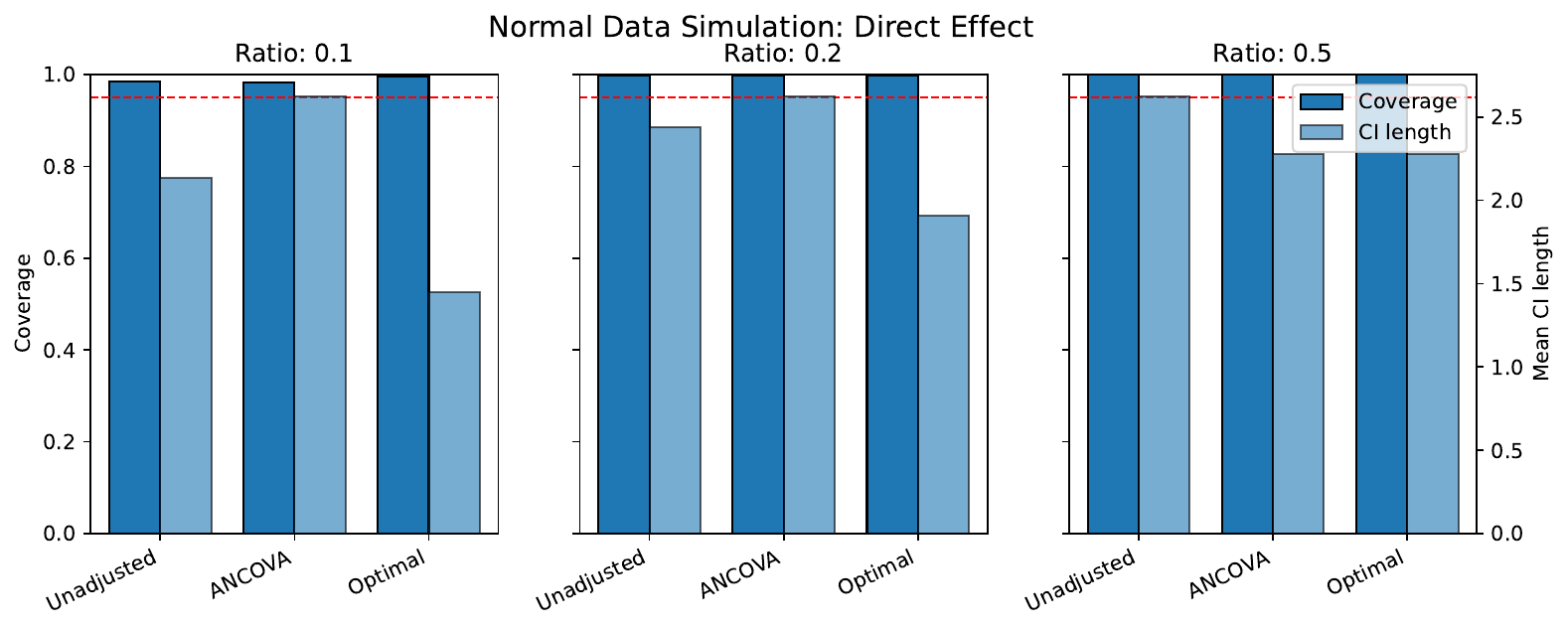}
    \caption{Coverage and CI length for Unadjusted, ANCOVA, and Optimal Non-Interacted adjustments averaged over 5000 Monte Carlo samples, on a fixed realization of the i.i.d. normal data generating process of \cref{ex:iid_normal_superpop} with parameters described in the text. The direct effect is the estimand. Parameters are described in the text. Monte Carlo randomizes over $W_i^\tB, W_j^\tS$. The three panels correspond to three scenarios for ratio $I_T/I = J_T/J$, $\set{0.1,0.2,0.5}.$ }
    \label{fig:normal_data_direct_effect_ci}
\end{figure}

In the second set of simulations, we instantiate a more realistic synthetic model taken from Example 2.5 of \cite{masoero2024multiple}. The example models a marketplace of content creators and advertisers, indexed by $i,j$ respectively. Each creator $i$ produces content with \textit{score} $q_i^c$ and each advertiser places ads with corresponding quality $q_j^a$. Each creator-advertiser pair generates revenue $y_{ij}$ equal to $m_{ij}(q_i^c + q_j^a)$ for some notion of compatibility $m_{ij}$. For some pre-negotiated parameters $r_i^c,r_j^a$, creators and advertisers are compensated $r_i^cy_{ij}, r_j^ay_{ij}$ respectively with the rest going to the platform. A binary intervention $\mathbf{w}$ is tested which allocates a subsidy $\eta$, affecting revenue as follows:
\[
y_{ij}(\mathbf{w}) = (m_{ij} + \eta w_{ij}) \set{q_i^c(\mathbf{w}) + q_j^a(\mathbf{w})}
\]
A simple model of incentives for both creators and advertisers leads to a Nash equilibrium where the revenue $y_{ij}(\mathbf{w})$ satisfies the local interference condition (see Example 2.5 of \cite{masoero2024multiple}). We consider nearly the same parameters as in Section 5 of that work, with $\eta = 5$, i.i.d. parameters $m_{ij} \sim \text{Exp}(1)$, and $r_i^c, r_j^a \sim \text{Unif}[0,0.2]$ across creators and advertisers
$i \in [I], j \in [J]$. We take $I = 200, J = 150$ and consider the covariate vector $X_{ij} = (\hat m_{ij}r_i^c, \hat m_{ij}r_j^a) \in \bR^2.$ We assume the planner has access to a noisy observation of the compatibility $m$, given by $\hat{m}_{ij} \sim m_{ij}(1 + \textsf{N}(0,0.1))$. Three specifications are considered:  $(I_T, J_T) = \set{(20, 15), (40, 30), (100, 75)}.$ We will test our methods for the buyer spillover effect in this example. The results are displayed in \cref{fig:marketplace_ib_effect_sampling_distribution,fig:marketplace_ib_effect_ci}.

\begin{figure}[H]
    \centering
    \includegraphics[width=\linewidth]{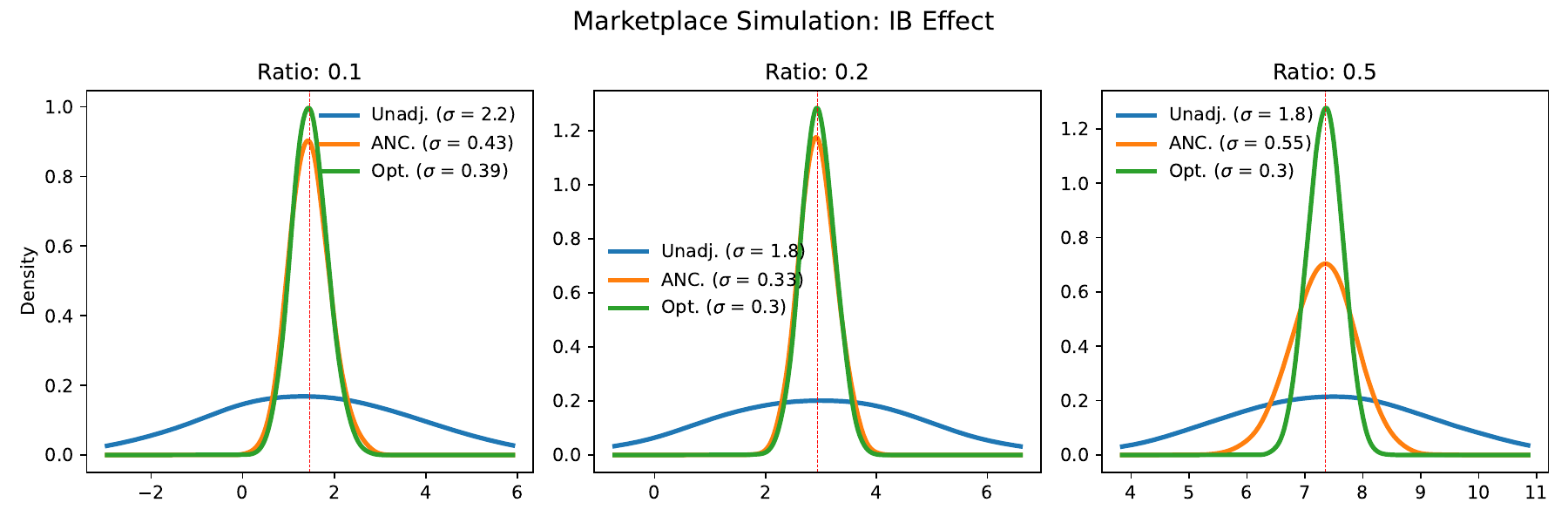}
    \caption{Sampling distributions for Unadjusted, ANCOVA, and Optimal Non-Interacted adjustments over 5000 Monte Carlo samples, on a fixed realization of the marketplace data generating process of \cite{masoero2024multiple}. The buyer spillover effect is the estimand. Parameters are described in the text. }
    \label{fig:marketplace_ib_effect_sampling_distribution}
\end{figure}

\begin{figure}[H]
    \centering
    \includegraphics[width=\linewidth]{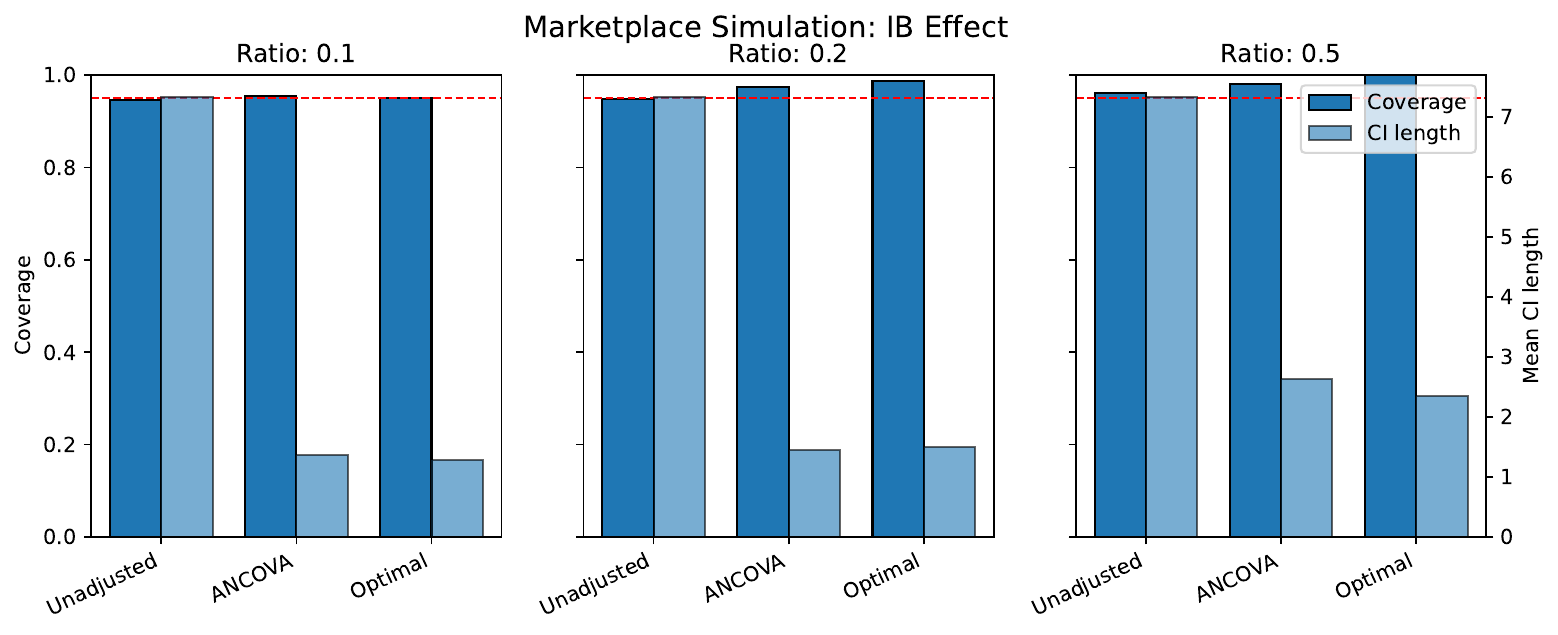}
    \caption{Coverage and CI length for Unadjusted, ANCOVA, and Optimal Non-Interacted adjustments averaged over 5000 Monte Carlo samples, on a fixed realization of the marketplace data generating process of \cite{masoero2024multiple}. The buyer spillover effect is the estimand. Parameters are described in the text. }
    \label{fig:marketplace_ib_effect_ci}
\end{figure}

Both methods for regression adjustment significantly outperform the unadjusted method in all metrics.  ANCOVA and Opt perform similarly in imbalanced experiments, with Opt being slightly more efficient. Opt significantly outperforms ANCOVA for balanced experiments in efficiency. However, the reduction in confidence interval length is not as large as the reduction in variance. It may be fruitful to consider optimal regression adjustments which optimize $\beta$ to minimize the length of the conservative confidence interval, but this is at odds with the natural objective of finding a minimum variance estimator.

Inspecting the two simulation settings, Opt significantly outperforms ANCOVA for different ratios of treatment and control, in different data generating processes. Although a theoretical analysis could be done to analyze when this occurs, the no-harm principle seems robust enough to practically recommend Opt over ANCOVA. In practice, we also recommend some attention to the choice of covariates. Notice that the direct effect estimator \eqref{eq:imputation_estimators} cannot use covariates which are constant across either buyers or sellers; similarly, the imputation estimator for the buyer and seller spillovers cannot take advantage of covariates which are the same across sellers and buyers, respectively.

\section{Discussion}
\label{sec:discussion}

In this work, we propose optimal regression adjustments for estimands in multiple randomization designs. These adjustments are model-robust and satisfy no-harm principles which guarantee improvement over other baselines such as ANCOVA or unadjusted MRD estimators. We establish central limit theorems for these optimal regression adjustments as well as methods for conservative inference. These results build on improved tools that also strengthen the results in \cite{masoero2024multiple}. Our tools can also be used to analyze optimal interacted regression adjustments; we provide formulas and numerical results for these interacted adjustments in \cref{sec:interacted_estimators}, but leave detailed theoretical analysis for future work.

\subsection{Directions for Future Research}

\paragraph{Further Questions on MRDs.}
As discussed in the introduction, a full investigation of the scaling limits of MRDs would constitute a substantial body of future research. Investigating all reasonable asymptotic limits for doubly randomized sums, like in \cref{thm:improved_clt}, would be a first step. For example, it is possible to have non-normal limits. Let $Y_{ij}$ be an $I \times I$ square matrix and equal to a rank one matrix $xx^\top$ for some vector $x \in \bR^I$. The doubly randomized sum can be factored into 
\[
\edavg{Y} = \left( \frac{1}{I_T}\sum_{i=1}^I W_i^\tB x_i \right) \left( \frac{1}{J_T}\sum_{j=1}^J W_j^\tS x_j \right).
\]
When $\bar{x} = 0,$ one can show that the last term in the bound of \cref{thm:improved_clt} fails to vanish and $\edavg{Y}$ converges weakly to the product of two independent normals. The situation bears some resemblance to the theory of second-order $U$-statistics. Characterizing all possible distributional limits and developing methods for uniform inference would be an interesting direction to explore. In a related but different setting, \cite{bhattacharya2024fluctuations} characterizes the possible distributional limits of \textit{quadratic chaos}: quantities of the form $\sum_{i,j} a_{i,j} X_i X_j$ where $a_{ij}$ is a symmetric $\set{0,1}$-valued symmetric matrix and $X_1,\dots,X_n$ are iid random variables with zero mean and variance $1$. It is shown that possible limits are independent sums of a Gaussian, a weighted sum of independent chi-squared random variables, and a Gaussian mixture with random variance. It is possible a similar decomposition holds for doubly randomized sums.

\paragraph{Parallels and Differences with the Cross-Sectional Case.} As we emphasize throughout the paper, several of our findings echo results in the cross-sectional case. However, several questions remain. \cite{lin2013agnostic} shows that the Huber-White standard errors in the cross-sectional interacted linear regression serve as conservative variance estimators suitable for inference. Is it possible to extend this result to the interacted TWFE for the direct effect? Secondly, \cite{lin2013agnostic} shows that the tyranny-of-minority and interacted regressions are both asymptotically the same. Preliminary investigation suggests that both analogues are false in the MRD setting for the direct effect, but further work would be interesting. Separately, \cite{lei2021regression} explore regression adjustments in a high-dimensional setting using debiasing techniques, which one can extend to the MRD setting.  Nonlinear regression adjustment strategies are also important here. Some outcomes in multiple randomization designs are naturally modelled as Poisson or logistic regressions. Extending the results of \cite{guo2023generalized} and \cite{cohen2024no} would be valuable.

\paragraph{Optimized Conservative Variance Bounds.} It would be interesting to find variance estimators which are less conservative than those proposed in \cref{sec:improvements}. For example, analyzing the direct effect formula in the appendix suggests alternative ways to derive conservative variance bounds, using Cauchy-Schwarz and throwing away negative covariance terms. \cite{harshaw2021optimized} provides an optimization approach to finding least-conservative variance estimators in cross-sectional experiments. Pursuing similar ideas could yield sharper inferences in multiple randomization designs.

\bibliographystyle{alpha}
\bibliography{main}

\appendix

\section{Variance and Covariance Formulas for MRD Estimators}
\label{sec:variance_formulas}

This section records the results of \cite{masoero2024multiple}, with some inconsequential differences in definitions, for example switching $1/I$ with $1/(I-1)$. As we work under an asymptotic framework, we will interchange $1/I$ and $1/(I-1)$. Firstly, we define the demeaned treatment indicators 
\begin{equation}
    D_i^\tB = W_i^\tB  - \frac{I_T}{I}, \text{ and } D_j^\tS = W_j^\tS  - \frac{J_T}{J}.
\end{equation}
Lemma A.3 of \cite{masoero2024multiple} shows that $D_i^\tB,D_j^\tS$ are mean zero and characterizes their variance and covariance. Recall the notation
\[
\dbar{y}_\gamma := \frac{1}{IJ}\sum_{i=1}^I \sum_{j=1}^J y_{ij}(\gamma)
\]
Further, define the following average-like quantities:
\begin{align*}
\overline{y}_{i,\bullet}(\gamma) & := \frac{1}{J}\sum_{j=1}^J y_{ij}(\gamma), \  \overline{y}_{\bullet,j}(\gamma) := \frac{1}{I}\sum_{i=1}^I y_{ij}(\gamma)
\\
\delta_{i}^\tB(\gamma) & := \overline{y}_{i,\bullet}(\gamma) - \bar{\bar{y}}_\gamma, \ \delta_{j}^\tS(\gamma) := \overline{y}_{\bullet,j}(\gamma) - \bar{\bar{y}}_\gamma \\
\delta_{ij}^{\tBS}(\gamma) & := y_{ij}(\gamma) - \overline{y}_{i,\bullet}(\gamma) - \overline{y}_{\bullet,j}(\gamma) + \bar{\bar{y}}_\gamma
\end{align*}
and the following variance-like quantities: 
\begin{align}
\nu_{\gamma}^{\tB} & := \frac{\sum_{i=1}^I[\delta_i^{\tB}(\gamma)]^2}{I-1}, \ \nu_{\gamma}^{\tS} := \frac{\sum_{j=1}^J[\delta_j^{\tS}(\gamma)]^2}{J-1} \\
\nu_{\gamma}^{\tBS} & := \frac{\sum_{i=1}^I\sum_{j=1}^J[\delta_{ij}^{\tBS} (\gamma)]^2}{(I-1)(J-1)} \\
\xi_{\gamma,\gamma'}^{\tB} & := \frac{\sum_{i=1}^I[\delta_i^{\tB}(\gamma) - \delta_i^{\tB}(\gamma')]^2}{I-1}, \ \xi_{\gamma,\gamma'}^{\tS} := \frac{\sum_{j=1}^J[\delta_j^{\tS}(\gamma) - \delta_j^{\tS}(\gamma')]^2}{J-1} \\ 
\xi_{\gamma,\gamma'}^{\tBS} & := \frac{\sum_{i=1}^I\sum_{j=1}^J[\delta_{ij}^{\tBS} (\gamma) - \delta_{ij}^{\tBS} (\gamma')]^2}{(I-1)(J-1)}.
\end{align}
Notice that the $\xi$ quantities are all zero when $\gamma = \gamma'.$ Next, we define the average residuals for each $\gamma \in \set{\cc,\ib,\is,\tr}:$
\begin{align}
    \overline{\e}_\gamma^{\tB} & = \frac{1}{I_\gamma}\sum_{i=1}^I D_i^{\tB} \delta_i^{\tB}(\gamma) \\
    \overline{\e}_\gamma^{\tS} & = \frac{1}{J_\gamma}\sum_{j=1}^J D_j^{\tS} \delta_j^{\tS}(\gamma) \\
    \dbar{\e}_\gamma^{\tBS} & = \frac{1}{I_\gamma J_\gamma}\sum_{i,j} D_i^{\tB} D_j^{\tS}\delta_{ij}^{\tBS}(\gamma).
\end{align}
Because $\delta_i^{\tB}, \delta_j^{\tS}$ sum to zero and $\delta_{ij}^{\tBS}$ has zero row and column sums, the quantities remain unchanged if we replace $D_i^\tB,D_j^\tS$ by $W_i^\tB,W_j^\tS$. The groupwise average estimators can be decomposed in terms of the average residuals, as the next lemma shows.

\begin{lemma}[Hoeffding Decomposition of Groupwise Averages, Lemma A.4 of \cite{masoero2024multiple}]
\label{lemma:hoeffding_decomposition}
\begin{align*}
\esttr & = \dbar{y}_{\tr} + \overline{\e}_{\tr}^{\tB} + \overline{\e}_{\tr}^{\tS} + \dbar{\e}_{\tr}^{\tBS} \\
\estib & = \dbar{y}_{\ib} + \overline{\e}_{\ib}^{\tB} - \overline{\e}_{\ib}^{\tS} - \dbar{\e}_{\ib}^{\tBS} \\
\estis & = \dbar{y}_{\is} - \overline{\e}_{\is}^{\tB} + \overline{\e}_{\is}^{\tS} - \dbar{\e}_{\is}^{\tBS} \\
\estcc & = \dbar{y}_{\cc} - \overline{\e}_{\cc}^{\tB} - \overline{\e}_{\cc}^{\tS} + \dbar{\e}_{\cc}^{\tBS} 
\end{align*}
Moreover, the residuals are mean zero for each $\gamma$ and uncorrelated between different groups:
\begin{align*}
    \E[\overline{\e}_{\gamma}^{\tB}] & = \E[\overline{\e}_{\gamma}^{\tS}] = \E\left[\dbar{\e}_{\gamma}^{\tBS} \right] = 0,\\
    \Cov(\overline{\e}_{\gamma}^{\tB},\overline{\e}_{\gamma'}^{\tS}) & = \Cov(\overline{\e}_{\gamma}^{\tB},\dbar{\e}_{\gamma'}^{\tBS}) = \Cov(\overline{\e}_{\gamma}^{\tS},\dbar{\e}_{\gamma'}^{\tBS})=0.
\end{align*}
\end{lemma}

\begin{proposition}[Generic Variance Formula]
\label{prop:generic_variance_formula}
\[
\text{Var}_{\gamma} := \Var(\estgen_{\gamma}) = \frac{I - I_\gamma}{I_\gamma I}\nu^{\tB}_\gamma + \frac{J - J_\gamma}{J_\gamma J}\nu^{\tS}_\gamma + \frac{I - I_\gamma}{I_\gamma I}\frac{J - J_\gamma}{J_\gamma J}\nu^{\tBS}_\gamma 
\]
\end{proposition}
\begin{proof}
By the Hoeffding decomposition above, the variance of $\estgen_\gamma$ is simply $\Var(\overline{\e}_{\gamma}^{\tB}) + \Var(\overline{\e}_{\gamma}^{\tS}) + \Var(\dbar{\e}_{\gamma}^{\tBS})$, and so we must calculate the variance of each term. By the computation in Lemma A.7 of \cite{masoero2024multiple}, we expand out the square in $\overline{\e}_{\gamma}^{\tB}$ and cancel terms using the fact that $\sum_i^I \delta_i^{\tB}(\gamma) = 0.$ Carrying out the details leads to the conclusion
\begin{align*}    
\Var_\gamma^{\tB} & = \frac{1}{I_\gamma^2} \frac{(I - I_\gamma)I_\gamma}{I^2(I-1)} \sum_{i=1}^I \delta_i^{\tB}(\gamma)^2 + \frac{1}{I_\gamma^2} \frac{(I - I_\gamma)I_\gamma}{I^2} \sum_{i=1}^I \delta_i^{\tB}(\gamma)^2 \\
& = \frac{1}{I_\gamma^2} \frac{(I - I_\gamma)I_\gamma}{I(I-1)} \sum_{i=1}^I \delta_i^{\tB}(\gamma)^2 \\
& = \frac{I - I_\gamma}{I_\gamma I} \nu_{\gamma}^{\tB}.
\end{align*}
The quantity $\Var(\overline{\e}_{\gamma}^{\tS})$ gives the analogous result by the same computation:
\[
\Var_\gamma^{\tS} = \frac{J - J_\gamma}{J_\gamma J} \nu_{\gamma}^{\tS}.
\]
Finally the doubly decentered term $\dbar{\e}_{\gamma}^{\tBS}$ is handled in a similar fashion by expanding out the square.
\begin{align*}
\Var_{\gamma}^{\tBS} & = \frac{1}{I_\gamma^2J_\gamma^2} \frac{I_\gamma J_\gamma(I - I_\gamma)(J - J_\gamma) }{I^2J^2} \frac{IJ}{(I-1)(J-1)} \sum_{i=1}^I \sum_{j=1}^J \delta_{ij}^{\tBS}(\gamma)^2 \\
& = \frac{I-I_\gamma}{I_\gamma I}\frac{J-J_\gamma}{J_\gamma J} \nu_{\gamma}^{\tBS}.
\end{align*}
\end{proof}

\begin{corollary}[Explicit Variance Formulas]
\begin{align*}
\Var_{\tr} & =  \frac{I_C}{I_T I}\nu^{\tB}_\tr + \frac{J_C}{J_T J}\nu^{\tS}_\tr + \frac{I_C}{I_T I}\frac{J_C}{J_T J}\nu^{\tBS}_\tr \\    
\Var_{\ib} & =  \frac{I_C}{I_T I}\nu^{\tB}_\ib + \frac{J_T}{J_C J}\nu^{\tS}_\ib + \frac{I_C}{I_T I}\frac{J_T}{J_C J}\nu^{\tBS}_\ib \\    
\Var_{\is} & =  \frac{I_T}{I_C I}\nu^{\tB}_\is + \frac{J_C}{J_T J}\nu^{\tS}_\is + \frac{I_T}{I_C I}\frac{J_C}{J_T J}\nu^{\tBS}_\is \\    
\Var_{\cc} & =  \frac{I_T}{I_C I}\nu^{\tB}_\cc + \frac{J_T}{J_C J}\nu^{\tS}_\cc + \frac{I_T}{I_C I}\frac{J_T}{J_C J}\nu^{\tBS}_\cc.  
\end{align*}
 
\end{corollary}

\begin{proposition}[Generic Covariance Formula]
\label{prop:generic_covariance_formula}
Let $\eta_{\gamma,\gamma'}^\tB \in \set{-1,+1}$ be $+1$ if groups $\gamma, \gamma'$ share the same rows, and $-1$ otherwise. Similarly, let $\eta_{\gamma,\gamma'}^\tS \in \set{-1,+1}$ be $+1$ if groups $\gamma, \gamma'$ share the same columns, and $-1$ otherwise. Then
\begin{align*}
\Cov_{\gamma,\gamma'} := \Cov(\estgen_{\gamma},\estgen_{\gamma'}) & =  \eta_{\gamma,\gamma'}^\tB \frac{I_CI_T}{I_\gamma I_{\gamma'}}\frac{1}{2I}\left(\nu^{\tB}_\gamma + \nu^{\tB}_{\gamma'} - \xi^{\tB}_{\gamma,\gamma'}\right) \\
 + &  \eta_{\gamma,\gamma'}^\tS \frac{J_CJ_T}{J_\gamma J_{\gamma'}}\frac{1}{2J} \left(\nu^{\tS}_\gamma + \nu^{\tS}_{\gamma'} - \xi^{\tS}_{\gamma,\gamma'} \right) \\
 + & \eta_{\gamma,\gamma'}^\tB \eta_{\gamma,\gamma'}^\tS \frac{I_CI_TJ_CJ_T}{I_\gamma J_\gamma I_{\gamma'}J_{\gamma'}}\frac{1}{2IJ} \left(\nu^{\tBS}_\gamma + \nu^{\tBS}_{\gamma'} - \xi^{\tBS}_{\gamma,\gamma'}\right)
\end{align*}
\end{proposition}
\begin{proof}
By \cref{lemma:hoeffding_decomposition}, it can be seen that
\[
\Cov(\edavg{y}_{\gamma}, \edavg{y}_{\gamma'}) = \eta_{\gamma,\gamma'}^\tB \Cov(\overline{\e}_{\gamma}^{\tB}, \overline{\e}_{\gamma'}^{\tB}) + \eta_{\gamma,\gamma'}^\tS \Cov(\overline{\e}_{\gamma}^{\tS}, \overline{\e}_{\gamma'}^{\tS}) + \eta_{\gamma,\gamma'}^\tB \eta_{\gamma,\gamma'}^\tS\Cov(\dbar{\e}_{\gamma}^{\tBS}, \dbar{\e}_{\gamma'}^{\tBS}).
\]
The computation in Lemma A.8 of \cite{masoero2024multiple} shows
\begin{align*}
    \Cov(\overline{\e}_{\gamma}^{\tB}, \overline{\e}_{\gamma'}^{\tB}) & = \frac{1}{I_\gamma I_{\gamma'}} \sum_{i=1}^I \left(\frac{I_C I_T}{I^2(I-1)} + \frac{I_CI_T}{I^2} \right) \delta_i^\tB(\gamma) \delta_i^\tB(\gamma')\\
    & = \frac{I_CI_T}{I_\gamma I_{\gamma'} I(I-1)} \sum_{i=1}^I \delta_i^\tB(\gamma) \delta_i^\tB(\gamma').
\end{align*}
Next, 
\begin{align*}
\xi^\tB_{\gamma,\gamma'} & = \nu^\tB_\gamma + \nu^\tB_{\gamma'} - \frac{2}{I-1}\sum_{i=1}^I \delta_i^\tB(\gamma) \delta_i^\tB(\gamma') \\
& = \nu^\tB_\gamma + \nu^\tB_{\gamma'} - \frac{2I_\gamma I_{\gamma'} I}{I_CI_T} \Cov(\overline{\e}_{\gamma}^{\tB}, \overline{\e}_{\gamma'}^{\tB}).
\end{align*}
Rearranging gives 
\[
\Cov(\overline{\e}_{\gamma}^{\tB}, \overline{\e}_{\gamma'}^{\tB}) = \frac{I_C I_T}{2I_\gamma I_{\gamma'} I}\left(\nu^\tB_\gamma + \nu^\tB_{\gamma'} - \xi^\tB_{\gamma,\gamma'}\right).
\]
An analogous computation gives
\[
\Cov(\overline{\e}_{\gamma}^{\tS}, \overline{\e}_{\gamma'}^{\tS}) = \frac{J_C J_T}{2J_\gamma J_{\gamma'} J}\left(\nu^\tS_\gamma + \nu^\tS_{\gamma'} - \xi^\tS_{\gamma,\gamma'}\right).
\]
For the buyer-seller term, the computation in Lemma A.8 of \cite{masoero2024multiple} gives
\[
\Cov(\dbar{\e}_{\gamma}^{\tBS}, \dbar{\e}_{\gamma'}^{\tBS}) = \frac{I_CI_TJ_CJ_T}{I_\gamma I_{\gamma'} J_\gamma J_{\gamma'}}\frac{1}{IJ(I-1)(J-1)}\sum_{i,j} \delta_{ij}^\tBS(\gamma)\delta_{ij}^\tBS(\gamma')
\]
Expanding $\xi^{\tBS}_{\gamma,\gamma'}$ and rearranging as before gives 
\[
\Cov(\dbar{\e}_{\gamma}^{\tBS}, \dbar{\e}_{\gamma'}^{\tBS})  = \frac{I_CI_TJ_CJ_T}{I_\gamma I_{\gamma'} J_\gamma J_{\gamma'}}\frac{1}{2IJ}\left(\nu^\tBS_\gamma + \nu^\tBS_{\gamma'} - \xi_{\gamma,\gamma'}^{\tBS}\right),
\]
as desired.
\end{proof}

The following corollary records the explicit formulas for the covariances.

\begin{corollary}[Explicit Covariance Formulas; Lemma A.8 of \cite{masoero2024multiple}]
\begin{align*}
    \Cov_{\tr,\ib} & =  \frac{I_C}{2I_{T}I}\left(\nu^{\tB}_{\tr} + \nu^{\tB}_{\ib} - \xi^{\tB}_{\tr,\ib}\right) - \frac{1}{2J} \left(\nu^{\tS}_{\tr} + \nu^{\tS}_{\ib} - \xi^{\tS}_{\tr,\ib} \right) - \frac{I_C}{I_T}\frac{1}{2IJ} \left(\nu^{\tBS}_{\tr} + \nu^{\tBS}_{\ib} - \xi^{\tBS}_{\tr,\ib}\right) \\
    \Cov_{\tr,\is} & = - \frac{1}{2I}\left(\nu^{\tB}_{\tr}+ \nu^{\tB}_{\is} - \xi^{\tB}_{\tr,\is}\right) + \frac{J_C}{J_T}\frac{1}{2J} \left(\nu^{\tS}_{\tr} + \nu^{\tS}_{\is} - \xi^{\tS}_{\tr,\is} \right) - \frac{J_C}{J_T}\frac{1}{2IJ} \left(\nu^{\tBS}_{\tr} + \nu^{\tBS}_{\is} - \xi^{\tBS}_{\tr,\is}\right) \\
    \Cov_{\tr,\cc} & = - \frac{1}{2I}\left(\nu^{\tB}_{\tr}+ \nu^{\tB}_{\cc} - \xi^{\tB}_{\tr,\cc}\right) - \frac{1}{2J} \left(\nu^{\tS}_{\tr} + \nu^{\tS}_{\cc} - \xi^{\tS}_{\tr,\cc} \right) + \frac{1}{2IJ} \left(\nu^{\tBS}_{\tr} + \nu^{\tBS}_{\cc} - \xi^{\tBS}_{\tr,\cc}\right) \\
    \Cov_{\ib,\is} & = - \frac{1}{2I}\left(\nu^{\tB}_{\ib}+ \nu^{\tB}_{\is} - \xi^{\tB}_{\ib,\is}\right) - \frac{1}{2J} \left(\nu^{\tS}_{\ib} + \nu^{\tS}_{\is} - \xi^{\tS}_{\ib,\is} \right) + \frac{1}{2IJ} \left(\nu^{\tBS}_{\ib} + \nu^{\tBS}_{\is} - \xi^{\tBS}_{\ib,\is}\right) \\
    \Cov_{\ib,\cc} & = - \frac{1}{2I}\left(\nu^{\tB}_{\ib} + \nu^{\tB}_{\cc} - \xi^{\tB}_{\ib,\cc}\right) + \frac{J_T}{J_C}\frac{1}{2J} \left(\nu^{\tS}_{\ib} + \nu^{\tS}_{\cc} - \xi^{\tS}_{\ib,\cc} \right) - \frac{J_T}{J_C}\frac{1}{2IJ} \left(\nu^{\tBS}_{\ib} + \nu^{\tBS}_{\cc} - \xi^{\tBS}_{\ib,\cc}\right) \\
    \Cov_{\is,\cc} & =  \frac{I_T}{I_C}\frac{1}{2I}\left(\nu^{\tB}_{\is}+ \nu^{\tB}_{\cc} - \xi^{\tB}_{\is,\cc}\right) - \frac{1}{2J} \left(\nu^{\tS}_{\is} + \nu^{\tS}_{\cc} - \xi^{\tS}_{\is,\cc} \right) - \frac{I_T}{I_C}\frac{1}{2IJ} \left(\nu^{\tBS}_{\is} + \nu^{\tBS}_{\cc} - \xi^{\tBS}_{\is,\cc}\right)
\end{align*}    

\end{corollary}

\subsection{Direct Effect}
\label{sec:optimal_direct_effect_derivation}

Recalling the direct effect estimator $\esttr - \estib - \estis + \estcc$, the population variance of the estimator can be written explicitly using the variance and covariance formulas earlier in this section. The resulting expression is:
\begin{align*}
& \frac{I_C}{I_T I}\nu^{\tB}_\tr + \frac{J_C}{J_T J}\nu^{\tS}_\tr + \frac{I_C}{I_T I}\frac{J_C}{J_T J}\nu^{\tBS}_\tr \\    
& + \frac{I_C}{I_T I}\nu^{\tB}_\ib + \frac{J_T}{J_C J}\nu^{\tS}_\ib + \frac{I_C}{I_T I}\frac{J_T}{J_C J}\nu^{\tBS}_\ib \\    
& +\frac{I_T}{I_C I}\nu^{\tB}_\is + \frac{J_C}{J_T J}\nu^{\tS}_\is + \frac{I_T}{I_C I}\frac{J_C}{J_T J}\nu^{\tBS}_\is \\    
& + \frac{I_T}{I_C I}\nu^{\tB}_\cc + \frac{J_T}{J_C J}\nu^{\tS}_\cc + \frac{I_T}{I_C I}\frac{J_T}{J_C J}\nu^{\tBS}_\cc \\
& - \frac{I_C}{I_{T}I}\left(\nu^{\tB}_{\tr} + \nu^{\tB}_{\ib} - \xi^{\tB}_{\tr,\ib}\right) + \frac{1}{J} \left(\nu^{\tS}_{\tr} + \nu^{\tS}_{\ib} - \xi^{\tS}_{\tr,\ib} \right) + \frac{I_C}{I_T}\frac{1}{IJ} \left(\nu^{\tBS}_{\tr} + \nu^{\tBS}_{\ib} - \xi^{\tBS}_{\tr,\ib}\right) \\
&  + \frac{1}{I}\left(\nu^{\tB}_{\tr}+ \nu^{\tB}_{\is} - \xi^{\tB}_{\tr,\is}\right) - \frac{J_C}{J_T}\frac{1}{J} \left(\nu^{\tS}_{\tr} + \nu^{\tS}_{\is} - \xi^{\tS}_{\tr,\is} \right) + \frac{J_C}{J_T}\frac{1}{IJ} \left(\nu^{\tBS}_{\tr} + \nu^{\tBS}_{\is} - \xi^{\tBS}_{\tr,\is}\right) \\
&  - \frac{1}{I}\left(\nu^{\tB}_{\tr}+ \nu^{\tB}_{\cc} - \xi^{\tB}_{\tr,\cc}\right) - \frac{1}{J} \left(\nu^{\tS}_{\tr} + \nu^{\tS}_{\cc} - \xi^{\tS}_{\tr,\cc} \right) + \frac{1}{IJ} \left(\nu^{\tBS}_{\tr} + \nu^{\tBS}_{\cc} - \xi^{\tBS}_{\tr,\cc}\right) \\
& - \frac{1}{I}\left(\nu^{\tB}_{\ib}+ \nu^{\tB}_{\is} - \xi^{\tB}_{\ib,\is}\right) - \frac{1}{J} \left(\nu^{\tS}_{\ib} + \nu^{\tS}_{\is} - \xi^{\tS}_{\ib,\is} \right) + \frac{1}{IJ} \left(\nu^{\tBS}_{\ib} + \nu^{\tBS}_{\is} - \xi^{\tBS}_{\ib,\is}\right) \\
& + \frac{1}{I}\left(\nu^{\tB}_{\ib} + \nu^{\tB}_{\cc} - \xi^{\tB}_{\ib,\cc}\right) - \frac{J_T}{J_C}\frac{1}{J} \left(\nu^{\tS}_{\ib} + \nu^{\tS}_{\cc} - \xi^{\tS}_{\ib,\cc} \right) + \frac{J_T}{J_C}\frac{1}{IJ} \left(\nu^{\tBS}_{\ib} + \nu^{\tBS}_{\cc} - \xi^{\tBS}_{\ib,\cc}\right) \\
& - \frac{I_T}{I_C}\frac{1}{I}\left(\nu^{\tB}_{\is}+ \nu^{\tB}_{\cc} - \xi^{\tB}_{\is,\cc}\right) + \frac{1}{J} \left(\nu^{\tS}_{\is} + \nu^{\tS}_{\cc} - \xi^{\tS}_{\is,\cc} \right) + \frac{I_T}{I_C}\frac{1}{IJ} \left(\nu^{\tBS}_{\is} + \nu^{\tBS}_{\cc} - \xi^{\tBS}_{\is,\cc}\right)
\end{align*}
A close inspection reveals that all first order terms $\nu_{\gamma}^{\tB},\nu_{\gamma}^{\tS}$ in the expression cancel. Further, grouping all the $\nu_{\gamma}^{\tBS}$ terms yields $\frac{1}{I_\gamma J_\gamma}\nu_{\gamma}^{\tBS}$ for each $\gamma \in \set{\tr,\ib,\is,\cc}.$  Thus we may simplify the expression above to
\begin{align*}
& \frac{1}{I_T J_T }\nu^{\tBS}_\tr + \frac{1}{I_T J_C }\nu^{\tBS}_\ib + \frac{1}{I_C J_T }\nu^{\tBS}_\is + \frac{1}{I_C J_C }\nu^{\tBS}_\cc \\   
& + \frac{I_C}{I_{T}I}\xi^{\tB}_{\tr,\ib}
  - \frac{1}{I}\xi^{\tB}_{\tr,\is} 
  + \frac{1}{I}\xi^{\tB}_{\tr,\cc}
 + \frac{1}{I} \xi^{\tB}_{\ib,\is}
 - \frac{1}{I} \xi^{\tB}_{\ib,\cc}
 + \frac{I_T}{I_C}\frac{1}{I}\xi^{\tB}_{\is,\cc} \\
 &  - \frac{1}{J}   \xi^{\tS}_{\tr,\ib}  
+ \frac{J_C}{J_T}\frac{1}{J} \xi^{\tS}_{\tr,\is}
+ \frac{1}{J} \xi^{\tS}_{\tr,\cc} 
+ \frac{1}{J} \xi^{\tS}_{\ib,\is} 
+ \frac{J_T}{J_C}\frac{1}{J} \xi^{\tS}_{\ib,\cc} 
- \frac{1}{J} \xi^{\tS}_{\is,\cc} \\
&  - \frac{I_C}{I_T}\frac{1}{IJ}  \xi^{\tBS}_{\tr,\ib} 
  - \frac{J_C}{J_T}\frac{1}{IJ}  \xi^{\tBS}_{\tr,\is} 
   - \frac{1}{IJ}  \xi^{\tBS}_{\tr,\cc}
  - \frac{1}{IJ}  \xi^{\tBS}_{\ib,\is}
 - \frac{J_T}{J_C}\frac{1}{IJ}  \xi^{\tBS}_{\ib,\cc}
  - \frac{I_T}{I_C}\frac{1}{IJ} \xi^{\tBS}_{\is,\cc}.
\end{align*}

\subsection{Total Effect}
\label{sec:total_effect_variance_formulas}

From the results of \cref{sec:variance_formulas}, the population variance formula for the total effect estimator is given as follows.
\[
\Var(\esttr - \estcc) = \Var(\esttr) + \Var(\estcc) - 2\Cov(\esttr,\estcc).
\]
Explicitly,

\begin{align}
\begin{split}
    \Var(\esttr) & = \frac{I_{\cc}}{I_{\tr}I}\nu_{\tr}^{\tB} + \frac{J_{\cc}}{J_{\tr}J}\nu_{\tr}^{\tS} + \frac{I_{\cc}J_{\cc}}{I_{\tr}IJ_{\tr}J}\nu_{\tr}^{\tBS} \\
    \Var(\estcc) & = \frac{I_{\tr}}{I_{\cc}I}\nu_{\cc}^{\tB} + \frac{J_{\tr}}{J_{\cc}J}\nu_{\cc}^{\tS} + \frac{I_{\tr}J_{\tr}}{I_{\cc}J_{\cc}IJ}\nu_{\cc}^{\tBS} \\
    \Cov(\esttr,\estcc) & = -\frac{1}{2I}\left(\nu_{\tr}^{\tB} + \nu_{\cc}^{\tB} - \xi_{\tr,\cc}^{\tB}\right) \\
     & - \frac{1}{2J}\left(\nu_{\tr}^{\tS} + \nu_{\cc}^{\tS} - \xi_{\tr,\cc}^{\tS}\right)\\
     & + \frac{1}{2IJ}\left(\nu_{\tr}^{\tBS} + \nu_{\cc}^{\tBS} - \xi_{\tr,\cc}^{\tBS}\right)
\end{split}
\end{align}

It is helpful to rearrange some terms in the variance formula above. Notice that $\alpha_{\tr}^{\tB} + \frac{1}{I} = \frac{1}{I_{\tr}}, \alpha_{\tr}^{\tS} + \frac{1}{J} = \frac{1}{J_{\tr}}$ and similarly with $\alpha_{\cc}^{\tB} + \frac{1}{I} = \frac{1}{I_{\cc}}, \alpha_{\cc}^{\tS} + \frac{1}{J} = \frac{1}{J_{\cc}}$. As a result,

\begin{align}
\begin{split}
\Var(\esttr - \estcc) & = \frac{1}{I_{\tr}}\nu_{\tr}^{\tB} + \frac{1}{J_{\tr}}\nu_{\tr}^{\tS} + \frac{1}{IJ}\left( \frac{I_{\cc}J_{\cc}}{I_{\tr}J_{\tr}} - 1\right)\nu_{\tr}^{\tBS}\\
& + \frac{1}{I_{\cc}}\nu_{\cc}^{\tB} + \frac{1}{J_{\cc}}\nu_{\cc}^{\tS} + \frac{1}{IJ}\left(\frac{I_{\tr}J_{\tr}}{I_{\cc}J_{\cc}} - 1\right)\nu_{\cc}^{\tBS}\\
& - \frac{1}{I}\xi_{\tr,\cc}^{\tB} - \frac{1}{J}\xi_{\tr,\cc}^{\tS} + \frac{1}{IJ}
\xi_{\tr,\cc}^{\tBS}\label{eq:var_total_effect}
\end{split}
\end{align}

\subsection{Spillover Effects}

For completeness, we also include variance formulas for estimators of the spillover effects. The \textit{buyer spillover} and \textit{seller spillover} effects are defined by $\dbar{y}_{\ib} - \dbar{y}_{\cc}, \dbar{y}_{\is} - \dbar{y}_{\cc}$ respectively. Consider the non-interacted regression adjustment \eqref{eq:imputation_estimators} for the buyer spillover, which we denote $\tau_{\text{bs}}(\beta)$, and $\tau_{\text{ss}}(\beta)$ for the seller spillover. Using the full formulas for the variances and covariances of group means, we have
\begin{align*}
    \Var(\tau_{\text{bs}}(\beta)) & =  \frac{I_C}{I_T I}\nu^{\tB}_\ib + \frac{J_T}{J_C J}\nu^{\tS}_\ib + \frac{I_C}{I_T I}\frac{J_T}{J_C J}\nu^{\tBS}_\ib \\ 
    + & \frac{I_T}{I_C I}\nu^{\tB}_\cc + \frac{J_T}{J_C J}\nu^{\tS}_\cc + \frac{I_T}{I_C I}\frac{J_T}{J_C J}\nu^{\tBS}_\cc \\
    + & \frac{1}{I}\left(\nu^{\tB}_{\ib} + \nu^{\tB}_{\cc} - \xi^{\tB}_{\ib,\cc}\right) - \frac{J_T}{J_C}\frac{1}{J} \left(\nu^{\tS}_{\ib} + \nu^{\tS}_{\cc} - \xi^{\tS}_{\ib,\cc} \right) + \frac{J_T}{J_C}\frac{1}{IJ} \left(\nu^{\tBS}_{\ib} + \nu^{\tBS}_{\cc} - \xi^{\tBS}_{\ib,\cc}\right).
\end{align*}

The seller terms cancel, yielding
\begin{align*}
    \Var(\tau_{\text{bs}}(\beta)) & =  \frac{1}{I_T}\nu^{\tB}_\ib  +  \frac{1}{I_C}\nu^{\tB}_\cc + \frac{J_T}{I_TJ_C J}\nu^{\tBS}_\ib + \frac{J_T}{I_CJ_C J}\nu^{\tBS}_\cc \\
    - & \frac{1}{I}\xi^{\tB}_{\ib,\cc} + \frac{J_T}{J_CJ} \xi^{\tS}_{\ib,\cc} - \frac{J_T}{J_C}\frac{1}{IJ}\xi^{\tBS}_{\ib,\cc}.
\end{align*}

Carrying out the same computation for the seller spillover effects, we obtain
\begin{align*}
    \Var(\tau_{\text{ss}}(\beta)) & =  \frac{I_T}{I_C I}\nu^{\tB}_\is + \frac{J_C}{J_T J}\nu^{\tS}_\is + \frac{I_T}{I_C I}\frac{J_C}{J_T J}\nu^{\tBS}_\is \\
    + & \frac{I_T}{I_C I}\nu^{\tB}_\cc + \frac{J_T}{J_C J}\nu^{\tS}_\cc + \frac{I_T}{I_C I}\frac{J_T}{J_C J}\nu^{\tBS}_\cc \\
    - & \frac{I_T}{I_C}\frac{1}{I}\left(\nu^{\tB}_{\is}+ \nu^{\tB}_{\cc} - \xi^{\tB}_{\is,\cc}\right) + \frac{1}{J} \left(\nu^{\tS}_{\is} + \nu^{\tS}_{\cc} - \xi^{\tS}_{\is,\cc} \right) + \frac{I_T}{I_C}\frac{1}{ IJ} \left(\nu^{\tBS}_{\is} + \nu^{\tBS}_{\cc} - \xi^{\tBS}_{\is,\cc}\right).
\end{align*}
Here, the buyer terms cancel, yielding
\begin{align*}
    \Var(\tau_{\text{ss}}(\beta)) & =  \frac{1}{J_T}\nu^{\tS}_\is  +  \frac{1}{J_C}\nu^{\tS}_\cc + \frac{I_T}{J_TI_C I}\nu^{\tBS}_\is + \frac{I_T}{J_CI_C I}\nu^{\tBS}_\cc \\
    + & \frac{I_T}{I_C}\frac{1}{I}\xi^{\tB}_{\is,\cc} - \frac{1}{J} \xi^{\tS}_{\is,\cc} - \frac{I_T}{I_C}\frac{1}{IJ} \xi^{\tBS}_{\is,\cc}.
\end{align*}

\section{Proofs for \cref{sec:methodology}}

\begin{proof}[Proof of \cref{ex:optimal_adjustment_direct_effect}]

To make $\tau$ identifiable, we enforce the following constraints:
\begin{align*}
& \sum_i (1 - W_i^\tB)\mu_i^\tB = \sum_j (1 - W_j^\tS)\mu_j^\tS = 0 \\
& \sum_i (1 - W_i^\tB)\delta_i^\tB = \sum_j (1 - W_j^\tS)\alpha_j^\tS = 0 \\
& \sum_i W_i^\tB(\mu_i^\tB + \alpha_i^\tB) = \sum_j W_j^\tS(\mu_j^\tS + \delta_j^\tS) = 0 \\
& \sum_i W_i^\tB(\delta_i^\tB + \tau_i^\tB) = \sum_j W_j^\tS(\alpha_j^\tS + \tau_j^\tS) = 0.
\end{align*}

Define 
\begin{align*}
    \Delta_{\tr} & := \mu + \alpha + \delta + \tau \\
    \Delta_{\ib} & := \mu + \alpha  \\
    \Delta_{\is} & := \mu + \delta  \\
    \Delta_{\cc} & := \mu 
\end{align*}
and analogously define $\Delta_{\gamma,i}^{\tB},\Delta^{\tS}_{\gamma,j}$ for the fixed effects. The WLS objective can then be written
\begin{align*}
    & \frac{1}{I_T^2 J_T^2}\sum_{i,j} W_i^{\tB} W_j^{\tS} \left(y_{ij} - X_{ij}^\top \beta - \Delta_{\tr} - \Delta_{\tr,i}^{\tB} - \Delta_{\tr,j}^{\tS} \right)^2 \\
    + & \frac{1}{I_T^2 J_C^2}\sum_{i,j} W_i^{\tB} (1-W_j^{\tS}) \left(y_{ij} - X_{ij}^\top \beta - \Delta_{\ib} - \Delta_{\ib,i}^{\tB} - \Delta_{\ib,j}^{\tS} \right)^2 \\
    + & \frac{1}{I_C^2 J_T^2}\sum_{i,j} (1-W_i^{\tB}) W_j^{\tS} \left(y_{ij} - X_{ij}^\top \beta - \Delta_{\is} - \Delta_{\is,i}^{\tB} - \Delta_{\is,j}^{\tS} \right)^2 \\
    + & \frac{1}{I_C^2 J_C^2}\sum_{i,j}(1- W_i^{\tB})(1- W_j^{\tS}) \left(y_{ij} - X_{ij}^\top \beta - \Delta_{\cc} - \Delta_{\cc,i}^{\tB} - \Delta_{\cc,j}^{\tS} \right)^2.
\end{align*}
The identifiability constraints are chosen exactly to force $\sum_{i \in \gamma} \Delta_{\gamma,i}^\tB = \sum_{j \in \gamma} \Delta_{\gamma,j}^\tS = 0$.
Recall the well-known fact that regardless of $\beta$, $\Delta_{\gamma},\Delta_{\gamma,i}^{\tB},\Delta_{\gamma,j}^{\tS}$ doubly decenters each of the matrices. Thus the optimal $\beta$ must minimize
\begin{align*}
    & \frac{1}{I_T^2 J_T^2}\sum_{i,j} W_i^{\tB} W_j^{\tS} \left(y_{ij}^{(dcr)} - X_{ij}^{(dcr)\top} \beta \right)^2 \\
    + & \frac{1}{I_T^2 J_C^2}\sum_{i,j} W_i^{\tB} (1-W_j^{\tS}) \left(y_{ij}^{(dcr)} - X_{ij}^{(dcr)\top} \beta\right)^2 \\
    + & \frac{1}{I_C^2 J_T^2}\sum_{i,j} (1-W_i^{\tB}) W_j^{\tS} \left(y_{ij}^{(dcr)} - X_{ij}^{(dcr)\top} \beta \right)^2 \\
    + & \frac{1}{I_C^2 J_C^2}\sum_{i,j} (1- W_i^{\tB})(1- W_j^{\tS})\left(y_{ij}^{(dcr)} - X_{ij}^{(dcr)\top} \beta\right)^2,
\end{align*}
where the double decentering is done within the respective groups $\gamma$. We can equivalently write the objective as 
\begin{equation}
\label{eq:direct_effect_proof_twfe}
\sum_{\gamma \in \Gamma} \frac{1}{I_\gamma^2 J_\gamma^2} \sum_{(i,j) \in \gamma} \left(y_{ij}^{(dcr)}(\gamma) - X_{ij}^{(dcr)}(\gamma)^{\top} \beta \right)^2,
\end{equation}
from which it is apparent that the optimal $\beta$ is given by $\tilde{Z}_{\dir}^{-1}\tilde{u}_{\dir}$. By the identifiability constraints, we have 
\[
\Delta_\gamma = \frac{1}{I_\gamma J_\gamma} \sum_{(i,j)\in \gamma} (y_{ij}(\gamma) - X_{ij}^\top \beta).
\]
Recalling that $\tau = \Delta_{\tr} - \Delta_{\ib} - \Delta_{\is} + \Delta_{\cc}$ we see that the fitted coefficient $\tau$ is equal to $\hat{\tau}_{\dir}(\beta),$ with $\beta$ minimizing \eqref{eq:direct_effect_proof_twfe}.
\end{proof}

\section{Additional Details on Improved Theory for MRDs}
\label{sec:appendix_improved_theory}

In this section, we prove \cref{thm:improved_clt} and present examples for which the result is stronger than the CLT of \cite{masoero2024multiple}. We can start by using a finite population central limit theorem in Wasserstein distance. The conclusion is slightly modified to suit our purposes. 
\begin{theorem}[Adaptation of Theorem 5.1 in \cite{goldstein20071}]
\label{thm:goldstein_adaptation}
Let $z_1,\dots,z_N$ be $N$ real numbers, and $\cl{M}$ be a uniformly random subset of $[N]$ of size $n$. Let $S = \sum_{k \in \cl{M}} z_k.$ Then
\[
d_W \left(\cl{L}\left( S - \E S\right),\textsf{N}(0,\Var(S)) \right) \leq 9 \frac{\sum_{k=1}^N |z_k - \overline{z}|^3}{\sum_{k=1}^N (z_k - \overline{z})^2}
\]
\end{theorem}
\begin{proof}
The result of Goldstein states that 
\[
d_W \left(\cl{L}\left( \frac{S - \E S}{\sqrt{\Var(S)}}\right),\textsf{N}(0,1) \right) \leq 4 \frac{n}{N}\frac{N-n}{N-1} \left(1 + \frac{n}{N}\right)^2 \frac{\sum_{k=1}^N |z_k - \overline{z}|^3}{\Var(S)^{3/2}}.
\]
We can assume $n \leq N/2$ without loss of generality since $-(S - \E S) = - \sum_{k \in \cl{M}} (z_k - \bar{z}) = \sum_{k \in \cl{M}^c} (z_k - \bar{z})$. The variance formula for sampling without replacement gives $\Var(S) = \frac{n(N-n)}{N(N-1)} \sum_{i=1}^N  (z_i - \bar{z})^2$. By scaling properties of Wasserstein distance, 
\begin{align*}
d_W \left(\cl{L}\left( S - \E S\right),\textsf{N}(0,\Var(S)) \right) & \leq \sqrt{\Var(S)} d_W \left(\cl{L}\left( \frac{S - \E S}{\sqrt{\Var(S)}}\right),\textsf{N}(0,1) \right)\\
& \leq 4 \frac{n}{N}\frac{N-n}{N-1} \left(1 + \frac{n}{N}\right)^2 \frac{\sum_{k=1}^N |z_k - \overline{z}|^3}{\Var(S)} \\
& = 4 \left(1 + \frac{n}{N}\right)^2 \frac{\sum_{k=1}^N |z_k - \overline{z}|^3}{\sum_{k=1}^N |z_k - \overline{z}|^2}.
\end{align*}
The proof is completed by noting that $1 + \frac{n}{N} \leq 3/2.$ 
\end{proof}

With these preliminaries, we are able to now prove \cref{thm:improved_clt}, which is restated here for convenience. The basic idea is to condition on $\scr{S}$, the collection of seller indicator variables $W_j^\tS$. Since the random variables in $\scr{S}$ are discrete, conditioning amounts to considering all possible values of $(W_j^{\tS})_{j=1}^J \in \set{0,1}^J,$ so there are no measurability issues. Fixing $\scr{S}$, by independence of $\cl{B}$ and $\scr{S}$, \cref{thm:goldstein_adaptation} applies on the finite population $\frac{1}{I_T} \eavg{y}_{i,\bullet}$. The theorem shows that conditional on $\tS, \edavg{y}$ is approximately normal with mean $\mu_\tS := \E \left[\edavg{y} \ | \ \scr{S} \right]$ and conditional variance $\sigma^2_{\scr{S}} := \Var\left( \edavg{y} \ | \ \scr{S} \right)$. Under some conditions, one can show that $\sigma^2_{\scr{S}}$ concentrates around its expectation $\E[\Var\left( \edavg{y} \ | \ \scr{S} \right)].$ Now, we can apply \cref{thm:goldstein_adaptation} again to $\mu_{\tS}$, showing that $\mu_{\tS}$ is approximately normal with mean $\mu$ and variance $\Var(\mu_{\tS}).$ Combining these facts roughly suggests that $\edavg{y}$ is approximately normal with mean $\mu$ and variance $\E\left[\Var\left( \edavg{y} \ | \ \scr{S} \right)\right] + \Var(\mu_{\tS}) = \Var(\edavg{y}).$ The proof makes this precise.

\thmimprovedclt*

Throughout this section, we will define the notation 
\begin{equation}
    d_{ij} := y_{ij} - \dbar{y}
\end{equation}
to denote the centered matrix $Y$.

\begin{proof}[Proof of \cref{thm:improved_clt}]
Let $\mu := \E \left[\edavg{y}\right], \mu_\tS := \E \left[\edavg{y} \ | \ \scr{S} \right].$ Similarly, let $\sigma^2_{\scr{S}} := \Var\left( \edavg{y} \ | \ \scr{S} \right)$ and recall the notation $\sigma_{\textsf{Tot}}^2 := \Var\left( \edavg{y}\right)$. Let $Z_1 \sim \textsf{N}(0,1)$ and independent of $\scr{S}$. By the law of total variance, we have $\sigma_{\textsf{Tot}}^2= \E[\sigma_{\scr{S}}^2] + \Var(\mu_\tS).$ Then by triangle inequality,
\begin{align}
\label{eq:clt_decomp_term_I}
d_W\left( \cl{L}(\edavg{y}), \textsf{N}\left( \mu, \sigma_{\textsf{Tot}}^2\right)\right) & \leq d_W\left( \cl{L}(\edavg{y}), \cl{L}\left(\mu_\tS + \sigma_{\scr{S}}Z_1 \right) \right) \\
\label{eq:clt_decomp_term_II}
& +  d_W\left( \cl{L}\left(\mu_\tS + \sigma_{\scr{S}}Z_1 \right), \cl{L}\left(\mu_\tS + \sqrt{\E\left[\sigma^2_{\scr{S}} \right]}Z_1 \right) \right) \\
\label{eq:clt_decomp_term_III}
& + d_W\left( \cl{L}\left(\mu_\tS + \sqrt{\E\left[\sigma^2_{\scr{S}} \right]}Z_1 \right), \textsf{N}(\mu,\Var[\mu_\tS]) * \textsf{N}\left(0, \E\left[\sigma_{\scr{S}}^2 \right] \right) \right)
\end{align}
Let us first bound \eqref{eq:clt_decomp_term_I}. Conditioning on $(W_j^\tS)_{j=1}^J = w$, we may write by \cref{lemma:wasserstein_tensorization} 
\[
d_W\left( \cl{L}(\edavg{y}), \cl{L}\left(\mu_\tS + \sigma_{\scr{S}}Z_1 \right) \right)\leq \E\left[d_W\left( \cl{L}(\edavg{y}), \cl{L}\left(\mu_\tS + \sigma_{\scr{S}}Z_1 \right) \ \bigg| \ (W_j^\tS)_{j=1}^J = w  \right)   \right].
\]
Apply \cref{thm:goldstein_adaptation}, fixing $(W_j^\tS)_{j=1}^J$  with $z_i = \frac{1}{I_T} \eavg{y}_{i,\bullet}$, to the sum $\frac{1}{I_T}\sum_i W_i^\tB \eavg{y}_{i,\bullet}$. The conditional average and variance we can calculate explicitly as
\begin{align*}
     \mu_\tS & := \E[S | \scr{S}] = \frac{1}{I}\sum_{i=1}^I \eavg{y}_{i,\bullet}  \\
     \sigma_{\scr{S}}^2 & := \Var[S | \scr{S}] = \left(\frac{1}{I_T} - \frac{1}{I}\right)\frac{1}{I-1}\sum_{i=1}^I \left(\eavg{y}_{i,\bullet} -  \mu_\tS \right)^2.
\end{align*}
Note that 
\begin{align*}
\E[\mu_\tS] & = \E[S] = \dbar{y}. \\
\Var[\mu_\tS] & = \left(\frac{1}{J_T} - \frac{1}{J}\right)\frac{1}{J-1}\sum_{j=1}^J \left(\overline{y}_{\bullet,j} - \dbar{y} \right)^2 = \left(\frac{1}{J_T} - \frac{1}{J}\right)\nu^\tS(y) \\
\E[\Var[S | \scr{S}]] & = \Var[S] - \Var[\mu_\tS] = \left(\frac{1}{I_T} - \frac{1}{I} \right) \nu^\tB(y) + \left(\frac{1}{I_T} - \frac{1}{I} \right)\left(\frac{1}{J_T} - \frac{1}{J} \right) \nu^\tBS(y).
\end{align*}
\cref{thm:goldstein_adaptation} gives
\[
d_W\left( \cl{L}(\edavg{y}), \cl{L}\left(\mu_\tS + \sigma_{\scr{S}}Z_1 \right) \ \bigg| \ (W_j^\tS)_{j=1}^J = w  \right)  \leq \frac{9}{I_T} \frac{\sum_{i=1}^I \left| \eavg{y}_{i,\bullet} - \frac{1}{I}\sum_{i'=1}^I\eavg{y}_{i',\bullet} \right|^3 }{\sum_{i=1}^I \left( \eavg{y}_{i,\bullet} - \frac{1}{I}\sum_{i'=1}^I\eavg{y}_{i',\bullet} \right)^2},
\]
so 
\[
d_W\left( \cl{L}(\edavg{y}), \cl{L}\left(\mu_\tS + \sigma_{\scr{S}}Z_1 \right) \right) \leq \frac{9}{I_T} \E \left[ \max_{i=1}^I \left| \eavg{y}_{i,\bullet} - \frac{1}{I}\sum_{i'=1}^I\eavg{y}_{i',\bullet} \right| \right].
\]
We can bound the right hand side by the expected suprema of sub-Gaussian random variables via \cref{lemma:sampling_subgaussianity}, after removing the expected value:
\begin{align*}
\E \left[\max_i \left| \eavg{y}_{i,\bullet} - \frac{1}{I}\sum_{i'=1}^I\eavg{y}_{i',\bullet} \right| \right] & \leq 2 \E \max_i \left|\eavg{y}_{i,\bullet} - \dbar{y} \right|\\
& \lesssim \max_{i=1}^I \left| \overline{y}_{i,\bullet} - \dbar{y} \right| + \E \left[\max_i \left| \eavg{y}_{i,\bullet} - 
 \overline{y}_{i,\bullet} \right| \right] \\
& \lesssim \max_{i=1}^I \left| \overline{y}_{i,\bullet} - \dbar{y} \right| + \sqrt{\log I} \max_i \sqrt{\Var \eavg{y}_{i,\bullet}}  \\
& \lesssim  \max_{i=1}^I \left| \overline{y}_{i,\bullet} - \dbar{y} \right| + \sqrt{\frac{\log I}{J}}  \max_{i,j} |y_{ij} - \overline{y}_{i,\bullet}|.
\end{align*}

Next, consider the second term \eqref{eq:clt_decomp_term_II}. First, if we condition on a particular realization $(W_j^\tS)_{j=1}^J = w$, then the obvious coupling gives
\[
 d_W\left( \cl{L}\left(\mu_\tS + \sigma_{\scr{S}}Z_1 \right), \cl{L}\left(\mu_\tS + \sqrt{\E\left[\sigma^2_{\scr{S}} \right]}Z_1  \right) \ \bigg| \ (W_j^\tS)_{j=1}^J = w \right) \leq \left|\sigma_{\scr{S}} - \sqrt{\E\left[\sigma^2_{\scr{S}} \right]} \right|\E[|Z_1|].
\]
Thus, by \cref{lemma:wasserstein_tensorization}, we have
\[
d_W\left( \cl{L}\left(\mu_\tS + \sigma_{\scr{S}}Z_1 \right), \cl{L}\left(\mu_\tS + \sqrt{\E\left[\sigma^2_{\scr{S}} \right]}Z_1 \right) \right) \lesssim \E \left|\sigma_{\scr{S}} - \sqrt{\E\left[\sigma^2_{\scr{S}} \right]} \right|.
\]
We upper bound this using vector concentration for sampling without replacement, using \cref{lemma:variance_concentration} to obtain
\[
\E \left|\sigma_{\scr{S}} - \sqrt{\E\left[\sigma^2_{\scr{S}} \right]} \right| \lesssim \frac{1}{I} \sqrt{\sigma_{\textsf{Tot}} \norm{d_{ij}}_{\text{op}}} + \frac{\Var(\edavg
y)^{1/2}}{\sqrt{I}}.
\]

Finally, for the third term \eqref{eq:clt_decomp_term_III}, rewrite $\mu_\tS = \frac{1}{J_T} \sum_{j=1}^J W_j^\tS \overline{y}_{\bullet,j}$. We apply \cref{lemma:wasserstein_convolution} and \cref{thm:goldstein_adaptation} again to give the upper bound
\begin{align*}
     d_W\left( \cl{L}\left(\mu_\tS + \sqrt{\E\left[\sigma^2_{\scr{S}} \right]}Z_1 \right), \textsf{N}(\mu,\Var[\mu_\tS]) * \textsf{N}\left(0, \E\left[\sigma_{\scr{S}}^2 \right] \right) \right) & \leq d_W\left(\cl{L}(\mu_\tS), \textsf{N}(\mu,\Var[\mu_\tS])\right) \\
 \leq & \frac{9}{J_T} \frac{\sum_{j=1}^J \left| \overline{y}_{\bullet,j} - \dbar{y} \right|^3 }{\sum_{j=1}^J \left( \overline{y}_{\bullet,j} - \dbar{y} \right)^2}.
\end{align*}
Putting the bounds together and simplifying using e.g.~$I_T \gtrsim I$, we get
\begin{align*}
    d_W\left( \cl{L}_{\edavg{y}}, \textsf{N}\left( \mu, \sigma_{\textsf{Tot}}^2\right)\right) & \lesssim \frac{1}{I} \max_{i=1}^I \left| \overline{y}_{i,\bullet} - \dbar{y} \right| + \frac{\sqrt{\log I}}{I^{3/2}}  \max_{i,j} |y_{ij} - \overline{y}_{i,\bullet}| \\
    & + \frac{1}{I} \sqrt{\sigma_{\textsf{Tot}}\norm{d_{ij}}_{\text{op}}} + \frac{\Var(\edavg
y)^{1/2}}{\sqrt{I}} \\
    & +  \frac{1}{J} \max_{j=1}^J \left| \overline{y}_{\bullet,j} - \dbar{y} \right|.
\end{align*}
Dividing through by the square root of the variance, we obtain
\begin{align*}
d_W\left( \cl{L}\left(\frac{\edavg{y}-\mu}{\sigma_{\textsf{Tot}}} \right), \textsf{N}\left( 0, 1 \right)\right) & \lesssim  \frac{1}{\sqrt{I}} + \frac{1}{I} \frac{\max_{i=1}^I \left| \overline{y}_{i,\bullet} - \dbar{y} \right|}{\sigma_{\textsf{Tot}}} + \frac{1}{J} \frac{\max_{j=1}^J \left| \overline{y}_{\bullet,j} - \dbar{y} \right|}{\sigma_{\textsf{Tot}}} \\
& + \frac{\sqrt{\log I}}{I^{3/2}}  
 \frac{\max_{i,j} |y_{ij} - \overline{y}_{i,\bullet}|}{\sigma_{\textsf{Tot}}} + \left(\frac{1}{I^2} \frac{\norm{y_{ij} - \dbar{y}}_{\text{op}}}{\sigma_{\textsf{Tot}}} \right)^{1/2}.
\end{align*} 
\end{proof}

\begin{lemma}
\label{lemma:variance_concentration}
We have the bound   
\[
\E \left|\sigma_{\scr{S}} - \sqrt{\E\left[\sigma^2_{\scr{S}} \right]} \right| \lesssim \frac{1}{I} \sqrt{\sigma_{\textsf{Tot}}\norm{d_{ij}}_{\operatorname{op}}} + \frac{\Var(\edavg
y)^{1/2}}{\sqrt{I}}.
\]
\end{lemma}
\begin{proof}
Recall that 
\[
\sigma_{\scr{S}}^2 \propto \frac{1}{I} \left(\frac{1}{I} \sum_{i=1}^I \eavg{y}_{i,\bullet}^2 - \mu_\tS^2 \right).
\]
Since $\sigma_{\scr{S}}$ remains unchanged when shifting $y_{ij}$ by a constant, we can replace $y_{ij}$ by $d_{ij} = y_{ij} - \dbar{y}.$ Let $m_{2,\tS} := \frac{1}{I} \sum_{i=1}^I \eavg{d}_{i,\bullet}^2$, so $I\sigma_{\scr{S}}^2 \propto m_{2,\tS} - \mu_\tS^2$. Then by the simple inequality $\left|\sqrt{x} - \sqrt{y} \right| \leq \sqrt{|x - y|}$,
\[
\E \left|\sigma_{\scr{S}} - \sqrt{\E\left[\sigma^2_{\scr{S}} \right]} \right| \leq \E \sqrt{\left|\sigma^2_\tS - \E\left[\sigma^2_{\scr{S}} \right] \right|}.
\]
Applying Jensen's inequality to the right-hand side above, then using the triangle inequality and
the elementary bound $\sqrt{x + y} \lesssim \sqrt{x} + \sqrt{y}$, we have 
\begin{equation}
\label{eq:CLT_proof_variance_concentration_helper}
\sqrt{I}\E \left|\sigma_{\scr{S}} - \sqrt{\E\left[\sigma^2_{\scr{S}} \right]} \right| \leq \sqrt{I\E|\sigma_{\scr{S}} - \E[\sigma_{\scr{S}}]|}\lesssim \sqrt{\E\left|m_{2,\tS} - \E[m_{2,\tS}] \right|} + \sqrt{\E\left|\mu_\tS^2 - \E[\mu_\tS^2] \right|}.
\end{equation}
We will upper bound both these quantities using subgaussianity and \cref{lemma:subexponential}. Firstly, $\sqrt{I m_{2,\tS}} = \sqrt{\sum_{i=1}^I \eavg{d}_{i,\bullet}^2}$ is subgaussian after centering using \cref{lemma:vector_concentration_sampling_without_replacement}. Apply the lemma taking $u_j$ to be the $J$ vectors $\frac{1}{J_T}(d_{1,j},\dots,d_{I,j})$, so $\sum_j W_j^\tS u_j = (\eavg{d}_{1,\bullet}^2, \ldots, \eavg{d}_{I,\bullet}^2)$ and $\|\sum_j W_j^\tS u_j\|_2^2 = Im_{2,\tS}$. Thus, 
\[
\Prob \left( \left| \sqrt{I m_{2,\tS}} - \E \sqrt{I m_{2,\tS}} \right| \geq t \right) \leq \exp\left(-\frac{J_T^2 t^2}{8 \norm{(d_{ij})}_{\text{op}}^2} \right):
\]
$\sqrt{m_{2,\tS}} - \E\sqrt{m_{2,\tS}}$ is subgaussian with parameter $\norm{(d_{ij})}_{\text{op}}/I^{3/2}$, up to absolute constants. Applying \cref{lemma:subexponential}, 
\begin{align*}
\E \left| m_{2,\tS} - \E[m_{2,\tS}]\right| & \lesssim \frac{\norm{(d_{ij})}_{\text{op}}^2 }{I^3} + \frac{\norm{(d_{ij})}_{\text{op}} }{I^{3/2}} \left|\E[\sqrt{m_{2,\tS}}]\right| \\
& \leq \frac{\norm{(d_{ij})}_{\text{op}}^2 }{I^3} + \frac{\norm{(d_{ij})}_{\text{op}} }{I^{3/2}} \sqrt{\E[m_{2,\tS}]}.
\end{align*}
Note the bounds $\norm{(d_{ij})}_{\text{op}} \leq \norm{(d_{ij})}_{\text{F}} \leq I^2\sigma_{\textsf{Tot}}$ by \cref{lemma:d_matrix_frobenius} along with
\begin{align*}
    \E[m_{2,\tS}] & = \frac{1}{I} \sum_i \E[\eavg{d}_{i,\bullet}^2]\\ 
    & \lesssim \frac{1}{I} \sum_i \left(\frac{1}{I^2} \sum_j (d_{ij} - \overline{d}_{i,\bullet})^2 + \overline{d}_{i,\bullet}^2 \right)\\ 
    & = \frac{1}{I^3} \sum_{i,j} (y_{ij} - \overline{y}_{i,\bullet})^2 + \frac{1}{I} \sum_i \overline{d}_{i,\bullet}^2 \\
    & = \frac{1}{I^3} \sum_{i,j} (\delta_{ij}^{\tB\tS} + \delta_j^{\tS})^2 + \frac{1}{I} \sum_i \overline{d}_{i,\bullet}^2
    \\
    & \lesssim \frac{1}{I}(\nu^\tBS + \nu^\tS) + \nu^\tB.
\end{align*}
Thus $\E[m_{2,\tS}] \lesssim I\sigma_{\textsf{Tot}}^2$. Putting these bounds together, we arrive at
\[
\E \left| m_{2,\tS} - \E[m_{2,\tS}]\right| \lesssim \frac{1}{I}\sigma_{\textsf{Tot}} \norm{d_{ij}}_{\text{op}}.
\]
Next, a direct application of \cref{lemma:sampling_subgaussianity} shows that $\mu_\tS - \E[\mu_\tS]$ is subgaussian with parameter $\sqrt{\frac{1}{J}\nu^\tS}$. Applying \cref{lemma:subexponential}, we find
\[
\E\left|\mu_\tS^2 - \E[\mu_\tS^2] \right| \lesssim \frac{1}{J}\nu^\tS + \sqrt{\frac{1}{J}\nu^\tS} \left|\E [\mu_{\tS}]\right| \lesssim \sigma_{\textsf{Tot}}^2,
\]
which follows since $\left|\E [\mu_{\tS}]\right| = 0$ by the assumption that $y_{ij}$ can be replaced by $d_{ij}.$ Putting these bounds together with \eqref{eq:CLT_proof_variance_concentration_helper}, we obtain
\[
\E \left|\sigma_{\scr{S}} - \sqrt{\E\left[\sigma^2_{\scr{S}} \right]} \right| \lesssim \frac{1}{I} \sqrt{\sigma_{\textsf{Tot}}\norm{d_{ij}}_{\text{op}}} + \frac{\sigma_{\textsf{Tot}}}{\sqrt{I}}.
\]
\end{proof}

\subsection{Examples}

We explore simple conditions and data generating processes under which we can expect asymptotic normality.
\begin{corollary}
\label{cor:simplified_clt}
Suppose $I,J \rightarrow \infty, I_T/I, J_T/J = O(1)$ and 
\begin{equation}
\label{eq:CLT_assumptions}
\max\left(\frac{1}{I} \max_{i=1}^I |\overline{y}_{i,\bullet} - \dbar{y}|, \frac{1}{J} \max_{j=1}^J |\overline{y}_{\bullet,j} - \dbar{y}|, \frac{\sqrt{\log I}}{I^{3/2}} \max_{ij} |y_{ij} - \overline{y}_{\bullet,j}|, \frac{1}{I^2} \norm{y_{ij} - \dbar{y}}_{\text{op}} \right) = o(\sigma_{\textsf{Tot}})
\end{equation}
Then 
\[
\frac{\edavg{y} - \dbar{y}}{\sigma_{\textsf{Tot}}} \Rightarrow \textsf{N}(0,1).
\]
\end{corollary}
These are the primitive conditions for the right hand side of \cref{thm:improved_clt} to go to zero. Next, we may apply the corollary above to obtain CLTs for MRD estimators.

\cormrdclt*
\begin{proof}
First, notice that for each empirical average $\edavg{y}_\gamma$, we may write
\[
\edavg{y}_\gamma = \frac{1}{I_T J_T} \sum_{i,j} W_i^\tB W_j^\tS z_{ij}(\gamma)
\]
for some ``pseudo"-potential outcomes $z_{ij}.$ Thus,
\begin{equation}
\label{eq:mrd_estimator_unified}
\hat{\tau}_{\mathbf{c}} = \frac{1}{I_T J_T} \sum_{i,j} W_i^\tB W_j^\tS (\sum_\gamma c_\gamma z_{ij}(\gamma))  
\end{equation}

Explicitly, 
\begin{align*}
& \sum_{i,j} (1-W_i^\tB) (1-W_j^\tS) y_{ij}(\cc)  = \sum_{i,j} y_{ij}(\cc) - \sum_{i,j} W_i^\tB y_{ij}(\cc)  - \sum_{i,j} W_j^\tS y_{ij}(\cc) + \sum_{i,j} W_i^\tB W_j^\tS y_{ij}(\cc)\\
& = \frac{IJ}{I_T J_T} \sum_{i,j} W_i^\tB W_j^\tS \dbar{y}(\cc) - \frac{J}{J_T} \sum_{i,j} W_i^\tB W_j^\tS  \overline{y}_{i,\bullet}(\cc)  - \frac{I}{I_T}\sum_{i,j} W_i^\tB W_j^\tS \overline{y}_{\bullet,j}(\cc) + \sum_{i,j} W_i^\tB W_j^\tS y_{ij}(\cc) \\
& = \sum_{i,j} W_i^\tB W_j^\tS \left( \frac{IJ}{I_T J_T} \dbar{y}_\cc - \frac{J}{J_T}\overline{y}_{i,\bullet}(\cc) -  \frac{I}{I_T}\overline{y}_{\bullet,j}(\cc) + y_{ij}(\cc) \right).
\end{align*}
Similar computations can be carried out for the other groups.
\begin{align*}
& \sum_{i,j} W_i^\tB (1-W_j^\tS) y_{ij}(\ib)  = \sum_{i,j} W_i^\tB y_{ij}(\ib) - \sum_{i,j} W_i^\tB W_j^\tS y_{ij}(\ib)\\
& = \frac{J}{J_T} \sum_{i,j} W_i^\tB W_j^\tS  \overline{y}_{i,\bullet}(\ib) - \sum_{i,j} W_i^\tB W_j^\tS y_{ij}(\ib) \\
& = \sum_{i,j} W_i^\tB W_j^\tS \left( \frac{J}{J_T}\overline{y}_{i,\bullet}(\ib) - y_{ij}(\ib) \right).
\end{align*}
Finally,
\begin{align*}
& \sum_{i,j} (1-W_i^\tB) W_j^\tS y_{ij}(\is)  = \sum_{i,j} W_j^\tS y_{ij}(\is) - \sum_{i,j} W_i^\tB W_j^\tS y_{ij}(\is)\\
& = \frac{I}{I_T} \sum_{i,j} W_i^\tB W_j^\tS  \overline{y}_{\bullet,j}(\is) - \sum_{i,j} W_i^\tB W_j^\tS y_{ij}(\is) \\
& = \sum_{i,j} W_i^\tB W_j^\tS \left( \frac{I}{I_T}\overline{y}_{\bullet,j}(\is) - y_{ij}(\is) \right).
\end{align*}




Let $\gamma_\cc := I_TJ_T/I_CI_C, \gamma_\is := J_T/J_C, \gamma_\ib := I_T/I_C$. Thus,
\begin{align}
    z_{ij}(\tr) & :=  y_{ij}(\tr)\label{eq:zij_tr}\\
    z_{ij}(\ib) & :=  \gamma_\ib\frac{I}{I_T}\overline{y}_{i,\bullet}(\ib) - \gamma_\ib y_{ij}(\ib)\label{eq:zij_ib}\\
    z_{ij}(\is) & := \gamma_\is\frac{J}{J_T}\overline{y}_{\bullet,j}(\is) - \gamma_\is y_{ij}(\is)\label{eq:zij_is}\\
    z_{ij}(\cc) & := \gamma_\cc\frac{IJ}{I_T J_T} \dbar{y}_\cc - \gamma_\cc\frac{J}{J_T}\overline{y}_{\bullet,j}(\cc) -  \gamma_\cc\frac{I}{I_T}\overline{y}_{i,\bullet}(\cc) + \gamma_\cc y_{ij}(\cc)\label{eq:zij_cc}.
\end{align}
Inspecting these formulas, it is evident that for each $\gamma$,
\begin{align*}
    & \left|\overline{z}_{i,\bullet}(\gamma) - \dbar{z}_\gamma\right| \lesssim \left|\overline{y}_{i,\bullet}(\gamma) - \dbar{y}_\gamma\right|\\
    & \left|\overline{z}_{\bullet,j}(\gamma) - \dbar{z}_\gamma\right| \lesssim \left|\overline{y}_{\bullet,j}(\gamma) - \dbar{y}_\gamma\right|\\
    & \left|z_{ij}(\gamma) - \overline{z}_{i,\bullet}(\gamma)\right| \lesssim \left|y_{ij}(\gamma) - \overline{y}_{i,\bullet}(\gamma)\right| + \left|\overline{y}_{\bullet,j}(\gamma) - \dbar{y}_\gamma\right| \\
    & \left|z_{ij}(\gamma) - \overline{z}_{\bullet,j}(\gamma)\right| \lesssim \left|y_{ij}(\gamma) - \overline{y}_{\bullet,j}(\gamma)\right| + \left|\overline{y}_{i,\bullet}(\gamma) - \dbar{y}_\gamma\right|.
\end{align*}

Also, note that by the triangle inequality for the operator norm,
\[
\norm{z_{ij}(\gamma) - \dbar{z}_\gamma}_{\text{op}} \lesssim \norm{y_{ij}(\gamma) - \dbar{y}_\gamma}_{\text{op}} + \norm{\overline{y}_{i,\bullet}(\gamma) - \dbar{y}_\gamma}_{\text{op}} + \norm{\overline{y}_{\bullet,j}(\gamma) - \dbar{y}_\gamma}_{\text{op}}.
\]
Further, $\norm{\overline{y}_{i,\bullet}(\gamma) - \dbar{y}_\gamma}_{\text{op}} \lesssim I \max_i \left|\overline{y}_{i,\bullet}(\gamma) - \dbar{y}_\gamma\right|$ and $\norm{\overline{y}_{\bullet,j}(\gamma) - \dbar{y}_\gamma}_{\text{op}} \lesssim J \max_j \left|\overline{y}_{\bullet,j}(\gamma) - \dbar{y}_\gamma\right|$.
Using these relations and the given assumption, it is not difficult to check \eqref{eq:CLT_assumptions} for the pseudo-outcomes  $\sum_\gamma c_\gamma z_{ij}(\gamma)$. Thus \cref{cor:simplified_clt} applies.
\end{proof}

Another condition could just bound the residuals of the pseudo-outcomes directly in terms of the variance. Our theorem generalizes the MRD of \cite{masoero2024multiple} in some ways. 
\begin{corollary}
Suppose that $y_{ij}(\gamma)$ are all bounded by $C$, for $\gamma \in \set{\tr,\ib,\is,\cc}$ and suppose $I^2\Var(\hat{\tau}_{\mathbf{c}}) \rightarrow \infty.$ Then
\[
\frac{\hat{\tau}_{\mathbf{c}} - \tau_{\mathbf{c}}}{\Var(\hat{\tau}_{\mathbf{c}})^{1/2}} \Rightarrow \textsf{N}(0,1).
\]
\end{corollary}
\begin{proof}
By the argument of Corollary \ref{cor:clt_mrd_estimators}, we may write
\[
\hat{\tau}_{\mathbf{c}} = \frac{1}{I_T J_T} \sum_{i,j} W_i^\tB W_j^\tS z'_{ij}  
\]
for some quantities $z'_{ij}$ stated explicitly in \eqref{eq:zij_tr}-\eqref{eq:zij_cc}. By the boundedness assumption, $|z'_{ij}| \leq C'$. Further, by the inequality $\norm{M}_{\operatorname{op}} \lesssim \sqrt{IJ} \max_{ij} |M_{ij}|$ for any $I \times J$ matrix $M$, we see that $\norm{z_{ij}' - \dbar{z}'}_{\operatorname{op}} \leq C'I$. Then by assumption $I^{-2}\norm{z_{ij}' - \dbar{z}'}_{\operatorname{op}} = o_p(\Var(\hat{\tau}_{\mathbf{c}})^{1/2})$. Thus, \cref{cor:simplified_clt} implies the desired weak convergence.
\end{proof}

\begin{example}[IID normal superpopulation]
\label{ex:iid_normal_superpop}
Consider potential outcome matrices which are a single realization from the data generating process
\begin{equation}
\label{eq:example_iid_normal}
y_{ij}(\gamma) \stackrel{i.i.d.}{\sim} \sf{N}(\mu_\gamma,1)
\end{equation}
for some constants $\mu_\gamma$. Then $\nu^{\tB}(y_{ij}(\gamma)) = \frac{1}{I} \sum_{i=1}^I (\bar{y}_{i,\bullet}(\gamma) - \dbar{y}_\gamma)^2 = \frac{1}{I} \sum_{i=1}^I O(1/\sqrt{J})^2 = O(1/J)$. On the other hand, $\nu^{\tBS}(y_{ij}(\gamma)) = O(1).$ Therefore, $\Var(\edavg{y}_\gamma) = \Theta(1/IJ).$ Similarly, one can also show that $\Var(\hat{\tau}_c) = \Theta(1/IJ)$. This does not fit into the framework of the CLT of \cite{masoero2024multiple}, since the potential outcomes are unbounded and the variance decays too quickly. In contrast, our results can establish asymptotic normality in this case. \\

Using \cref{cor:clt_mrd_estimators} it suffices to check \eqref{eq:CLT_assumptions_MRD_estimands}. In particular, for the data generating process in \eqref{eq:example_iid_normal}, the following estimates hold with high probability:
\begin{align*}
    & \max_{i=1}^I |\overline{y}_{i,\bullet}(\gamma)- \dbar{y}_{\gamma}| = O(\sqrt{J^{-1}\log I}) \\
    & \max_{j=1}^J |\overline{y}_{\bullet,j}(\gamma)- \dbar{y}_{\gamma}| = O(\sqrt{I^{-1}\log J}) \\
    & \max_{i,j} |y_{ij}(\gamma) - \overline{y}_{\bullet,j}(\gamma)| = O(\sqrt{\log IJ}) \\
    & \norm{y_{ij}(\gamma) - \dbar{y}_\gamma}_{\operatorname{op}} = O(\sqrt{I}).
\end{align*}
The last claim is a standard fact; see e.g. \cite{vershynin2018high}. Thus \eqref{eq:CLT_assumptions_MRD_estimands} is satisfied and the CLT holds. 
\end{example}

\begin{example}[Sparse Uniform Subgraph]
\label{ex:sparse_uniform_subgraph}
Suppose that $y_{ij}(\gamma) \sim X_iY_jZ_{ij}$ where $X_i, Y_j \overset{\mathrm{iid}}\sim \mathrm{Ber}(1/2)$ and $Z_{ij}\overset{\mathrm{iid}}\sim \mathrm{Ber}(4\mu)$ for $\frac{1}{2} \gg \mu \gg (I \wedge J)^{-1}$. In the limit as $\mu \downarrow 0$, we can show that up to logarithmic factors,
\[\Var\{\edavg{y}_\gamma|X,Y,Z\} \asymp \frac{\mu^2}{I} + \frac{\mu^2}{J} + \frac{\mu}{IJ}.\]
Meanwhile, with high probability,
\begin{align*}
    & \max_{i=1}^I |\bar{y}_{i,\bullet}(\gamma)- \dbar{y}_{\gamma}| = O(\sqrt{\mu J^{-1}}\log I) \\
    & \max_{j=1}^J |\bar{y}_{\bullet,j}(\gamma)- \dbar{y}_{\gamma}| = O(\sqrt{\mu I^{-1}}\log J) \\
    & \max_{i,j} |y_{ij}(\gamma) - \bar{y}_{\bullet,j}(\gamma)| =  O(1) \\
    & \norm{y_{ij}(\gamma) - \dbar{y}_\gamma}_{\operatorname{op}} = O\left( [(I+J)\mu]^{\frac{1}{2}}+(IJ\mu)^{\frac{1}{4}}\right).
\end{align*}
The first two estimates follow from Bernstein's inequality for bounded random variables, while the last follows from applying Latala's theorem conditional upon $(X,Y)$. Thus Eq.~\eqref{eq:CLT_assumptions_MRD_estimands} is satisfied and the CLT holds whenever $\mu \gg (I \wedge J)^{-1}$. This is a major improvement over \cite{masoero2024multiple}, which requires $\mu \gg (I \wedge J)^{-1/2}$. 
\end{example}

\begin{example}[Sparse Heavy Tailed]
\label{ex:sparse_heavy_tailed}
We consider another example different from those presented in the main text. Suppose that $y_{ij}(\gamma) \sim X_iY_jZ_{ij}$ where $X_i, Y_j \overset{\mathrm{iid}}\sim \mathrm{Ber}(\sqrt{2\mu})$ and $Z_{ij}\overset{\mathrm{iid}}\sim \mathrm{Ber}(1/2)$ for $\frac{1}{2} \gg \mu \gg (I \wedge J)^{-1}$. 
In the limit as $\mu \downarrow 0$, we can show that up to logarithmic factors,
\[\Var\{\edavg{y}_\gamma|X,Y,Z\} \asymp \frac{\mu}{I} + \frac{\mu}{J} + \frac{\mu}{IJ}.\]
To see this, note that with very high probability $I^{-1}\sum_i X_i, J^{-1} \sum_j Y_j$ are of order $ \sqrt\mu$ and $\dbar{y}$ is of order $\mu$.
Meanwhile, with high probability,
\begin{align*}
    & \max_{i=1}^I |\overline{y}_{i,\bullet}(\gamma)- \dbar{y}_{\gamma}| = O(1) \\
    & \max_{j=1}^J |\overline{y}_{\bullet,j}(\gamma)- \dbar{y}_{\gamma}| = O(1) \\
    & \max_{i,j} |y_{ij}(\gamma) - \overline{y}_{\bullet,j}(\gamma)| =  O(1) \\
    & \norm{y_{ij}(\gamma) - \dbar{y}_\gamma}_{\operatorname{op}} = O\left([\sqrt{I} + \sqrt{J}]\mu^{\frac{1}{4}}\right).
\end{align*} 
Thus, once again, convergence holds when $\mu \gg (I \wedge J)^{-1}$.
\end{example}

\subsection{Auxiliary Results for CLT}

\begin{lemma}[Tensorization of Wasserstein Distance]
\label{lemma:wasserstein_tensorization}
Consider the mixture distributions $\sum_i \pi_i F_i, \sum_i \pi_i G_i$ for $\sum_i \pi_i = 1$ and $\pi_i \geq 0.$ $(F_i)_i, (G_i)_i$ are distributions. Then
\[
d_W\left(\sum_i \pi_i F_i, \sum_i \pi_i G_i \right) \leq \sum_i \pi_i d_W\left(F_i,G_i\right).
\]
\end{lemma}
\begin{proof}
Let $(X_i,Y_i)$ be couplings between $F_i,G_i$ for all $i$. Let $P$ be a random variable that equals $i$, for $i=1,\dots,n$, with probability $\pi_i$, independent of $(X_i,Y_i)$. Then $(X_{P}, Y_P)$ is a coupling between $\sum_i \pi_i F_i, \sum_i \pi_i G_i$. Thus, $\E\left[ |X_P - Y_P| \right] \leq \sum_i \pi_i \E[|X_i - Y_i|]$, which implies that 
\[
d_W\left(\sum_i \pi_i F_i, \sum_i \pi_i G_i \right) \leq \sum_i \pi_i \E[|X_i - Y_i|].
\]
Take the infimum over the right hand side to conclude the desired property.
\end{proof}

\begin{lemma}
\label{lemma:subexponential} 
Consider a random variable $X$ and set $\mu := \E X$. Suppose that $X - \mu$ is subgaussian with parameter $\sigma$. Then
\[
\E \left| X^2 - \E[X^2]\right| \lesssim \sigma^2 + \sigma \left|\mu\right|.
\]
\end{lemma}
\begin{proof}
By properties of sub-Gaussian random variables, $\E[(X - \mu)^2] \lesssim \sigma^2$ (see e.g. \cite{vershynin2018high} Proposition 2.5.2). Note the expansion $X^2 - \E\left[X^2\right] = (X-\mu)^2 - \Var(X) + 2\mu (X-\mu)$. Taking absolute values and expectations, we obtain $\E \left| X^2 - \E[X^2]\right| \leq 2\Var(X) + 2|\mu| \E \left|X - \mu\right|$. Using subgaussianity again, we arrive at $\E \left| X^2 - \E[X^2]\right| \lesssim \sigma^2 + \sigma \left| \mu\right|$ as desired.
\end{proof}

\begin{lemma}
\label{lemma:wasserstein_convolution}
Consider independent random variables $X,Y$ with laws $\cl{L}_X, \cl{L}_Y$. Let $\cl{G}$ be another distribution. Then
\[
d_W\left( \cl{L}_{X} * \cl{L}_{Y}, \cl{G} * \cl{L}_Y \right) \leq d_W\left( \cl{L}_X, \cl{G} \right).
\]
\end{lemma}
\begin{proof}
Let $(X,X')$ be a coupling of $\cl{L}_X$ and $\cl{G}.$ Draw $Y \sim \cl{L}_Y$ independent of $X,X'$. Then $\left( X+Y, X' + Y \right)$ is a coupling of $\cl{L}_{X} * \cl{L}_{Y}, \cl{G} * \cl{L}_Y$. By the coupling definition of Wasserstein distance, $d_W\left( \cl{L}_{X} * \cl{L}_{Y}, \cl{G} * \cl{L}_Y \right) \leq \E|X - X'|$. Taking the infimum of all couplings between $\cl{L}_X,G$, yields $d_W\left( \cl{L}_{X} * \cl{L}_{Y}, \cl{G} * \cl{L}_Y \right) \leq d_W\left( \cl{L}_X, \cl{G} \right)$.
\end{proof}

\subsection{Results for Conservative Variance Estimation}
\label{sec:cons_variance_estimation}

In this section, we will prove consistency of the conservative variance estimator.
\conservativevariance*

\begin{proof}
We claim that it suffices to show that $\hat{\Sigma}_\gamma = (1+o_p(1))\Var(\edavg{y}_\gamma)$. Recall the elementary facts that (i) $\sqrt{1+o_p(1)} = 1+o_p(1)$ and (ii) for any nonnegative random variables $X,Y,$ $X(1+o_p(1)) + Y(1+o_p(1)) = (X+Y) (1+o_p(1)).$ Furthermore, if $X = (1+o_p(1))Y$ for $Y$ nonnegative, then the continuous mapping theorem implies that $X^+ = (1 + o_p(1))Y.$ Assembling the pieces, 
\begin{align*}
     \left(\sum_{\gamma \in \Gamma} |c_\gamma| \sqrt{\hat{\Sigma}_\gamma^+} \right)^2 & = \left(\sum_{\gamma \in \Gamma} |c_\gamma| \sqrt{(1+o_p(1))\Var(\edavg{y}_\gamma)} \right)^2 \\
     & = \left(\sum_{\gamma \in \Gamma} (1+o_p(1))|c_\gamma| \sqrt{\Var(\edavg{y}_\gamma)} \right)^2 \\
     & = \left(\sum_{\gamma \in \Gamma} |c_\gamma| \sqrt{\Var(\edavg{y}_\gamma)} \right)^2 (1+o_p(1)),
\end{align*}
which would imply $\frac{\hat{V}_{\mathbf{c}}}{V_\mathbf{c}}  \xrightarrow{p} 1.$ Therefore in the rest of this proof, we focus on showing that $\hat{\Sigma}_\gamma = (1+o_p(1))\Var(\edavg{y}_\gamma)$. Recalling that $\E[\hat{\Sigma}_\gamma] = \Var(\edavg{y}_\gamma)$ we need only show that 
\[
\frac{\Var(\hat{\Sigma}_\gamma)}{\Var(\edavg{y}_\gamma)^2} \rightarrow 0.
\]

We will focus on each term in the expression $\hat{\Sigma}_\gamma$. Recalling the definition of $\alpha_\gamma^\tB$ and $\alpha_\gamma^\tS$ defined in \eqref{eq:alpha_gamma}. Clearly, $1 - \alpha^\tB_\gamma - \alpha^\tS_\gamma + \alpha^\tB_\gamma\alpha^\tS_\gamma, 1 - \alpha^\tB_\gamma, 1 - \alpha^\tS_\gamma$ are all bounded away from zero, it suffices to show that 
\begin{align}
& \frac{\Var(\alpha^\tB \hat{\nu}_\gamma^\tB)}{\Var(\edavg{y}_\gamma)^2} \rightarrow 0 \\
& \frac{\Var(\alpha^\tS \hat{\nu}_\gamma^\tS)}{\Var(\edavg{y}_\gamma)^2} \rightarrow 0 \\
& \frac{\Var(\alpha^\tB \alpha^\tS \hat{\nu}_\gamma^\tBS)}{\Var(\edavg{y}_\gamma)^2} \rightarrow 0 \\
& \frac{1}{I^4}\frac{\Var\left(\frac{1}{I_\gamma J_\gamma} \sum_{(i,j) \in \gamma} \left( y_{ij}(\gamma) - \hat{\overline{y}}_{\bullet,j}(\gamma) \right)^2 \right)}{\Var(\edavg{y}_\gamma)^2} \rightarrow 0
\label{eq:cons_var_goestozero_IV}
\end{align}
Showing the first three asymptotic statements is done in \cref{prop:direct_effect_single_beta_prob_convergence} and  \cref{prop:sigma_buyer_seller_term_estimator_converges}. For \eqref{eq:cons_var_goestozero_IV}, elementary algebra shows that 
\[
\frac{1}{I_\gamma J_\gamma} \sum_{i,j \in \gamma}  \left( y_{ij}(\gamma) - \hat{\overline{y}}_{i,\bullet}(\gamma) \right)^2 = \frac{1}{I_\gamma J_\gamma} \sum_{i,j \in \gamma}  \left( y_{ij}(\gamma) - \dbar{y}_\gamma \right)^2 - \frac{1}{I_\gamma} \sum_{i \in \gamma}  \left( \hat{\overline{y}}_{i,\bullet}(\gamma) -\dbar{y}_\gamma \right)^2.
\]
\cref{lemma:empirical_buyer_variance_small} and \cref{lemma:square_of_group_average_convergence} show that the variance of both of these terms, divided by $I^4$, is $o(\Var(\edavg{y})^2)$. The last term in \eqref{eq:conservative_var_est} is symmetric with the fourth term, which is handled by \eqref{eq:cons_var_goestozero_IV}. 

\end{proof}

As the corresponding proofs are similar, throughout the section we will assume $\gamma = \tr$ and also omit the dependence, where clear.

\begin{proposition}
\label{prop:direct_effect_single_beta_prob_convergence}
Under the assumptions of \cref{thm:consistent_variance_estimator},
\[
\frac{\Var(\alpha^\tB \hat{\nu}^\tB)}{\Var(\edavg{y})^2} \rightarrow 0, \quad \frac{\Var(\alpha^\tS \hat{\nu}^\tS)}{\Var(\edavg{y})^2} \rightarrow 0.
\]
\end{proposition}
\begin{proof}
By symmetry we will focus on the buyer term. Note first that Lemma A.12 of \cite{masoero2024multiple}, or alternatively simple algebra, shows that
\[
\hat{\nu}^\tB_\gamma = \frac{1}{I_\gamma} \sum_{i \in \cl{I}_\gamma} \left( \hat{\overline{y}}_{i,\bullet}(\gamma) - \dbar{y}_\gamma \right)^2 - \left(\hat{\dbar{y}}_\gamma - \dbar{y}_\gamma \right)^2.
\]
Using \cref{lemma:variance_of_sum}, it suffices to bound the variance of the individual terms. This is accomplished in \cref{lemma:square_of_double_average_small variance,lemma:empirical_buyer_variance_small}.

\end{proof}

\begin{proposition}
\label{prop:sigma_buyer_seller_term_estimator_converges}
Under the assumptions of Theorem 6.2, we have
\[
\frac{\Var(\alpha^\tB \alpha^\tS \hat{\nu}^\tBS)}{\Var(\edavg{y})^2} \rightarrow 0
\]
\end{proposition}
\begin{proof}
Lemma A.15 of \cite{masoero2024multiple} shows
\[
\hat{\nu}^\tBS_\gamma = \frac{1}{I_\gamma J_\gamma} \sum_{i \in \cl{I}_\gamma, j \in \cl{J}_\gamma} \left( y_{ij}(\gamma) - \dbar{y}_\gamma \right)^2 - \frac{1}{I_\gamma} \sum_{i \in \cl{I}_\gamma} \left( \hat{\overline{y}}_{i,\bullet}(\gamma) - \dbar{y}_\gamma \right)^2 - \frac{1}{I_\gamma} \sum_{j \in \cl{J}_\gamma} \left( \hat{\overline{y}}_{\bullet, j}(\gamma) - \dbar{y}_\gamma \right)^2 +  \left(\hat{\dbar{y}}_\gamma - \dbar{y}_\gamma \right)^2.
\]
Using \cref{lemma:variance_of_sum}, it suffices to bound the variance of the individual terms. \cref{lemma:square_of_double_average_small variance,lemma:empirical_buyer_variance_small,lemma:square_of_group_average_convergence} in tandem show that each of the terms is $o(I^4 \Var(\edavg{y})^2)$ as desired.
\end{proof}

\subsection{Lemmas for Proofs}

\begin{lemma}
\label{lemma:square_of_double_average_small variance}
\[
\frac{1}{I^2} \frac{\Var \left[\left(\hat{\dbar{y}}_\gamma - \dbar{y}_\gamma \right)^2 \right]}{\Var(\edavg{y})^2} \rightarrow 0
\]
\end{lemma}
\begin{proof}
Bounding the variance by the second moment and using \cref{lemma:fourth_moment_bound_double_avg}, we conclude that 
\[
\Var \left[\left(\hat{\dbar{y}}_\gamma - \dbar{y}_\gamma \right)^2 \right] \leq C \Var\left(\edavg{y} \right)^2.
\]
Thus 
\[
\frac{\Var \left[\alpha^\tB\left(\hat{\dbar{y}}_\gamma - \dbar{y}_\gamma \right)^2 \right]}{\Var(\edavg{y})^2} \leq \frac{C}{I^2}.
\]

\end{proof}

\begin{lemma}
\label{lemma:empirical_buyer_variance_small}
\[
\frac{1}{I^2}\frac{\Var\left[ \frac{1}{I_\gamma} \sum_{i \in \cl{I}_\gamma} \left( \hat{\overline{y}}_{i,\bullet}(\gamma) - \dbar{y}_\gamma \right)^2\right]}{\Var(\edavg{y})^2} \rightarrow 0.
\]    
\end{lemma}
\begin{proof}
The ANOVA decomposition lets us compute this variance as 
\begin{equation}
\label{eq:lemma_buyer_term_cons_variance_anova}
\begin{split}
    & \Var\left[ \frac{1}{I} \sum_{i=1}^I \left( \hat{\overline{y}}_{i,\bullet}(\gamma) - \dbar{y}_\gamma \right)^2\right]. \\
    + & 
    \E \left[\frac{C}{I(I-1)}\sum_{i=1}^I \left[   \left( \hat{\overline{y}}_{i,\bullet}(\gamma) - \dbar{y}_\gamma \right)^2 - \frac{1}{I}\sum_{i'} \left( \hat{\overline{y}}_{i',\bullet}(\gamma) - \dbar{y}_\gamma \right)^2 \right]^2 \right]. 
\end{split}
\end{equation}

For the first term, apply \cref{lemma:vector_concentration_sampling_without_replacement} taking the vectors $u_j := \frac{1}{J_T}\left(d_{1j},\dots,d_{Ij} \right)^\top$. Then $\sum_j^J W_j^\tS u_j = \left(\eavg{d}_{i,\bullet} \right)_i^\top$. Then let
\[
D := \norm{\left( \eavg{d}_{i,\bullet}\right)_i} = \left(\sum_i \eavg{d}_{i,\bullet}^2 \right)^{1/2}.
\]
\cref{lemma:vector_concentration_sampling_without_replacement} gives that $D - \E D$ is a sub-Gaussian random variable with sub-Gaussian norm $\sigma_D := \frac{1}{J_T}\norm{(d_{ij})_{ij}}_{\text{op}}.$ Applying \cref{lemma:var_squared_and_central_moments} we find that 
\begin{align*}
& \Var\left[ \frac{1}{I} \sum_{i=1}^I \left( \hat{\overline{y}}_{i,\bullet} - \dbar{y} \right)^2\right]  = \frac{1}{I^2} \Var(D^2) \\
& \leq 
\frac{1}{I^2}\E\left[(D - \E D)^4 \right] + 4\frac{\E[D]}{I^2} \E\left[|D - \E D|^3 \right]+ 4 \frac{\E[D]^2}{I^2} \E\left[(D - \E D)^2 \right]
\end{align*}
Subgaussian bounds on the moments give us the upper bound, up to an absolute constant scaling factor,
\begin{equation}
\label{eq:lemma_buyer_term_var_estimator_helper}
\frac{1}{I^6}\norm{(d_{ij})}^4_{\text{op}} + \frac{1}{I^5}\E[D] \norm{(d_{ij})}^3_{\text{op}} + \frac{1}{I^4}\E[D]^2 \norm{(d_{ij})}^2_{\text{op}}
\end{equation}
We have$\frac{1}{I^4} \norm{(d_{ij})}^2_{F} \lesssim \Var(\edavg{y})$ and $\frac{1}{I^2}\E[D^2] \lesssim \Var(\edavg{y})$ by \cref{lemma:d_matrix_frobenius}. Using $\E[D] \leq \E[D^2]^{1/2}$ and $\norm{(d_{ij})}_{\text{op}} \leq \norm{(d_{ij})}_{F}$ then yields for the first term in \eqref{eq:lemma_buyer_term_var_estimator_helper}
\[
\frac{1}{I^6}\norm{(d_{ij})}^4_{\text{op}} \leq \frac{1}{I^6}\norm{(d_{ij})}^2_{F} \norm{(d_{ij})}^2_{\text{op}} \lesssim \frac{1}{I^2}\Var(\edavg{y}) \norm{(d_{ij})}_{\text{op}}^2
\]
For the second term in \eqref{eq:lemma_buyer_term_var_estimator_helper}, 
\begin{align*}
\frac{1}{I^5}\E[D] \norm{(d_{ij})}^3_{\text{op}} & \leq \frac{1}{I^2} \sqrt{ \frac{1}{I^2}\E[D^2]} \cdot \left(\frac{1}{I^2} \norm{(d_{ij})}_{F} \right) \norm{(d_{ij})}_{\text{op}}^2 \\
& \lesssim \frac{1}{I^2} \Var(\edavg{y}) \norm{(d_{ij})}_{\text{op}}^2
\end{align*}
For the third term in \eqref{eq:lemma_buyer_term_var_estimator_helper}, 
\begin{align*}
\frac{1}{I^4}\E[D]^2 \norm{(d_{ij})}^2_{\text{op}} & \leq \frac{1}{I^2} \left(\frac{1}{I^2} \E[D^2] \right) \norm{(d_{ij})}_{\text{op}}^2  \\
& \lesssim \frac{1}{I^2} \Var(\edavg{y}) \norm{(d_{ij})}_{\text{op}}^2.
\end{align*}
Summarizing the established bounds, we see that
\[
\frac{1}{I^2} \frac{\Var\left[ \frac{1}{I} \sum_{i=1}^I \left( \hat{\overline{y}}_{i,\bullet} - \dbar{y} \right)^2\right]}{\Var(\edavg{y})^2} \leq \frac{C}{I^4} \frac{\norm{(d_{ij})}_{\text{op}}^2}{\Var(\edavg{y})},
\]
which goes to zero by assumption.\\

For the second term of \eqref{eq:lemma_buyer_term_cons_variance_anova}, upper bound it by
\[
\frac{C}{I^2}\sum_{i=1}^I \E \left[   \left( \hat{\overline{y}}_{i,\bullet}- \dbar{y} \right)^4  \right]
\]
Rewrite $\left( \hat{\overline{y}}_{i,\bullet}- \dbar{y} \right) = \left( \hat{\overline{y}}_{i,\bullet} - \overline{y}_{i,\bullet}\right) + \left(\overline{y}_{i,\bullet} - \dbar{y} \right)$ and expand, yielding 
\begin{align}
\frac{1}{I^2}\sum_{i=1}^I \E \left[   \left( \hat{\overline{y}}_{i,\bullet}- \overline{y}_{i,\bullet} \right)^4  \right] & \lesssim \frac{1}{I^2}\sum_{i=1}^I \E \left[   \left( \hat{\overline{y}}_{i,\bullet}- \dbar{y} \right)^4  \right] \label{eq:anova_case2_term1} \\
+ & \frac{1}{I^2}\sum_{i=1}^I \E \left[   \left( \hat{\overline{y}}_{i,\bullet}- \overline{y}_{i,\bullet} \right)^3  \right] \left(\overline{y}_{i,\bullet} - \dbar{y} \right) \label{eq:anova_case2_term2}\\
+ & \frac{1}{I^2}\sum_{i=1}^I \E \left[   \left( \hat{\overline{y}}_{i,\bullet}- \overline{y}_{i,\bullet}\right)^2  \right] \left(\overline{y}_{i,\bullet} - \dbar{y} \right)^2 \label{eq:anova_case2_term3}\\
+ & \frac{1}{I^2}\sum_{i=1}^I \left(\overline{y}_{i,\bullet} - \dbar{y} \right)^4 \label{eq:anova_case2_term4}
\end{align}
We bound each of these terms separately.

\begin{enumerate}
    \item For $\eqref{eq:anova_case2_term1}$, \cref{lemma:sampling_subgaussianity} gives us the upper bound
    \begin{align*}
    \frac{1}{I^2} \sum_{i=1}^I \Var(\eavg{y}_{i,\bullet})^2 & \lesssim  \frac{1}{I^2} \sum_{i=1}^I \left[ \frac{1}{J^2} \sum_{j=1}^J (y_{ij} - \overline{y}_{i,\bullet})^2 \right]^2 \\
    & \lesssim \frac{1}{I^5} \sum_{i,j} (y_{ij} - \overline{y}_{i,\bullet})^4 \\
    & \lesssim \frac{1}{I^5} \sum_{i,j} (y_{ij} - \overline{y}_{i,\bullet})^2 \cdot \max |\overline{y}_{i,\bullet} - y_{ij}|^2 \\
    & \lesssim \frac{1}{I}  \Var(\edavg{y}) \cdot \max |\overline{y}_{i,\bullet} - y_{ij}|^2  
    \end{align*}
    \item For $\eqref{eq:anova_case2_term2}$, applying the same subgaussian moment bound and Jensen's inequality yields
    \begin{align*}
    & \frac{1}{I^2}\sum_{i=1}^I \E \left[   \left( \hat{\overline{y}}_{i,\bullet}- \overline{y}_{i,\bullet} \right)^3  \right] \left(\overline{y}_{i,\bullet} - \dbar{y} \right) \\
    \lesssim & \max_i |\overline{y}_{i,\bullet} - \dbar{y}| \frac{1}{I^2}\sum_{i=1}^I \left[ \frac{1}{J^2} \sum_{j=1}^J (y_{ij} - \overline{y}_{i,\bullet})^2 \right]^{3/2} \\
    \lesssim & \max_i |\overline{y}_{i,\bullet} - \dbar{y}| \frac{1}{I^2 J^{3/2} J} \sum_{i,j} |y_{ij} - \overline{y}_{i,\bullet}|^3 \\
    \lesssim & \frac{1}{I^{1/2}} \max_i |\overline{y}_{i,\bullet} - \dbar{y}| \max_{ij} |y_{ij} - \overline{y}_{i,\bullet}| \Var[\edavg{y}].
    \end{align*}  
    \item For $\eqref{eq:anova_case2_term3}$ apply the same logic:
    \begin{align*}
    & \frac{1}{I^2}\sum_{i=1}^I \E \left[   \left( \hat{\overline{y}}_{i,\bullet}- \overline{y}_{i,\bullet} \right)^2  \right] \left(\overline{y}_{i,\bullet} - \dbar{y} \right)^2 \\
    \lesssim & \max_i |\overline{y}_{i,\bullet} - \dbar{y}|^2 \frac{1}{I^2}\sum_{i=1}^I \left[ \frac{1}{J^2} \sum_{j=1}^J (y_{ij} - \overline{y}_{i,\bullet})^2 \right] \\
    \lesssim & \max_i |\overline{y}_{i,\bullet} - \dbar{y}|^2 \Var(\edavg{y})
    \end{align*}  
    \item For $\eqref{eq:anova_case2_term4}$, we can easily upper bound the term by $\frac{1}{I}\nu^\tB(y_{ij})  \max_i |\overline{y}_{i,\bullet} - \dbar{y}|^2 \leq \Var(\edavg{y}) \max_i |\overline{y}_{i,\bullet} - \dbar{y}|^2$.
\end{enumerate}

Combining the individual bounds gives the upper bound
\begin{align*}
    &\frac{1}{I^2} \frac{I^{-2}\sum_{i=1}^I \E \left[   \left( \hat{\overline{y}}_{i,\bullet}- \dbar{y} \right)^4  \right]}{\Var(\edavg{y})^2} \\
    & \lesssim \frac{1}{I^3} \frac{\max_{ij}|y_{ij} - \overline{y}_{i,\bullet}|^2}{\Var(\edavg{y})} + \frac{1}{I^{5/2}} \frac{\max_i |\overline{y}_{i,\bullet} - \dbar{y}| \max_{ij} |y_{ij} - \overline{y}_{i,\bullet}|}{ \Var(\edavg{y})}  + \frac{1}{I^2} \frac{\max_i |\overline{y}_{i,\bullet} - \dbar{y}|^2}{ \Var(\edavg{y})}
\end{align*}
Under the given assumptions, the expression goes to zero.

\end{proof}

\begin{lemma}
\label{lemma:square_of_group_average_convergence}
\[
\frac{1}{I^2 J^2} \frac{\Var\left[\frac{1}{I_\gamma J_\gamma} \sum_{i \in \cl{I}_\gamma, j \in \cl{J}_\gamma} \left( y_{ij}(\gamma) - \dbar{y}_\gamma \right)^2\right]}{\Var(\edavg{y})^2} \xrightarrow{p} 0
\]
\end{lemma}
\begin{proof}
Fix any $\gamma \in \Gamma$ and let $d_{ij} = y_{ij}(\gamma) - \dbar{y}_\gamma$. By the variance formula in \cite{masoero2024multiple},
\[
\Var \left[ \frac{1}{I_\gamma J_\gamma} \sum_{i \in \cl{I}_\gamma, j \in \cl{J}_\gamma} d_{ij}(\gamma)^2 \right] \lesssim \frac{1}{I}\nu^\tB_\gamma(\set{d_{ij}^2}) + \frac{1}{J}\nu^\tS_\gamma(\set{d_{ij}^2}) + \frac{1}{IJ}\nu^\tBS_\gamma(\set{d_{ij}^2}) 
\]
Thus, it suffices to bound $\nu^\tB(d_{ij}^2)$ in terms of $\nu^\tB(y_{ij})$. Observe that 
\begin{align*}
\nu^\tB(d_{ij}^2) \lesssim \frac{1}{I J} \sum_{i=1}^I \sum_{j=1}^J d_{ij}^4 & \leq \max_{ij} d_{ij}^2\frac{1}{IJ} \sum_{ij} d_{ij}^2 \\
& \lesssim \max_{ij} d_{ij}^2 \cdot \left[\nu^\tB(y_{ij}) + \nu^\tS(y_{ij}) + \nu^\tBS(y_{ij}) \right].
\end{align*}
The same bounds hold for $\nu^\tS(d_{ij}^2), \nu^\tBS(d_{ij}^2)$. Thus,
\[
\frac{1}{I^2 J^2} \frac{\Var\left[\frac{1}{I_\gamma J_\gamma} \sum_{i \in \cl{I}_\gamma, j \in \cl{J}_\gamma} \left( y_{ij}(\gamma) - \dbar{y}_\gamma \right)^2\right]}{\Var(\edavg{y})^2} \leq \frac{1}{I^3} \frac{\max_{ij}|d_{ij}|^2}{\Var(\edavg{y})}.
\]
By assumption, the right hand side goes to zero.
\end{proof}

\begin{lemma}
\label{lemma:var_squared_and_central_moments}
Let $X$ be a random variable with finite fourth moment. Set $\mu := \E X$. Then
\begin{equation}
\label{eq:varx2-identity}
\Var(X^2) \leq \E\left[(X - \mu)^4 \right] + 4|\mu| \E\left[|X - \mu|^3 \right]+ 4 \mu^2 \E\left[(X - \mu)^2 \right].
\end{equation}
\end{lemma}
\begin{proof}Set $Y := X-\mu$. Then $X^2 = (\mu+Y)^2 = \mu^2 + 2\mu Y + Y^2$ and hence
\[
\Var(X^2) = \Var(2\mu Y + Y^2)
           = \Var(Y^2) + 4\mu^2 \Var(Y) + 4\mu\,\Cov(Y,Y^2),
\]
Since $\Var(Y)=\E[Y^2]$, $\Var(Y^2)=\E[Y^4]-\E[Y^2]^2$, and $\Cov(Y,Y^2) = \E[Y^3]$:
\[
\Var(X^2) = \E[(X-\mu)^4] - \E[(X-\mu)^2]^2 + 4\mu\,\E[(X-\mu)^3] + 4\mu^2\,\E[(X-\mu)^2].
\]
Using triangle inequality and removing the negative term gives the result.
\end{proof}

\begin{lemma}[ANOVA decomposition for variance]
\label{lemma:anova_variance_decomp}
\begin{align*}
& \frac{1}{IJ}\sum_{ij} (y_{ij} - \overline{y}_{i,\bullet})^2 = \nu^{\tBS} + \nu^{\tS} \\
& \frac{1}{IJ}\sum_{ij} (y_{ij} - \overline{y}_{\bullet,j})^2 = \nu^{\tBS} + \nu^{\tB} \\
& \frac{1}{IJ}\sum_{ij} (y_{ij} - \dbar{y})^2 = \nu^{\tBS} + \nu^{\tB} + \nu^{\tS}
\end{align*}
\end{lemma}

\begin{lemma}
\label{lemma:d_matrix_frobenius}
Recall the notation $d_{ij} := y_{ij} - \dbar{y}$. Let $D := \norm{\left( \eavg{d}_{i,\bullet}\right)_i} = \left(\sum_i \eavg{d}_{i,\bullet}^2 \right)^{1/2}.$ We have $\norm{d_{ij}}_F^2  \lesssim I^4\Var(\edavg{y})$ and $\E[D^2] \lesssim I^2\Var(\edavg{y}).$
\end{lemma}
\begin{proof}
By \cref{lemma:anova_variance_decomp}, $\norm{d_{ij}}_F^2 = \sum_{i,j} (y_{ij} - \dbar{y})^2 = IJ(\nu^{\tBS} + \nu^{\tB} + \nu^{\tS})$. By \cref{prop:generic_variance_formula}, $\Var(\edavg{y}) \lesssim \frac{1}{I}\nu^\tB + \frac{1}{J}\nu^\tS + \frac{1}{IJ}\nu^\tBS$. This shows the first bound. For the second bound, the sampling without replacement variance formula gives
\begin{align*}
\E[D^2] & \lesssim \frac{1}{J^2} \sum_{ij} (d_{ij} - \overline{d}_{i,\bullet})^2 + \sum_i \overline{d}_{i,\bullet}^2 \\
& = \frac{1}{J^2} \sum_{ij} (y_{ij} - \overline{y}_{i,\bullet})^2 + \sum_{i=1}^I (\overline{y}_{i,\bullet} - \dbar{y})^2 \\
& \lesssim \nu^\tBS + \nu^\tS + I \nu^\tB && \text{(\cref{lemma:anova_variance_decomp})}\\
& \lesssim I^2 \Var(\edavg{y}).
\end{align*}
\end{proof}

\begin{lemma}[Trivial Variance Bound]
\label{lemma:variance_of_sum}
For any two random variables $X,Y$, $\Var(X + Y) \leq 2\Var(X) + 2\Var(Y)$. More generally, let $X_1,\dots,X_n$ be any collection of random variables. Then
\[
\Var\left(\frac{1}{n}\sum_{i=1}^n X_i \right) \leq \frac{1}{n} \sum_{i=1}^n \Var(X_i).
\]
\end{lemma}
\begin{proof}
Cauchy-Schwarz.     
\end{proof}

\section{Sub-Gaussianity \& Moment Bounds}
\label{sec:subgaussianity_sampling_without_replacement}

Results in Appendix A of \cite{lei2021regression} give us the following.

\begin{lemma}[Sampling without replacement has subgaussian tails] For scalars $u_1, \ldots, u_I$, 
\label{lemma:sampling_subgaussianity}
\[
\Prob\left( \left|\sum_{i=1}^I W_i^\tB u_i - I_T \overline{u} \right| > t \right) \leq 2\exp(-\frac{t^2}{8 I \sigma^2_u})
\]
where $\overline{u} = \frac{1}{I} \sum_i^I u_i$ and $\sigma^2_u := \frac{1}{I}\sum_{i=1}^I  u_i^2 - \left(\frac{1}{I} \sum_i u_i \right)^2$. Thus,
\[
\Prob\left( \left|\widehat{\overline{u}} - \overline{u} \right| > t \right) \leq 2\exp(-\frac{t^2}{8 I\sigma^2_u/I_T^2}) 
\]
where $\widehat{\overline{u}} = \frac{1}{I_T} \sum_i W_i^\tB u_i$.
\end{lemma}
By standard formulas, $\Var(\widehat{\overline{u}}) = \frac{(I - I_T)}{I_T (I-1)}\sigma_u^2$. Under the asymptotic scalings for $I,I_T$ in this paper, we have $I\sigma^2_u/I_T^2 \leq C\Var(\widehat{\overline{u}})$ for some absolute constant $C$. Therefore $\widehat{\overline{u}}$ is also subgaussian with parameter $O(\Var(\widehat{\overline{u}}))$.

\begin{lemma}[Vector sampling without replacement has subgaussian tails]
\label{lemma:vector_concentration_sampling_without_replacement}
Let $u_1,\dots,u_I$ now be a finite population of vectors. Then
\[
\Prob \left(\norm{\sum_i W_i^\tB u_i}_2 \geq \E \norm{\sum_i W_i^\tB u_i}_2 + t \right) \leq \exp\left(-\frac{t^2}{8 \norm{U}_{\operatorname{op}}^2} \right)
\]
where the $i$th row of $U$ is $u_i^\top.$
\end{lemma}
See the Proof of Lemma A.3 in \cite{lei2021regression}. The other side also holds:
\[
\Prob \left(\norm{\frac{1}{I_T}\sum_i W_i^\tB u_i}_2 \leq \E \norm{\frac{1}{I_T}\sum_i W_i^\tB u_i}_2 - t \right) \leq \exp\left(-\frac{I_T^2 t^2}{8 \norm{U}_{\operatorname{op}}^2} \right).
\]

Subgaussianity of sampling without replacement gives the following moment bounds.

\begin{lemma}
\label{lemma:moment_bound_sampling_wout_replacement}
For any natural number $k$ and scalars $u_1, \ldots, u_I$, such that $\overline{u} = 0.$
\[
\E\left[\left(\frac{1}{I_T} \sum_{i=1}^I W_i^\tB u_i \right)^{2k} \right] \leq C_k \sigma_u^{2k} 
\]
for some constant $C_k$ that only depends on $k$, where $\sigma_u^2$ is the finite population variance of the average,
\[
\sigma_u^2 = \left(\frac{1}{I_T} - \frac{1}{I}\right) \frac{1}{I-1} \sum_{i=1}^I \left( u_i - \overline{u}\right)^2.
\]
\end{lemma}

Recall the quantities
\begin{equation}
    \overline{\e}_\gamma^{\tB} = \frac{1}{I_\gamma}\sum_{i=1}^I W_i^{\tB} \delta_i^{\tB}(\gamma), \quad \overline{\e}_\gamma^{\tS} = \frac{1}{J_\gamma}\sum_{j=1}^J W_j^{\tS} \delta_j^{\tS}(\gamma), \quad 
    \dbar{\e}_\gamma^{\tBS} = \frac{1}{I_\gamma J_\gamma}\sum_{i,j} W_i^{\tB} W_j^{\tS}\delta_{ij}^{\tBS}(\gamma).
\end{equation}

\cref{lemma:hoeffding_decomposition} shows that these quantities are uncorrelated. 

\begin{lemma}[Even Moment Bound for Decentered Double Average]
\label{lemma:moment_bound_decentered_double_avg}
Let $k$ be a nonnegative integer. Then for each $\gamma \in \set{\tr,\ib,\is,\cc},$
\[
\E \left[(\dbar{\e}^{\tBS}_\gamma)^{2k}\right] \lesssim C_k \Var(\dbar{\e}^\tBS_\gamma)^k,
\]
for some constant that only depends on $k$. 

\end{lemma}
\begin{proof}
The proof is the same for each $\gamma$ so we will ignore dependence on $\gamma.$ Write $\dbar{\e}^{\tBS} = \frac{1}{I_T} \sum_{i=1}^I W_i^\tB \eavg{\delta}_{i,\bullet}$ where $\eavg{\delta}_{i,\bullet} = \frac{1}{J_T} \sum_j W_j^\tS \delta_{ij}^\tBS$. Conditioning first on $\tS$, notice that $\sum_i \eavg{\delta}_{i,\bullet} = 0$. Using the subgaussian moment bounds for the sample average $\dbar{\e}^{\tBS}$ from the previous lemma, we find that
\begin{align*}
\E[\dbar{\e}^{\tBS \ 2k}] & \lesssim \frac{C_k}{I^k} \E\left[ \left( \frac{1}{I-1} \sum_{i=1}^I \eavg{\delta}_{i,\bullet}^2\right)^k \right].
\end{align*}
Expanding out the $k$th power gives terms of the form $\E\left[ \eavg{\delta}_{i_1,\bullet}^2 \dots \eavg{\delta}_{i_k,\bullet}^2\right]$ for possibly non-unique indices $i_1,\dots,i_k$. Using the generalized Holder inequality and sub-Gaussianity of the sample average $\eavg{\delta}_{i_j,\bullet}$ to bound $\E[\eavg{\delta}_{i_j,\bullet}^{2k}]$ (using the previous lemma again),
\[
\E\left[ \eavg{\delta}_{i_1,\bullet}^2 \dots \eavg{\delta}_{i_k,\bullet}^2\right] \leq \prod_{j=1}^k \E\left[\eavg{\delta}_{i_j,\bullet}^{2k}\right]^{1/k} \lesssim C_k^k \prod_{j=1}^k \Var(\eavg{\delta}_{i_j,\bullet}).
\]
This implies 
\[
\E\left[(\dbar{\e}^{\tBS})^{2k}\right] \lesssim \frac{C_k}{I^k} \E \left[ \frac{1}{I-1} \sum_{i=1}^I \Var(\eavg{\delta}_{i,\bullet})\right]^k .
\]
Recall that $\Var(\eavg{\delta}_{i,\bullet}) \lesssim \frac{1}{J^2} \sum_j (\delta_{ij}^{\tBS})^2.$
Thus, $\E\left[(\dbar{\e}^{\tBS})^{ 2k}\right] \leq C_k \left(\frac{1}{IJ} \nu^{\tBS} \right)^k$ from which the conclusion follows.
\end{proof}

\begin{lemma}[4th Moment Bound for Double Average]
\label{lemma:fourth_moment_bound_double_avg}
\[
\E\left[\left(\edavg{y}- \dbar{y}\right)^4\right]  \leq C\Var(\edavg{y})^2,
\]
for some constant $C$.
\end{lemma}
\begin{proof}
Recalling the decomposition $\edavg{y} - \dbar{y} = \overline{\e}^{\tB} + \overline{\e}^\tS + \dbar{\e}^\tBS$
up to signs, the fourth moment is given by
\begin{align*}
\E\left[\left(\edavg{y}- \dbar{y}\right)^4\right] & = \sum_{\substack{a,b,c \geq 0 \\ a+b+c = 4}} \E\left[(\overline{\e}^{\tB})^{a}(\overline{\e}^{\tS})^{b}(\overline{\e}^{\tBS})^{c}  \right]  \\
& \leq \sum_{\substack{a,b,c \geq 0 \\ a+b+c = 4}} \left(\E\left[(\overline{\e}^\tB)^{4a}\right] \right)^{1/4} \left(\E\left[(\overline{\e}^{\tS})^{4b}\right] \right)^{1/4} \left(\E\left[(\overline{\e}^{\tBS})^{2c}\right] \right)^{1/2}  \\
& \leq \sum_{\substack{a,b,c \geq 0 \\ a+b+c = 4}} \left(\Var\left[\overline{\e}^{\tB}\right] \right)^{a/2} \left(\Var\left[\overline{\e}^{\tS}\right] \right)^{b/2} \left(\Var\left[\overline{\e}^{\tBS}\right] \right)^{c/2} \\
& = \left(\sqrt{\Var\left[\overline{\e}^{\tB}\right]} + \sqrt{\Var\left[\overline{\e}^{\tS}\right]} + \sqrt{\Var\left[\overline{\e}^{\tBS}\right]} \right)^4 \\
& \leq C\left(\Var\left[\overline{\e}^{\tB}\right] + \Var\left[\overline{\e}^{\tS}\right] + \Var\left[\overline{\e}^{\tBS}\right] \right)^2,
\end{align*}
where we have applied the moment bounds \cref{lemma:moment_bound_sampling_wout_replacement}, \cref{lemma:moment_bound_decentered_double_avg}, and Cauchy Schwarz. The last bound follows from Jensen's inequality.
\end{proof}

\section{Proofs for \cref{sec:asymptotic_theory}}
\label{sec:proofs_for_asymptotic_theory}

In this section we show the details on proving consistency of the optimal adjustment. As all our results are asymptotic in nature, we will replace all appearances of $I-1, J-1$ by $I,J$ respectively. For our asymptotic statements, we also interchange $I,J$ freely since $I/J \rightarrow \rho \in (0,\infty)$. Throughout this section, define $\hat{e}_c$ to be $\hat{u}_c$ replacing the potential outcomes $y_{ij}(\gamma)$ by the residuals $e_{ij}(\gamma) = y_{ij}(\gamma) - X_{ij}^\top\tilde{\beta}_c$, and let $\tilde{e}_c$ denote the vector $\tilde{u}_c$ of \eqref{eq:def_inner_product_beta_components} with $y_{ij}$ replaced by $e_{ij}$. Further, define $\hat{e}_\gamma^\theta, \tilde{e}_\gamma^\theta$ as the analogues of $\hat{u}_\gamma^\theta, \tilde{u}_\gamma^\theta$ in the same way. We first record some elementary facts that follow from the assumptions of \cref{sec:asymptotic_theory}.

\begin{lemma}
\label{lemma:betas_bounded}
Under \cref{assmp:finite_population_limits,assmp:non_vanishing_buyer_seller_variation}, $\tilde{\beta}_c = O(1)$.
\end{lemma}
\begin{proof}
Write $\tilde{\beta}_c = (I\tilde{Z}_c)^{-1} (I \tilde{u}_c).$ \cref{assmp:non_vanishing_buyer_seller_variation} shows $(I\tilde{Z}_c)^{-1} = O(1)$ entrywise while combining \cref{assmp:finite_population_limits} with the fact that each $a_{\gamma,\theta}$ is $O(1/I)$ shows that the entries of $I \tilde{u}_c$ are $O(1).$
\end{proof}

\begin{lemma}[Naive Variance Bounds]
\label{lemma:naive_variance_bounds}
Under \cref{assmp:finite_population_limits,assmp:non_vanishing_buyer_seller_variation}, $\Var(\hat{\tau}_c(\tilde \beta_c)) = O(1/I)$.
\end{lemma}
\begin{proof}
By inspecting the variance formulas in \cref{prop:generic_variance_formula,prop:generic_covariance_formula}, it suffices to prove that $\nu_\gamma^\theta(e_{ij}) = O(1)$ and $\xi_{\gamma,\gamma'}^\theta(e_{ij}) = O(1)$ for $\gamma,\gamma' \in \set{\tr,\ib,\is,\cc}$ and $\theta \in \set{\tB,\tS,\tBS}$. Observe the following expansions:
\begin{equation*}
\begin{split}
\nu^\theta_{\gamma}(e_{ij}) & = \beta_c^\top Z^\theta \beta_c - 2u^{\theta\top}_\gamma \beta_c + \nu^\theta_{\gamma} \\
\xi^\theta_{\gamma,\gamma'}(e_{ij}) & = \xi^\theta_{\gamma,\gamma'}.  
\end{split}
\end{equation*}
\cref{assmp:finite_population_limits} and \cref{lemma:betas_bounded} give $\nu^\theta_{\gamma}(e_{ij}) = O(1)$, while Cauchy-Schwarz shows $\xi^\theta_{\gamma,\gamma'} = O(1).$
\end{proof}

\begin{lemma}[Decomposition of $\beta_c$]
\label{lemma:beta_decomposition}
$\hat \beta_c - \tilde{\beta}_c = \hat{Z}_c^{-1} \hat{e}_c$. Moreover, $\tilde{Z}_c^{-1} \tilde{e}_c = 0.$
\end{lemma}
\begin{proof}
For the first claim, recall \eqref{eq:inner_prod_term_estimators} and the definitions $\hat Z_c^{-1} \hat u_c = \hat \beta_c$ and $\tilde Z_c^{-1}\tilde u_c = \tilde \beta_c$. Then,

\begin{align*}
\hat{e}_\gamma^{\tB} & = \frac{1}{I_\gamma}\sum_{i\in \mathcal{I}_\gamma} \eavg{X}_{i,\bullet}^{(d)}(\gamma)\eavg{e}_{i,\bullet}^{(d)}(\gamma) \\
& = \frac{1}{I_\gamma}\sum_{i\in \mathcal{I}_\gamma} \eavg{X}_{i,\bullet}^{(d)}(\gamma)\left(\eavg{y}_{i,\bullet}^{(d)}(\gamma) -  \eavg{X}_{i,\bullet}^{(d)}(\gamma)^\top \tilde\beta_c\right) \\
& = \hat{u}_\gamma^\tB - \hat Z_\gamma^\tB \tilde\beta_c.
\end{align*}
The same argument gives $\hat{e}_\gamma^{\tS} = \hat{u}_\gamma^\tS - \hat Z_\gamma^\tS \tilde\beta_c, \hat{e}_\gamma^{\tBS} = \hat{u}_\gamma^\tBS - \hat Z_\gamma^\tBS \tilde\beta_c$, from which it follows that $\hat{e}_c = \hat{u}_c - \hat{Z}_c \tilde\beta_c$. 
Left-multiplication by $\hat Z_c^{-1}$ gives the result. Next, recall \eqref{eq:inner_prod_terms}. Replacing $y_{ij}(\gamma)$ by $e_{ij}(\gamma)$ gives the following:
\begin{align*}
e_\gamma^{\tB} & = \frac{1}{I-1}\sum_{i=1}^I \bar{X}_{i,\bullet}^{(d)}\bar{e}_{i,\bullet}^{(d)}(\gamma) = \frac{1}{I-1}\sum_{i=1}^I \bar{X}_{i,\bullet}^{(d)}(\bar{y}_{i,\bullet}^{(d)}(\gamma) - \bar{X}_{i,\bullet}^{(d)\top} \tilde\beta_c) = u_\gamma^\tB - Z^\tB \tilde\beta_c.
\end{align*}
In the same way, $e_\gamma^\tS = u_\gamma^\tS - Z^\tS \tilde\beta_c, e_\gamma^\tBS = u_\gamma^\tBS - Z^\tBS \tilde\beta_c.$ As a result, $\tilde{e}_c = \tilde{u}_c - \tilde{Z}_c\tilde\beta_c$, which gives the second claim.
\end{proof}

Firstly, we prove consistency of the estimate for $\beta_c$; the result is restated here for convenience.

\consistentbeta*

\begin{proof}[Proof of \cref{prop:beta_consistency}]
By \cref{lemma:beta_decomposition} and invertibility, $\tilde{e}_c = 0.$ Therefore, 
\begin{equation}
\label{eq:beta_decomposition}
\hat{\beta}_c - \tilde{\beta}_c = \left(\hat{Z}_c^{-1} - \tilde{Z}_c^{-1} \right)(\hat{e}_c - \tilde{e}_c) + \tilde{Z}_c^{-1}\left(\hat{e}_c -\tilde{e}_c\right).
\end{equation}
Then $ \tilde{Z}_c^{-1}\left(\hat{e}_c -\tilde{e}_c\right) =  (I\tilde{Z}_c)^{-1}\left(I\hat{e}_c -I\tilde{e}_c\right)$. Note that $(I\tilde{Z}_c)^{-1} = O(1)$ by \cref{assmp:non_vanishing_buyer_seller_variation} and \cref{assmp:finite_population_limits}. Next, write
\begin{align*}
\left(I\hat{e}_c -I\tilde{e}_c\right) & = \sum_\gamma Ia_{\tB,\gamma}(\hat{e}^\tB_\gamma - \tilde{e}^\tB_\gamma) + \sum_\gamma Ia_{\tS,\gamma}(\hat{e}^\tS_\gamma - \tilde{e}^\tS_\gamma) + \sum_\gamma Ia_{\tBS,\gamma} (\hat{e}^\tBS_\gamma - \tilde{e}^\tBS_\gamma).
\end{align*}
By the formulas in \cref{sec:methodology}, $Ia_{\tB,\gamma} = O(1), Ia_{\tS,\gamma} = O(1)$ and $Ia_{\tBS,\gamma} = O(1/I).$ To show each of these terms is $o_p(1)$, we will first prove that entrywise, $\hat{Z}^\theta_\gamma - Z^\theta = o_p(1)$ and $\hat{e}^\theta_\gamma - e^\theta_\gamma = o_p(1)$ for all $\theta \in \set{\tB,\tS,\tBS}$. As the argument will be the same for each $\gamma,$ we omit dependence on $\gamma$ when clear. We apply \cref{prop:convergence_of_empirical_buyer_seller_term} with $a_{ij} = X_{ij}^{(k)}$ for $k = 1,\dots,d$ and $b_{ij} = e_{ij}$. Under \cref{assmp:covariate_assumption}, the condition $\max_{i} |\bar{X}_{i,\bullet} - \dbar{X}| = O(1)$ directly yields $\hat{e}^\tB_\gamma - \tilde{e}^\tB_\gamma = O_p(1/\sqrt{I})$. By symmetry, the same argument holds for the seller terms, so that $\hat{e}^\tS_\gamma - \tilde{e}^\tS_\gamma = O_p(1/\sqrt{I})$. Next, applying \cref{prop:convergence_of_empirical_dcr_term} with \cref{assmp:covariate_assumption} shows that $\hat{e}^\tBS_\gamma - \tilde{e}^\tBS_\gamma = o_p(1)$. Collecting terms, we have shown $(I \hat e_c - I \tilde e_c) = O_p(1/\sqrt{I}).$ \\

For the first term in \eqref{eq:beta_decomposition}, write $\left(\hat{Z}_c^{-1} - \tilde{Z}_c^{-1} \right)(\hat{e}_c - \tilde{e}_c) = \left((I\hat{Z}_c)^{-1} - (I\tilde{Z}_c)^{-1} \right)(I\hat{e}_c - I\tilde{e}_c)$. We will show $\norm{(I\hat{Z}_c)^{-1} - (I\tilde{Z}_c)^{-1}}_{\text{op}} = o_p(1)$. To bound the difference of matrix inverses, we will utilize Proposition E.1 of \cite{lei2021regression}, which states that for any symmetric matrices $A,B$ with $A$ positive definite, $A+B$ invertible, 
\[
\opnorm{(A+B)^{-1} - A^{-1}} \leq \frac{\opnorm{A^{-1}}^2 \opnorm{B}}{1 - \min(1, \opnorm{A^{-1}}\cdot \opnorm{B})}.
\]
Apply this by taking $A := I\tilde{Z}_c$ and $B := I\hat{Z}_c - I\tilde{Z}_c$. By supposition, $\norm{(I\tilde{Z}_c)^{-1}}_{\text{op}} = O(1)$, so it suffices to show that $\norm{I\hat{Z}_c - I\tilde{Z}_c}_{\text{op}} = o_p(1).$ We will show that  the entries of $I\hat{Z}_c - I\tilde{Z}_c$ are $o_p(1)$. Using again the formulas in \cref{sec:methodology}, we have
\begin{align*}
\left(I\hat{Z}_c -I\tilde{Z}_c\right) & = \sum_\gamma Ia_{\tB,\gamma}(\hat{Z}^\tB_\gamma - \tilde{Z}^\tB_\gamma) + \sum_\gamma Ia_{\tS,\gamma} (\hat{Z}^\tS_\gamma - \tilde{Z}^\tS_\gamma) + \sum_\gamma Ia_{\tBS,\gamma} (\hat{Z}^\tBS_\gamma - \tilde{Z}^\tBS_\gamma).
\end{align*}

\Cref{assmp:covariate_assumption} implies
\begin{align*}
& \frac{1}{I}\max_i |\bar{X}_{i,\bullet}^{(k)} - \dbar{X}^{(k)}| = o(1/\sqrt{I}), \quad \frac{1}{I}\max_{ij} |X_{ij}^{(k)} - \bar{X}_{i,\bullet}^{(k)}| = o(1/\sqrt{I})
\end{align*} 
for each coordinate $k$, and symmetrically for the seller averages. Applying \cref{prop:convergence_of_empirical_buyer_seller_term} with two coordinates of the covariates, $X_{ij}^{(k_1)}, X_{ij}^{(k_2)}$ for $k_1, k_2 \in \set{1,\dots,d}$ shows $Ia_{\tB,\gamma}(\hat{Z}^\tB_\gamma - \tilde{Z}^\tB_\gamma)$ and $Ia_{\tS,\gamma}(\hat{Z}^\tS_\gamma - \tilde{Z}^\tS_\gamma)$ are $o_p(1)$. Applying \cref{prop:convergence_of_empirical_dcr_term} shows $Ia_{\tBS,\gamma}\left(\hat{Z}^\tBS_\gamma - \tilde{Z}^\tBS_\gamma\right) = o_p(1).$ 
\end{proof}

\begin{proof}[Proof of \cref{prop:same_asymptotic_distribution}]
 We aim to show that
\[
\hat{\tau}_{c}(\hat\beta_{c}) - \hat{\tau}_{c}(\tilde{\beta}_{c}) = o_p\left( \sqrt{\Var(\hat{\tau}_{c}(\tilde{\beta}_{c}))} \right).
\]
We may write
\begin{align}
\hat{\tau}_{c}(\hat\beta_{c}) - \hat{\tau}_{c}(\tilde{\beta}_{c}) & = \mathbb{X}^\top(\hat{\beta}_c - \tilde{\beta}_c) \\
\mathbb{X} & := \sum_\gamma c_\gamma \edavg{X}_\gamma. 
\end{align}
Since $\sum_\gamma c_\gamma = 0, \mathbb{X}$ is mean zero. Thus $\norm{\mathbb{X}}_2 = O_p\left(\sum_{k=1}^d \sqrt{\Var(\mathbb{X}^{(k)})} \right)$. \Cref{assmp:finite_population_limits} ensures $\nu^\theta(X^{(k)}) = O(1)$ for each $\theta \in \set{\tB,\tS,\tBS}$. \cref{prop:generic_variance_formula,prop:generic_covariance_formula} then show $\Var(\mathbb{X}^{(k)}) = O(1/I)$ for all $k$. By \cref{prop:beta_consistency}, we conclude
\begin{equation*}
\mathbb{X}^\top(\hat{\beta}_c - \tilde{\beta}_c) = O_p(1/I),
\end{equation*}
which is $o_p(\Var(\hat{\tau}_c(\tilde{\beta}_c))^{1/2})$ by \cref{assmp:non_vanishing_buyer_seller_variation}.
\end{proof}

\begin{proof}[Proof of \cref{thm:valid_inference}]

Define $\hat V_c(b)$ to be the conservative variance estimator applied to the matrices $y_{ij} - X_{ij}^\top b$ and $V_c(b)$ similarly where $V_c$ was defined in \eqref{eq:conservative_upper_bound_on_variance}. Next, recall $B_I(M) := \set{b : \norm{b - \tilde{\beta}_c} \leq MI^{-1/2}}$. \cref{prop:proof_of_uniform_ratio_convergence} will show that under \cref{assmp:finite_population_limits,assmp:main_assumption,assmp:covariate_assumption,assmp:non_vanishing_buyer_seller_variation,assmp:per_gamma_variance_lb}, we have 
\begin{equation}
\label{eq:uniform_ratio_convergence_of_conservative_estimator}
\sup_{b \in B_I(M)} \left| \frac{\hat{V}_c(b)}{V_c(b)} - 1 \right| = o_p(1)
\end{equation}
for every fixed $M > 0.$ We will finish the proof using this condition. From \cref{sec:methodology}, the oracle variance is a quadratic function of $\beta$:
\[
\Var(\hat \tau_c(\beta)) = \beta^\top \tilde Z_c\beta-2\tilde u_c^\top \beta + C_c.
\]
Since $\tilde \beta_c=\tilde Z_c^{-1}\tilde u_c$, completing the square gives the exact identity
\[
\Var(\hat \tau_c (\beta))-\Var(\hat \tau_c (\tilde \beta _c))=
(\beta-\tilde \beta_c)^\top \tilde Z_c(\beta-\tilde \beta_c)
\qquad\text{for every }\beta\in \bR^d.
\]
In particular,
\[
\Var(\hat \tau_c (\hat \beta_c))-\Var(\hat \tau_c (\tilde \beta_c))
=(\hat \beta_c-\tilde \beta_c)^\top \tilde Z_c(\hat \beta_c-\tilde \beta_c)\ge 0.
\]
We now quantify its size. By the formula
\[
\tilde Z_c=\sum_{\gamma\in \Gamma} a_{\gamma,B}Z^B+\sum_{\gamma\in \Gamma} a_{\gamma,S}Z^S+
\sum_{\gamma\in \Gamma} a_{\gamma,BS}Z^{BS},
\]
with $a_{\gamma,B}=O(I^{-1})$, $a_{\gamma,S} = O(I^{-1})$, and
$a_{\gamma,BS} = O(I^{-2})$, while $Z^B,Z^S,Z^{BS}=O(1)$ entrywise by \cref{assmp:finite_population_limits}. We then have $\norm{\tilde Z_c}_{\text{op}} = O(I^{-1}).$ Therefore by \cref{prop:beta_consistency},
\[
0\le \Var(\hat \tau_c (\hat \beta_c)) - \Var(\hat \tau_c (\tilde \beta_c))
\le \norm{\tilde Z_c}_{\text{op}} \norm{\hat \beta_c-\tilde \beta_c}_2^2
=O(I^{-1})\cdot O_p(I^{-1})
=O_p(I^{-2}).
\]
Because $\Var(\hat \tau_c (\tilde \beta_c)) =\omega(I^{-2})$ by \cref{assmp:non_vanishing_buyer_seller_variation}, it follows that
\[
\Var(\hat \tau_c(\hat \beta_c))=\Var(\hat \tau_c(\tilde \beta_c))\{1+o_p(1)\}.
\]
Since $V_c(\beta)\ge \Var(\hat \tau_c (\beta))$ for every $\beta$, we deduce
\[
V_c(\hat \beta_c)\ge \Var(\hat \tau_c (\hat \beta_c)) = \Var(\hat \tau_c (\tilde \beta_c))\{1+o_p(1)\}.
\]

It remains to replace the population quantity $V_c(\hat \beta_c)$ by the sample quantity
$\hat V_c(\hat \beta_c)$. We claim that
\[
\frac{\hat V_c(\hat \beta_c)}{V_c(\hat \beta_c)}\to_p 1.
\]
Fix $\varepsilon,\eta>0$. By \cref{prop:beta_consistency}, there exists $M<\infty$ such that for all large $I$,
\[
\Prob\bigl(\norm{\hat \beta_c-\tilde \beta_c}_2> M I^{-1/2}\bigr)<\eta.
\]
On the event $\{\norm{\hat \beta_c-\tilde \beta_c}_2\le M I^{-1/2}\}$,
\[
\left|\frac{\hat V_c(\hat \beta_c)}{V_c(\hat \beta_c)}-1\right|
\le
\sup_{\norm{b-\tilde \beta_c}_2\le M I^{-1/2}}
\left|\frac{\hat V_c(b)}{V_c(b)}-1\right|.
\]
\eqref{eq:uniform_ratio_convergence_of_conservative_estimator} makes the right-hand side $o_p(1)$, so for all sufficiently large $I$,
\[
\Prob\left(\left|\frac{\hat V_c(\hat \beta_c)}{V_c(\hat \beta_c)}-1\right|>\varepsilon\right)
\le 2\eta.
\]
Since $\eta$ is arbitrary, the claim follows. Combining the last two displays gives
\[
\frac{\hat V_c(\hat \beta_c)}{\Var(\hat\tau_c(\tilde \beta_c))}
=
\frac{\hat V_c(\hat \beta_c)}{V_c(\hat \beta_c)}
\cdot
\frac{V_c(\hat \beta_c)}{\Var(\hat\tau_c(\tilde \beta_c))}
\ge (1-o_p(1))(1+o_p(1))
=1-o_p(1).
\]

Next, let
\[
Z_I := \frac{\hat \tau_c(\tilde \beta_c) - \tau_c}{\Var(\hat\tau_c(\tilde \beta_c))^{1/2}}
\qquad 
R_I:=\frac{\hat \tau_c(\hat \beta_c)-\hat \tau_c(\tilde \beta_c)}{\Var(\hat\tau_c(\tilde \beta_c))^{1/2}},
\qquad
S_I:=\left(\frac{\hat V_c(\hat \beta_c)}{\Var(\hat\tau_c(\tilde \beta_c))}\right)^{1/2}.
\]
By \cref{thm:clt_population_regression}, $Z_I \Rightarrow N(0,1)$, while the argument of \cref{prop:same_asymptotic_distribution} shows $R_I=o_p(1)$. We have just shown that $S_I\ge 1-o_p(1)$, which means that for every fixed
$\delta\in(0,1)$, $\Prob(S_I\ge 1-\delta)\to 1.$ The coverage event can be written as
\[
\left\{\tau_c\in
\hat \tau_c(\hat \beta_c)\pm z_{1-\alpha/2}\hat V_c(\hat \beta_c)^{1/2}
\right\}
=
\left\{|Z_I+R_I|\le z_{1-\alpha/2}S_I\right\}.
\]
Hence, for every fixed $\delta\in(0,1)$,
\begin{align*}
\Prob\left(\tau_c\in
\hat \tau_c(\hat \beta_c)\pm z_{1-\alpha/2}\hat V_c(\hat \beta_c)^{1/2}
\right)
&\ge
\Prob\left(|Z_I+R_I|\le z_{1-\alpha/2}(1-\delta),\; S_I\ge 1-\delta\right) \\
&\ge
\Prob\left(|Z_I+R_I|\le z_{1-\alpha/2}(1-\delta)\right)-\Prob(S_I<1-\delta).
\end{align*}
By Slutsky's theorem, $Z_I+R_I\Rightarrow N(0,1)$, while the second term tends to $0$. Therefore,
\[
\liminf_{I\to\infty}
\Prob\left(\tau_c\in
\hat \tau_c(\hat \beta_c)\pm z_{1-\alpha/2}\hat V_c(\hat \beta_c)^{1/2}
\right)
\ge
\Prob\left(|N(0,1)|\le z_{1-\alpha/2}(1-\delta)\right).
\]
Letting $\delta\downarrow 0$ yields
\[
\liminf_{I\to\infty}
\Prob\left(\tau_c\in
\hat \tau_c(\hat \beta_c)\pm z_{1-\alpha/2}\hat V_c(\hat \beta_c)^{1/2}
\right)
\ge 1-\alpha.
\]
\end{proof}

\begin{proof}[Proof of \cref{ex:inference_for_direct_effect}]
Recall \cref{ex:optimal_adjustment_direct_effect} shows that $\tilde{Z}_{\dir} = \sum_\gamma \frac{1}{I_\gamma J_\gamma} Z^\tBS, \tilde{u}_{\dir} = \sum_\gamma \frac{1}{I_\gamma J_\gamma} u^\tBS_\gamma$ and $\hat{Z}_{\dir} = \sum_\gamma \frac{1}{I_\gamma J_\gamma} \hat Z^\tBS_\gamma, \hat{u}_{\dir} = \sum_\gamma \frac{1}{I_\gamma J_\gamma} \hat u^\tBS_\gamma$. From \cref{assmp:finite_population_limits} it is immediate that $\tilde{\beta}_{\dir} = O(1)$. We will establish the following claims:
\begin{enumerate}
    \item The variance of the direct effect estimator is $\Theta(1/IJ)$.
    \item $\hat{\beta}_{\dir} - \tilde{\beta}_{\dir} = o_p(1)$.
    \item Asymptotic normality of $\hat{\tau}_{\dir}(\hat{\beta}_{\dir})$.
    \item Asymptotically conservative confidence intervals using the plug-in.
\end{enumerate}

For the first claim, it suffices to prove $\nu^\theta(e_{ij}(\gamma)) = O(1/I)$ and also $\xi^\theta_{\gamma,\gamma'} = O(1/I)$ for $\theta \in \set{\tB, \tS}$. Since $\xi^\theta_{\gamma,\gamma'} \lesssim \nu^\theta_{\gamma} + \nu^\theta_{\gamma'},$ we need only prove $\nu^\theta(e_{ij}(\gamma)) = O(1/I)$. Expand \[
\nu^\theta(e_{ij}(\gamma)) = \beta_{\dir}^\top Z^\theta \beta_{\dir} - 2u^{\theta\top}_\gamma \beta_{\dir} + \nu^\theta_\gamma.
\]
Note that each entry $(k_1,k_2)$ of $Z^\theta$ is bounded by a constant multiple of $\nu^\theta (X^{(k_1)}) + \nu^\theta (X^{(k_2)})$, which is $O(1/I)$ by assumption. Similarly, each coordinate $k$ of $u^\theta_\gamma$ is bounded by a constant multiple of $\nu^\theta (X^{(k)}) + \nu^\theta_\gamma$, which is also $O(1/I).$ The result follows.\\ 

Towards the second claim,  \cref{lemma:beta_decomposition} gives
\begin{equation}
\label{eq:direct_effect_beta}
\hat{\beta}_{\dir} - \tilde{\beta}_{\dir} = \tilde{Z}_{\dir}^{-1} (\hat e_{\dir} - \tilde{e}_{\dir}) + (\hat{Z}^{-1}_{\dir} - \tilde{Z}^{-1}_{\dir}) \hat{e}_{\dir}.
\end{equation}
Rewrite the first term in \eqref{eq:direct_effect_beta} as $(IJ \tilde{Z}_{\dir})^{-1} IJ(\hat e_{\dir} - \tilde{e}_{\dir}).$ It is easy to see that $(IJ \tilde{Z}_{\dir})^{-1} = O(1),$ while $IJ(\hat e_{\dir} - \tilde{e}_{\dir}) = o_p(1)$ since $\hat{e}_\gamma^\tBS - \tilde{e}_{\gamma}^\tBS = o_p(1), \forall \gamma \in \set{\tr,\ib,\is,\cc}$ by \cref{prop:convergence_of_empirical_dcr_term}. The second term in \eqref{eq:direct_effect_beta} will be handled in similar fashion as in \cref{prop:beta_consistency}.
Write $(\hat{Z}^{-1}_{\dir} - \tilde{Z}^{-1}_{\dir}) \hat{e}_{\dir} = \left((IJ\hat{Z}_{\dir})^{-1} - (IJ\tilde{Z}_{\dir})^{-1} \right)IJ\hat{e}_{\dir}$. We have $\norm{IJ\hat{e}_{\dir}}_2 = O_p(1)$ combining \cref{assmp:finite_population_limits} and \cref{lemma:betas_bounded}. Apply again Proposition E.1 of \cite{lei2021regression} with $A := IJ\tilde{Z}_{\dir}$ and $B := IJ(\hat{Z}_{\dir} - \tilde{Z}_{\dir})$. By supposition, $\norm{(IJ\tilde{Z}_{\dir})^{-1}}_{\text{op}} = O(1)$, so it suffices to show that $\norm{IJ(\hat{Z}_{\dir} - \tilde{Z}_{\dir})}_{\text{op}} = o_p(1).$ From the explicit form of the direct effect adjustment, it suffices to show that each entry of $\hat{Z}^\tBS_\gamma - \tilde{Z}^\tBS$ is $o_p(1)$. This follows directly from \cref{prop:convergence_of_empirical_dcr_term} under \cref{assmp:finite_population_limits,assmp:covariate_assumption}.\\

Consider now the third claim. Applying \cref{thm:clt_population_regression} gives 
\[
\frac{\hat{\tau}_{\dir}(\tilde{\beta}_{\dir}) - \hat{\tau}_{\dir}}{\sqrt{\Var(\hat{\tau}_{\dir}(\tilde{\beta}_{\dir}))}} \Rightarrow N(0,1).
\]
Next, observe  
\[
\hat{\tau}_{\dir}(\hat{\beta}_{\dir}) - \hat{\tau}_{\dir}(\tilde{\beta}_{\dir}) = \left(\edavg{X}_\tr - \edavg{X}_\ib - \edavg{X}_\is + \edavg{X}_\cc\right)^\top (\hat \beta_{\dir} - \tilde \beta_{\dir}).
\]
The term $\mathbb{X} := \edavg{X}_\tr - \edavg{X}_\ib - \edavg{X}_\is + \edavg{X}_\cc$ has mean zero. The variance of each coordinate $k$ can be read off from \cref{sec:optimal_direct_effect_derivation} and is equal to $\sum_{\gamma} \frac{1}{I_\gamma J_\gamma} \nu^{\tBS}(X_{ij}^{(k)})$ since the cross-terms $\xi^\theta_{\gamma,\gamma'}$ are zero. Under \cref{assmp:finite_population_limits}, $\Var(\mathbb{X}^{(k)}) = O(1/IJ)$ while $\Var(\hat{\tau}_{\dir}(\tilde{\beta}_{\dir})) = \Theta(1/IJ)$. This shows $\hat{\tau}_{\dir}(\hat{\beta}_{\dir}) - \hat{\tau}_{\dir}(\tilde{\beta}_{\dir}) = o_p\left(\sqrt{\Var(\hat{\tau}_{\dir}(\tilde{\beta}_{\dir}))}\right)$, so $\hat{\tau}_{\dir}(\hat{\beta}_{\dir})$ has the same asymptotic distribution as that of $\hat{\tau}_{\dir}(\tilde{\beta}_{\dir})$. \\

For the last claim, the proof of \cref{thm:valid_inference} reduces to showing
\[
\hat{V}_{\dir}(\hat\beta_{\dir}) - \hat{V}_{\dir}(\tilde\beta_{\dir}) = o_p(1/IJ) = o_p\left(\Var(\hat{\tau}_{\dir}(\tilde{\beta}_{\dir})) \right).    
\]

By inspecting the forms of the conservative variance estimators, one can see that
\begin{equation}
\label{eq:cons_variance_estimator_difference}
\begin{split}
\hat \nu^\theta_\gamma(\hat\beta_c) - \hat \nu^\theta_\gamma(\tilde\beta_c) & = 2\hat u_\gamma^{\theta \top}\left(\hat\beta_c  -\tilde\beta_c \right) + \left(\hat\beta_c^\top \hat Z^\theta_\gamma \hat\beta_c - \tilde\beta_c^\top \hat Z^\theta_\gamma \tilde\beta_c \right) \\
& = 2\hat u_\gamma^{\theta \top}\left(\hat\beta_c  -\tilde\beta_c \right) + (\hat \beta_c +\tilde \beta_c)^\top \hat Z^\theta_\gamma (\hat \beta_c - \tilde{\beta}_c) 
\end{split}
\end{equation}
for each $\theta \in \set{\tB,\tS,\tBS}, \gamma \in \set{\tr,\ib,\is,\cc}.$ By \eqref{eq:cons_variance_estimator_difference}, it suffices to show that $\hat{u}^\tBS_\gamma, \hat Z^\tBS_\gamma = O_p(1)$ and $\hat{u}^\theta_\gamma, \hat Z^\theta_\gamma = O_p(1/I)$ for $\theta \in \set{\tB, \tS}$. The former claim is immediate from \cref{prop:convergence_of_empirical_dcr_term} and $u^\tBS_\gamma, Z^\tBS_\gamma = O_p(1)$ via  \cref{assmp:finite_population_limits}. For the latter, we adapt the proof of \cref{prop:convergence_of_empirical_buyer_seller_term} to show that $\hat{u}^\tB_\gamma - \tilde{u}^\tB_\gamma = O_p(1/I)$, $\hat Z^\tB_\gamma - \tilde Z^\tB = O_p(1/I)$, and similarly for the seller terms. Fix a coordinate $k$ and omit the dependence on $\gamma$. Then \cref{prop:convergence_of_empirical_buyer_seller_term} stochastically bounds the $k$th coordinate of the difference $\hat{u}^\tB_\gamma - \tilde{u}^\tB_\gamma$ by
\begin{equation}
\label{eq:direct_effect_buyer_term_analysis_helper}
\begin{split}
& \frac{1}{J} + \left(\frac{1}{I}\nu^{\tB}(X^{(k)})\right)^{1/2} + \left(\frac{1}{I}\nu^{\tB}(y)\right)^{1/2}  \\
+ & \left(\frac{1}{I} \max_i (\bar{y}_{i,\bullet} -\dbar{y})^2 \nu^\tB(X^{(k)})\right)^{1/2}.
\end{split}
\end{equation}
We will show each of these terms are $O(1/I)$. To do so, 
we bound
\begin{align*}
\frac{1}{I}\max_{i} |\bar{y}_{i,\bullet} -\dbar{y}| & = \frac{1}{I}\max_{i} |\bar{e}_{i,\bullet} -\dbar{e}| + \frac{1}{I}\max_{i} |\bar{X}_{i,\bullet}^\top\tilde\beta_{\dir} - \dbar{X}^\top \tilde\beta_{\dir}| \\
& = o(\sqrt{\Var(\hat \tau_{\dir})}) + \frac{1}{I} \sum_{k=1}^d \max_{i} |\bar{X}_{i,\bullet}^{(k)} - \dbar{X}^{(k)}| |\tilde{\beta}_{\dir}^{(k)}| \\
& = O(1/I),
\end{align*}
since $\max_{i} |\bar{X}_{i,\bullet}^{(k)} -\dbar{X}^{(k)}| = O(1)$. By the assumptions of \cref{ex:inference_for_direct_effect}, $\nu^{\tB}(X^{(k)}) = O(1/I),$ $\max_{i} |\bar{X}_{i,\bullet}^{(k)} -\dbar{X}^{(k)}| = O(1)$ and $\nu^{\tB}(y) = O(1/I).$ This shows that every term in the first line of \eqref{eq:direct_effect_buyer_term_analysis_helper} is $O(1/I).$ For the last term,
\begin{align*}
\frac{1}{I} \max_i (\bar{y}_{i,\bullet} -\dbar{y})^2 & \lesssim \frac{1}{I}\max_{i} (\bar{e}_{i,\bullet} -\dbar{e})^2 + \frac{1}{I}\max_{i} |\bar{X}_{i,\bullet}^\top\tilde\beta_{\dir} - \dbar{X}^\top \tilde\beta_{\dir}|^2 \\
& = o(I\Var(\hat \tau_{\dir})) + O(1/I) \\
& = O(1/I).
\end{align*}
The analysis of $\hat Z^\tB_\gamma - \tilde Z^\tB$ is done the same way, applying \cref{prop:convergence_of_empirical_buyer_seller_term} entrywise with two coordinates of the covariate vector. 
\end{proof}

\subsection{Auxiliary Results}


\begin{proposition}
\label{prop:convergence_of_empirical_buyer_seller_term}
Let $(a_{ij}), (b_{ij})$ be a sequence of $I \times J$ matrices, indexed by $n$. Suppose that the limits of $\nu^\theta(a)$ and $\nu^\theta(b), \theta \in \set{\tB,\tS,\tBS}$ exist and are finite, where $\nu^\theta(-)$ is defined in \eqref{eq:sigma_terms}. Suppose either $\frac{1}{I} \max_{i} |\bar{a}_{i,\bullet} - \dbar{a}|^2 = o(1)$ or $\frac{1}{I} \max_{i} |\bar{b}_{i,\bullet} -\dbar{b}|^2 = o(1)$. Then
\begin{equation}
\label{eq:convergence_of_empirical_buyer_seller_term}
\frac{1}{I_\gamma} \sum_{i \in \cl{I}_\gamma} \eavg{a}^{(d)}_{i,\bullet}\eavg{b}^{(d)}_{i,\bullet} - \frac{1}{I} \sum_{i=1}^I \bar{a}^{(d)}_{i,\bullet}\bar{b}^{(d)}_{i,\bullet} = o_p(1).
\end{equation}
Under the stronger condition that either $\max_{i} |\bar{a}_{i,\bullet} - \dbar{a}|$ or $\max_{i} |\bar{b}_{i,\bullet} - \dbar{b}|$ is $O(1)$, \eqref{eq:convergence_of_empirical_buyer_seller_term} is $O(1/\sqrt{I})$. By symmetry, the analogous result is true when randomizing over the sellers. 
\end{proposition}
\begin{proof}
Combining the definition of $ \eavg{a}^{(d)}_{i,\bullet},\,\eavg{b}^{(d)}_{i,\bullet}$ with a simple manipulation gives
\begin{equation*}
    \frac{1}{I_\gamma} \sum_{i \in \cl{I}_\gamma} \eavg{a}^{(d)}_{i,\bullet}\eavg{b}^{(d)}_{i,\bullet} =
\frac{1}{I_\gamma} \sum_{i \in \cl{I}_\gamma} (\eavg{a}_{i,\bullet} - \edavg{a})(\eavg{b}_{i,\bullet} - \edavg{b}) = \frac{1}{I_\gamma} \sum_{i \in \cl{I}_\gamma} (\eavg{a}_{i,\bullet} - \dbar{a})(\eavg{b}_{i,\bullet} - \dbar{b})  -(\edavg{a} - \dbar{a})(\edavg{b} - \dbar{b}).
\end{equation*}
We have $(\edavg{a} - \dbar{a})(\edavg{b} - \dbar{b}) = O_p\left(\sqrt{\Var(\edavg{a})\Var(\edavg{b})} \right)$, by Cauchy-Schwarz and Chebyshev. Under the assumption that $\nu^\theta(a) = O(1)$ and $\nu^\theta(b) = O(1)$ for all $\theta \in \tB,\tS,\tBS$, we have $\Var(\edavg{a}) = O(1/I)$ and similarly with $\Var(\edavg{b})$ by Proposition \ref{prop:generic_variance_formula}. Thus
\begin{align*}
    \frac{1}{I_\gamma} \sum_{i \in \cl{I}_\gamma} \eavg{a}^{(d)}_{i,\bullet}\eavg{b}^{(d)}_{i,\bullet} - \frac{1}{I} \sum_{i=1}^I \bar{a}^{(d)}_{i,\bullet}\bar{b}^{(d)}_{i,\bullet} & = \frac{1}{I_\gamma} \sum_{i \in \cl{I}_\gamma} (\eavg{a}_{i,\bullet} - \dbar{a})(\eavg{b}_{i,\bullet} - \dbar{b}) - \frac{1}{I} \sum_{i =1}^I \bar{a}^{(d)}_{i,\bullet}\bar{b}^{(d)}_{i,\bullet} + O_p(1/I)
\end{align*}

Decompose the quantity on the right hand side as follows:
\begin{align}
& \frac{1}{I_\gamma} \sum_{i \in \cl{I}_\gamma} (\eavg{a}_{i,\bullet} - \bar{a}_{i,\bullet})(\eavg{b}_{i,\bullet} - \bar{b}_{i,\bullet}) \label{eq:empirical_buyer_seller_decomp_1}\\
  & + \frac{1}{I_\gamma} \sum_{i \in \cl{I}_\gamma} (\bar{a}_{i,\bullet} - \dbar{a}) (\eavg{b}_{i,\bullet} - \bar{b}_{i,\bullet})  \label{eq:empirical_buyer_seller_decomp_2}\\
  & + \frac{1}{I_\gamma} \sum_{i \in \cl{I}_\gamma} (\eavg{a}_{i,\bullet} - \bar{a}_{i,\bullet})(\bar{b}_{i,\bullet} -\dbar{b})   \label{eq:empirical_buyer_seller_decomp_3}\\
  & + \frac{1}{I_\gamma} \sum_{i \in \cl{I}_\gamma} (\bar{a}_{i,\bullet} -\dbar{a}) (\bar{b}_{i,\bullet} -\dbar{b}) - \frac{1}{I} \sum_{i=1}^I (\bar{a}_{i,\bullet} -\dbar{a}) (\bar{b}_{i,\bullet} -\dbar{b})
\end{align}

\begin{enumerate}
    \item For the first term \eqref{eq:empirical_buyer_seller_decomp_1}, the mean is given by $\frac{1}{I} \sum_{i=1}^I \Cov(\eavg{a}_{i,\bullet}, \eavg{b}_{i,\bullet})$ using the tower property. Using the covariance formula for sampling without replacement, this is upper bounded by a constant multiple times
    \[
    \frac{1}{I J^2} \sum_{i,j} (a_{ij} - \bar{a}_{i,\bullet})(b_{ij} -  \bar{b}_{i,\bullet}) .
    \]
    Applying Cauchy-Schwarz gives the bound
    \[
     \frac{1}{J}  \left(\frac{1}{IJ}\sum_{i,j} (a_{ij} - \bar{a}_{i,\bullet})^2 \right)^{1/2} \left(\frac{1}{IJ}\sum_{i,j} (b_{ij} - \bar{b}_{i,\bullet})^2 \right)^{1/2} = \frac{1}{J} \left( (\nu^\tBS(a) + \nu^\tS(a))(\nu^\tBS(b) + \nu^\tS(b)) \right)^{\frac{1}{2}},
    \]
    which is $O(1/J)$ by assumption. Let us now inspect the variance of the term. We apply the law of total variance, conditioning on the buyer indicator variables, to obtain
    \begin{equation}
    \label{eq:case1helper}
    \Var\left( \frac{1}{I_\gamma} \sum_{i \in \cl{I}_\gamma} \Cov(\eavg{a}_{i,\bullet},\eavg{b}_{i,\bullet}) \right) + \E\left[\Var\left( \frac{1}{I_\gamma} \sum_{i \in \cl{I}_\gamma} (\eavg{a}_{i,\bullet} - \bar{a}_{i,\bullet})(\eavg{b}_{i,\bullet} - \bar{b}_{i,\bullet}) \middle| \tB \right) \right]
    \end{equation}
    Firstly, 
    \[
        \Var\left( \frac{1}{I_\gamma} \sum_{i \in \cl{I}_\gamma} \Cov(\eavg{a}_{i,\bullet},\eavg{b}_{i,\bullet}) \right)  \lesssim \frac{1}{I^2} \sum_i \Cov(\eavg{a}_{i,\bullet},\eavg{b}_{i,\bullet})^2 \le \left(\frac{1}{I} \sum_i |\Cov(\eavg{a}_{i,\bullet},\eavg{b}_{i,\bullet})| \right)^2 = O(1/J^2).
    \]
    
    To handle the second term in \eqref{eq:case1helper}, define the vectors $u_j \in \bR^{I_T}, j =1,\dots,J$ with entries $a_{ij} - \bar{a}_{i,\bullet}, i \in \cl{I}_T$ and $v_j \in \bR^{I_T}$ with entries $b_{ij} - \bar{b}_{i,\bullet}, i \in \cl{I}_T$. Then
\[
\sum_{i \in \cl{I}_\gamma} (\eavg{a}_{i,\bullet} - \bar{a}_{i,\bullet})(\eavg{b}_{i,\bullet} - \bar{b}_{i,\bullet}) = \frac{1}{J_T^2}\inner{\sum_{j} W_j^\tS u_j}{\sum_{j} W_j^\tS v_j} = \frac{1}{J_T^2}W^{\tS\top} A W^{\tS}
\]
for a matrix $A$ with entries $A_{ij} = \inner{u_i}{v_j}.$ Notice that $\sum_j u_j = 0, \sum_j v_j = 0$ so the row and column sums of $A_{ij}$ are zero. Applying Lemma A.5 of \cite{lei2021regression}, we have
\[
\E\left[\Var\left( \frac{1}{I_\gamma} \sum_{i \in \cl{I}_\gamma} (\eavg{a}_{i,\bullet} - \bar{a}_{i,\bullet})(\eavg{b}_{i,\bullet} - \bar{b}_{i,\bullet}) | \tB \right) \right] \leq \frac{1}{I_T^2J_T^4}\norm{A}_F^2 = 
\frac{1}{J_T^4}\sum_{j,k} \E\left( \frac{1}{I_T} u_j^\top v_k \right)^2 
\]
Since $u_j^\top v_k = \sum_{i \in \cl{I}_T} (a_{ij} - \bar{a}_{i,\bullet})(b_{ik} - \bar{b}_{i,\bullet})$, we can use the mean and variance formulas for sampling without replacement to obtain
\[
\E\left( \frac{1}{I_T} u_j^\top v_k \right)^2 \lesssim \left(\frac{1}{I} \sum_{i=1} (a_{ij} - \bar{a}_{i,\bullet})(b_{ik} - \bar{b}_{i,\bullet})\right)^2 + \frac{1}{I^2} \sum_{i=1} (a_{ij} - \bar{a}_{i,\bullet})^2(b_{ik} - \bar{b}_{i,\bullet})^2.
\]
Thus,
\begin{align*}
& \frac{1}{J_T^4} \sum_{j,k = 1}^J \left(\frac{1}{I} \sum_{i=1} (a_{ij} - \bar{a}_{i,\bullet})(b_{ik} - \bar{b}_{i,\bullet})\right)^2 + \frac{1}{J_T^4} \sum_{j,k = 1}^J \frac{1}{I^2} \sum_{i=1} (a_{ij} - \bar{a}_{i,\bullet})^2(b_{ik} - \bar{b}_{i,\bullet})^2 \\
& \leq \frac{1}{J_T^4} \sum_{j,k = 1}^J \left(\frac{1}{I} \sum_{i=1} (a_{ij} - \bar{a}_{i,\bullet})(b_{ik} - \bar{b}_{i,\bullet})\right)^2 + \frac{1}{J_T^4} \sum_{j,k = 1}^J \left(\frac{1}{I} \sum_{i=1} |(a_{ij} - \bar{a}_{i,\bullet})(b_{ik} - \bar{b}_{i,\bullet})|\right)^2 \\
 \leq & \frac{2}{J_T^4} \sum_{j,k = 1}^J \left(\frac{1}{I} \sum_{i=1} (a_{ij} - \bar{a}_{i,\bullet})^2\right) \left(\frac{1}{I} \sum_{i=1} (b_{ik} - \bar{b}_{i,\bullet})^2\right) \\
 \leq & \frac{2}{J^2} (\nu^\tBS(a) + \nu^\tS(a)) (\nu^\tBS(b) + \nu^\tS(b)).
\end{align*}

The total contribution for the first term is thus $O(1/J)$.

\item For the second term \eqref{eq:empirical_buyer_seller_decomp_2}, 
    conditioning on $\tB$ shows that the expression is mean zero. Thus it suffices to compute the variance. Using the law of total variance and conditioning on the buyer variables, the variance is given by
    \[
    \E\left[\Var\left[ \frac{1}{I_\gamma} \sum_{i \in \cl{I}_\gamma} \bar{a}^{(d)}_{i,\bullet} \eavg{b}_{i,\bullet} \ | \ \tB \right] \right].
    \]
    The inner term is the finite sample variance of the population indexed by $j$ given by $\frac{1}{I_\gamma} \sum_{i \in \cl{I}_\gamma} (\bar{a}_{i,\bullet} - \dbar{a})b_{ij}$, yielding
    \begin{align*}
    \frac{1}{J^2} \sum_{j=1}^J \E \left[ \left(\frac{1}{I_\gamma} \sum_{i \in \cl{I}_\gamma} \bar{a}^{(d)}_{i,\bullet}(b_{ij} - \bar{b}_{i,\bullet}) \right)^2\right] & = \frac{1}{J^2} \sum_{j=1}^J \left[ \frac{1}{I} \sum_{i = 1}^I \bar{a}^{(d)}_{i,\bullet}(b_{ij} - \bar{b}_{i,\bullet}) \right]^2 \\
    & + \frac{1}{I^2 J^2} \sum_{i,j}^{I,J} \bar{a}^{(d)2}_{i,\bullet}(b_{ij} - \bar{b}_{i,\bullet})^2 - \frac{1}{IJ^2} \sum_{j=1}^J \left( \frac{1}{I} \sum_{i=1}^I \bar{a}_{i,\bullet}(b_{ij} - \bar{b}_{i,\bullet}) \right)^2 \\
    & \leq \frac{1}{J^2} \sum_{j=1}^J \left[ \frac{1}{I} \sum_{i = 1}^I \bar{a}^{(d)}_{i,\bullet}(b_{ij} - \bar{b}_{i,\bullet}) \right]^2 + \frac{1}{I^2 J^2} \sum_{i,j}^{I,J} \bar{a}^{(d)2}_{i,\bullet}(b_{ij} - \bar{b}_{i,\bullet})^2 \\
    & \leq  \frac{2}{J^2} \sum_{j=1}^J \left[ \frac{1}{I} \sum_{i = 1}^I \left|\bar{a}^{(d)}_{i,\bullet}(b_{ij} - \bar{b}_{i,\bullet}) \right| \right]^2.
    \end{align*}
    Applying Cauchy-Schwarz,
    \[
    \frac{1}{J^2} \sum_{j=1}^J \left[ \frac{1}{I} \sum_{i = 1}^I \left|(\bar{a}_{i,\bullet} -\dbar{a})(b_{ij} - \bar{b}_{i,\bullet})\right| \right]^2 \leq \frac{1}{IJ^2} \sum_{i,j}^{I,J} (b_{ij} - \bar{b}_{i,\bullet})^2 \cdot \frac{1}{I} \sum_{i=1}^I (\bar{a}_{i,\bullet} -\dbar{a})^2
    \]
    In turn, bound this by 
    \[
    \frac{1}{J}\left(\nu^\tBS(b) + \nu^\tS(b) \right) \frac{1}{I} \sum_{i=1}^I (\bar{a}_{i,\bullet} -\dbar{a})^2 = \frac{1}{I}\left(\nu^\tBS(b) + \nu^\tS(b) \right) \nu^{\tB}(a) = O\left( \frac{1}{I}\nu^{\tB}(a)\right)
    \]
    Thus the total contribution from the second term is $O_p\left(\left(\frac{1}{I}\nu^{\tB}(a)\right)^{1/2} \right)$.
    
    \item The third term is symmetric to the second term and is  $O_p\left(\left(\frac{1}{I}\nu^{\tB}(b)\right)^{1/2}\right)$.
    \item For the last term, bound it above by $O_p\left(\sqrt{\Var(\frac{1}{I_\gamma} \sum_{i \in \cl{I}_\gamma} \bar{a}^{(d)}_{i,\bullet} \bar{b}^{(d)}_{i,\bullet})} \right)$. Note that 
    \begin{align*}
        \Var\left(\frac{1}{I_\gamma} \sum_{i \in \cl{I}_\gamma} \bar{a}^{(d)}_{i,\bullet} \bar{b}^{(d)}_{i,\bullet} \right) & \lesssim \frac{1}{I^2}\sum_{i=1}^I (\bar{a}_{i,\bullet} -\dbar{a})^2 (\bar{b}_{i,\bullet} -\dbar{b})^2 \\
        & \leq \frac{1}{I} \max_i (\bar{b}_{i,\bullet} -\dbar{b})^2 \frac{1}{I}\sum_{i=1}^I (\bar{a}_{i,\bullet} -\dbar{a})^2 \\
        & \leq \frac{1}{I} \max_i (\bar{b}_{i,\bullet} -\dbar{b})^2 \nu^\tB(a).
    \end{align*}
We could have equivalently used the bound $\frac{1}{I} \max_i (\bar{a}_{i,\bullet} -\dbar{a})^2 \nu^\tB(b)$ for the last term.
\end{enumerate}

We have thus shown the left hand side of \eqref{eq:convergence_of_empirical_buyer_seller_term} is stochastically bounded by
\[
\frac{1}{J} +  \left(\frac{1}{I}\nu^{\tB}(a)\right)^{1/2} +  \left(\frac{1}{I}\nu^{\tB}(b)\right)^{1/2} + \left(\frac{1}{I} \max_i (\bar{b}_{i,\bullet} -\dbar{b})^2 \nu^\tB(a) \wedge  \frac{1}{I} \max_i (\bar{a}_{i,\bullet} -\dbar{a})^2 \nu^\tB(b) \right)^{1/2}
\]

By assumption, $\nu^\tB(a),\nu^\tB(b) = O(1)$. Under the condition that either $\frac{1}{I} \max_i (\bar{b}_{i,\bullet} -\dbar{b})^2 = o(1)$ or $\frac{1}{I} \max_i (\bar{a}_{i,\bullet} -\dbar{a})^2 = o(1)$, the result is $o(1)$. Under the stronger condition that $\max_i |\bar{a}_{i,\bullet} -\dbar{a} |$ or $\max_i |\bar{b}_{i,\bullet} -\dbar{b} | = O(1)$, the result is $O_p(1/\sqrt{I}).$
\end{proof}
\begin{proposition}
\label{prop:convergence_of_empirical_dcr_term}
Let $(a_{ij}), (b_{ij})$ be a sequence of $I \times J$ matrices, indexed by $n$. Suppose the weaker set of conditions of \cref{prop:convergence_of_empirical_buyer_seller_term} and further that either $\frac{1}{I}\max_{ij}|a_{ij} - \dbar{a}|^2 = o(1)$ or $\frac{1}{I}\max_{ij}|b_{ij} - \dbar{b}|^2 = o(1)$. Then
\begin{equation}
\label{eq:convergence_of_empirical_dcr_term_eq}
\frac{1}{I_\gamma J_\gamma} \sum_{(i,j) \in \gamma} \hat{a}^{(dcr)}_{ij}\hat{b}^{(dcr)}_{ij} - \frac{1}{IJ} \sum_{i,j} a^{(dcr)}_{ij}b^{(dcr)}_{ij} = o_p(1).
\end{equation}
\end{proposition}
\begin{proof}
Since double decentering forces each row and column sum to be zero, it follows that $\sum_{(i,j) \in \gamma} \hat{a}^{(dcr)}_{ij}\hat{b}^{(dcr)}_{ij} = \sum_{(i,j) \in \gamma} a_{ij}\hat{b}^{(dcr)}_{ij}$ and similarly $\sum_{i,j} a_{ij}^{(dcr)}b^{(dcr)}_{ij} = \sum_{i,j} a_{ij}b^{(dcr)}_{ij}.$ From this, one can show that the left hand side of \eqref{eq:convergence_of_empirical_dcr_term_eq} is equal to 
\begin{align*} 
\widehat{\bar{\bar{ab}}} - \bar{\bar{ab}} & -  \left( \frac{1}{I_\gamma} \sum_{i \in \cl{I}_\gamma} 
\eavg{a}_{i,\bullet} \eavg{b}_{i,\bullet} -
\frac{1}{I} \sum_{i = 1}^I \bar{a}_{i,\bullet} \bar{b}_{i,\bullet}\right) -  \left( \frac{1}{J_\gamma} \sum_{j \in \cl{J}_\gamma} \eavg{a}_{j,\bullet} \eavg{b}_{i,\bullet} -
\frac{1}{J} \sum_{j = 1}^J \bar{a}_{j,\bullet} \bar{b}_{j,\bullet}\right) + \left(\edavg{a} \edavg{b} - \dbar{a}\dbar{b} \right).
\end{align*}
Now, observe that
\begin{align*}
    & \widehat{\bar{\bar{ab}}} - \bar{\bar{ab}} =  \frac{1}{I_\gamma J_\gamma} \sum_{i,j \in \gamma} (a_{ij} - \dbar{a})(b_{ij}-\dbar{b}) - \frac{1}{IJ}\sum_{ij} (a_{ij} - \dbar{a})(b_{ij}-\dbar{b}) + (\dbar{a} \edavg{b} + \edavg{a} \dbar{b} - 2\dbar{a}\dbar{b}), \\
    & \frac{1}{I_\gamma} \sum_{i \in \cl{I}_\gamma} 
    \eavg{a}_{i,\bullet} \eavg{b}_{i,\bullet} -
    \frac{1}{I} \sum_{i = 1}^I \bar{a}_{i,\bullet} \bar{b}_{i,\bullet} = \left( \frac{1}{I_\gamma} \sum_{i \in \cl{I}_\gamma} 
    \eavg{a}_{i,\bullet}^{(d)} \eavg{b}_{i,\bullet}^{(d)} -
    \frac{1}{I} \sum_{i = 1}^I \bar{a}_{i,\bullet}^{(d)} \bar{b}_{i,\bullet}^{(d)} \right) + (\edavg{a}\edavg{b} - \dbar{a}\dbar{b}) \\
& \frac{1}{J_\gamma} \sum_{j \in \cl{J}_\gamma} \eavg{a}_{\bullet,j} \eavg{b}_{\bullet,j} -
\frac{1}{J} \sum_{j = 1}^J \bar{a}_{\bullet,j} \bar{b}_{\bullet,j} = \left( \frac{1}{J_\gamma} \sum_{j \in \cl{J}_\gamma} \eavg{a}_{\bullet,j}^{(d)} \eavg{b}_{\bullet,j}^{(d)} -
\frac{1}{J} \sum_{j = 1}^J \bar{a}_{\bullet,j}^{(d)} \bar{b}_{\bullet,j}^{(d)}\right) + (\edavg{a}\edavg{b} - \dbar{a}\dbar{b}).
\end{align*}
Thus, the left hand side of \eqref{eq:convergence_of_empirical_dcr_term_eq} can equivalently be written as
\begin{align}
    \label{eq:convergence_dcr_term_I}
    & \frac{1}{I_\gamma J_\gamma} \sum_{i,j \in \gamma} (a_{ij} - \dbar{a})(b_{ij}-\dbar{b}) - \frac{1}{IJ}\sum_{ij} (a_{ij} - \dbar{a})(b_{ij}-\dbar{b}) \\
    \label{eq:convergence_dcr_term_II}
    - & \left(\frac{1}{I_\gamma} \sum_{i \in \cl{I}_\gamma} 
    \eavg{a}_{i,\bullet}^{(d)} \eavg{b}_{i,\bullet}^{(d)} -
    \frac{1}{I} \sum_{i = 1}^I \bar{a}_{i,\bullet}^{(d)} \bar{b}_{i,\bullet}^{(d)}\right) \\
    \label{eq:convergence_dcr_term_III}
    - & \left(\frac{1}{J_\gamma} \sum_{j \in \cl{J}_\gamma} 
    \eavg{a}_{\bullet,j}^{(d)} \eavg{b}_{\bullet,j}^{(d)} -
    \frac{1}{J} \sum_{j = 1}^J \bar{a}_{\bullet,j}^{(d)} \bar{b}_{\bullet,j}^{(d)}\right) \\
    - & (\edavg{a} - \dbar{a}) (\edavg{b} - \dbar{b}).
\end{align}
By the same argument as \cref{prop:convergence_of_empirical_buyer_seller_term} the last three terms are $o_p(1)$. For the first term \eqref{eq:convergence_dcr_term_I}, we apply \cref{lemma:variance_of_davg_kronecker_prod}. Here, note by Jensen that
\begin{equation}
\label{eq:convergence_dcr_term_helper}
\frac{1}{I}\nu^\tB(a^{(d)}b^{(d)}) \lesssim  \frac{1}{I^2} \sum_{i=1}^I \left(\frac{1}{J} \sum_{j=1}^J a_{ij}^{(d)} b_{ij}^{(d)} \right)^2 \leq \frac{1}{I}\max_{ij}(a_{ij} -\dbar{a})^2 \cdot \frac{1}{IJ} \sum_{ij} (b_{ij} - \dbar{b})^2.
\end{equation}

Under the assumed condition that $\frac{1}{I}\max_{ij} |a_{ij} - \dbar{a}|^2 = o(1)$, \eqref{eq:convergence_dcr_term_helper} is $o(1)$. By symmetry a similar bound holds for $\frac{1}{J}\nu^\tS(a^{(d)}b^{(d)}).$ Lastly, by the same argument as in the beginning of the proof, $\nu^\tBS(y) = \frac{1}{IJ} \sum_{ij}^{I,J} y_{ij}^2 - \frac{1}{I} \sum_{i=1}^I \bar{y}_{i,\bullet}^2 - \frac{1}{J} \sum_{j=1}^J \bar{y}_{\bullet,j}^2 + \dbar{y}^2 \leq \frac{1}{IJ} \sum_{ij}^{I,J} y_{ij}^2 + \dbar{y}^2$. Combining this with Jensen's inequality,
\begin{align*}
\frac{1}{IJ}\nu^\tBS(a^{(d)}b^{(d)}) & \lesssim \frac{1}{I^2J^2} \sum_{i,j}(a_{ij} -\dbar{a})^2 (b_{ij} -\dbar{b})^2 + \frac{1}{IJ} \left(\frac{1}{IJ} \sum_{i,j}(a_{ij} -\dbar{a})(b_{ij} -\dbar{b}) \right)^2 \\
& \leq \frac{2}{I^2J^2} \sum_{i,j}(a_{ij} -\dbar{a})^2 (b_{ij} -\dbar{b})^2 \\
& \lesssim \frac{1}{I^2} \max_{i,j} (a_{ij} -\dbar{a})^2 \cdot \frac{1}{IJ}\sum_{i,j} (b_{ij} -\dbar{b})^2 = o(1).
\end{align*}
\eqref{eq:convergence_dcr_term_I} could have been equivalently bounded by switching $a$ and $b$, so it is also $o_p(1)$ if $\frac{1}{I} \max_{ij} \left|b_{ij} -\dbar{b}\right|^2 = o_p(1).$ Collecting the terms shows that  \eqref{eq:convergence_dcr_term_I} is $o_p(1)$.

\end{proof}


\begin{lemma}
\label{lemma:variance_of_davg_kronecker_prod}
Under the same assumptions as above, we have
\[
\edavg{ab} - \dbar{ab} = O_p\left(\sqrt{\frac{1}{I}\nu^\tB(ab) + \frac{1}{J}\nu^\tS(ab) + \frac{1}{IJ}\nu^\tBS(ab)} \right).
\]
\end{lemma}
\begin{proof}
Simply apply the variance formula in \cref{sec:variance_formulas}.
\end{proof}

\begin{proposition}
\label{prop:proof_of_uniform_ratio_convergence}
Under \cref{assmp:finite_population_limits,assmp:main_assumption,assmp:covariate_assumption,assmp:non_vanishing_buyer_seller_variation,assmp:per_gamma_variance_lb}, \eqref{eq:uniform_ratio_convergence_of_conservative_estimator} holds for all $M > 0$.
\end{proposition}
\begin{proof}
Firstly, we will show that \cref{assmp:main_assumption} continues to hold when perturbing $\tilde{\beta}_c$ over a small ball $B_I(M)$, for every fixed $M$. Set $e_{ij}^\star(\gamma) = y_{ij}(\gamma) - X_{ij}^\top \tilde \beta_c$. Precisely, we will show
\begin{align}
    \sup_{b\in B_I(M)}
    & \max\left\{
    \frac1I\max_i \abs{\bar e^{(d)}_{i,\bullet}(b,\gamma)},\;
    \frac1J\max_j \abs{\bar e^{(d)}_{\bullet,j}(b,\gamma)},\;
    \frac{\sqrt{\log I}}{I^{3/2}}\max_{i,j}\abs{e_{ij}(b,\gamma)-\bar e_{i,\bullet}(b,\gamma)},\right. \\
    \label{eq:uniform_main_assumption}
    & \left.\frac{\sqrt{\log J}}{J^{3/2}}\max_{i,j}\abs{e_{ij}(b,\gamma)-\bar e_{\bullet,j}(b,\gamma)},\;
    \frac1{I^2}\norm{e^{(dcr)}(b,\gamma)}_{\text{op}}
    \right\}
    =o\!\left(\Var(\hat \tau_c(\tilde \beta_c))^{1/2}\right).
\end{align}

Fix $\gamma\in\Gamma$ and $b\in B_I(M)$. Write $\Delta:=b-\tilde \beta_c,$ so that $\norm{\Delta}_2\le MI^{-1/2}$. Since $e_{ij}(b,\gamma)=e^\star_{ij}(\gamma)-X_{ij}^\top\Delta$, we have
\[
\bar e^{(d)}_{i,\bullet}(b,\gamma)
=
\bar e^{\star(d)}_{i,\bullet}(\gamma)
-
\sum_{k=1}^d \Delta_k\bigl(\bar X_{i,\bullet}^{(k)}-\dbar{X}^{(k)}\bigr).
\]
Hence
\begin{align*}
\frac1I\max_i \abs{\bar e^{(d)}_{i,\bullet}(b,\gamma)}
&\le
\frac1I\max_i \abs{\bar e^{\star(d)}_{i,\bullet}(\gamma)}
+
\frac{\norm{\Delta}_1}{I}\max_{1\le k\le d}\max_i \abs{\bar X_{i,\bullet}^{(k)}-\dbar{X}^{(k)}} \\
&\le
\frac1I\max_i \abs{\bar e^{\star(d)}_{i,\bullet}(\gamma)}
+
M\sqrt d\,I^{-3/2}
\max_{1\le k\le d}\max_i \abs{\bar X_{i,\bullet}^{(k)}-\dbar{X}^{(k)}}.
\end{align*}
By \cref{assmp:main_assumption}, the first term is $o(\Var(\hat \tau_c(\tilde \beta_c))^{1/2})$. The
second term is $O(I^{-3/2})=o(\Var(\hat \tau_c(\tilde \beta_c))^{1/2})$ because
$\Var(\hat \tau_c(\tilde \beta_c))=\omega(I^{-2})$ and $\max_i \abs{\bar X_{i,\bullet}^{(k)}-\dbar{X}^{(k)}} = O(1)$ from \cref{assmp:covariate_assumption}. Therefore
\[
\sup_{b\in B_I(M)} \frac1I\max_i \abs{\bar e^{(d)}_{i,\bullet}(b,\gamma)}
=o\!\left(\Var(\hat \tau_c(\tilde \beta_c))^{1/2}\right).
\]
The seller-side bound is identical:
\[
\sup_{b\in B_I(M)} \frac1J\max_j \abs{\bar e^{(d)}_{\bullet,j}(b,\gamma)}
=o\!\left(\Var(\hat \tau_c(\tilde \beta_c))^{1/2}\right).
\]

Next,
\[
e_{ij}(b,\gamma)-\bar e_{i,\bullet}(b,\gamma)
=
e^\star_{ij}(\gamma)-\bar e^\star_{i,\bullet}(\gamma)
-
\sum_{k=1}^d \Delta_k\bigl(X_{ij}^{(k)}-\bar X_{i,\bullet}^{(k)}\bigr).
\]
Therefore
\begin{align*}
&\frac{\sqrt{\log I}}{I^{3/2}}\max_{i,j}\abs{e_{ij}(b,\gamma)-\bar e_{i,\bullet}(b,\gamma)} \\
&\qquad\le
\frac{\sqrt{\log I}}{I^{3/2}}\max_{i,j}\abs{e^\star_{ij}(\gamma)-\bar e^\star_{i,\bullet}(\gamma)}
+
\frac{\sqrt{\log I}\,\norm{\Delta}_1}{I^{3/2}}
\max_{1\le k\le d}\max_{i,j}\abs{X_{ij}^{(k)}-\bar X_{i,\bullet}^{(k)}}.
\end{align*}
By \cref{assmp:covariate_assumption},
\[
\max_{i,j}\abs{X_{ij}^{(k)}-\bar X_{i,\bullet}^{(k)}}
\le
\max_{i,j}\abs{X_{ij}^{(k)}-\dbar{X}^{(k)}}
+
\max_i\abs{\bar X_{i,\bullet}^{(k)}-\dbar{X}^{(k)}}
=o(I^{1/2}).
\]
Since $d$ is fixed and $\norm{\Delta}_1\le M\sqrt d\,I^{-1/2}$, the second term above is
$o(\sqrt{\log I}\,I^{-3/2})=o(I^{-1})$, which is $o(\Var(\hat \tau_c(\tilde \beta_c))^{1/2})$ because
$\Var(\hat \tau_c(\tilde \beta_c))^{1/2}=\omega(I^{-1})$. Combining with \cref{assmp:main_assumption} gives
\[
\sup_{b\in B_I(M)}
\frac{\sqrt{\log I}}{I^{3/2}}\max_{i,j}\abs{e_{ij}(b,\gamma)-\bar e_{i,\bullet}(b,\gamma)}
=o\!\left(\Var(\hat \tau_c(\tilde \beta_c))^{1/2}\right).
\]
The column-centered bound is analogous. Because $J\asymp I$, it follows that
\[
\sup_{b\in B_I(M)}
\frac{\sqrt{\log J}}{J^{3/2}}\max_{i,j}\abs{e_{ij}(b,\gamma)-\bar e_{\bullet,j}(b,\gamma)}
=o\!\left(\Var(\hat \tau_c(\tilde \beta_c))^{1/2}\right).
\]

Finally,
\[
e^{(dcr)}_{ij}(b,\gamma)
=
e^{\star(dcr)}_{ij}(\gamma)-\sum_{k=1}^d \Delta_k X_{ij}^{(k),dcr},
\]
so, viewing these arrays as matrices,
\begin{align*}
\frac1{I^2}\norm{e^{(dcr)}(b,\gamma)}_{\text{op}}
&\le
\frac1{I^2}\norm{e^{\star(dcr)}(\gamma)}_{\text{op}}
+
\frac{\norm{\Delta}_1}{I^2}\max_{1\le k\le d}\norm{X^{(k),dcr}}_{\text{op}} \\
&\le
\frac1{I^2}\norm{e^{\star(dcr)}(\gamma)}_{\text{op}}
+
M\sqrt d\,I^{-1/2}\max_{1\le k\le d}\frac1{I^2}\norm{X^{(k),dcr}}_{\text{op}}.
\end{align*}
By the operator norm restriction on the covariates, the second term is $O(I^{-3/2})$, hence again
$o(\Var(\hat \tau_c(\tilde \beta_c))^{1/2})$. Therefore, for all $M > 0$,
\[
\sup_{b\in B_I(M)}\frac1{I^2}\norm{e^{(dcr)}(b,\gamma)}_{\text{op}}
=o\!\left(\Var(\hat \tau_c(\tilde \beta_c))^{1/2}\right).
\]
Collecting the five displayed bounds verifies \eqref{eq:uniform_main_assumption}. \\

Now, we will prove \eqref{eq:uniform_ratio_convergence_of_conservative_estimator}. Throughout the rest of the proof, we will define 
\[
\Sigma_\gamma(b) := \Var(\edavg{e}(b,\gamma)).
\]
Fix $M>0$. For $u\in\bR^d$ with $\norm{u}_2\le M$, set $b(u):=\tilde \beta_c+I^{-1/2}u.$ By Lemma~\ref{lem:quadratic-sigma}, for each $\gamma\in\Gamma$ the function
\[
q_{\gamma,I}(u):=\hat\Sigma_\gamma\bigl(b(u)\bigr)-\Sigma_\gamma\bigl(b(u)\bigr)
\]
belongs to $\mathcal Q_d$. Lemma~\ref{lem:stencil} therefore yields a finite set $\mathcal U_M$ and
constant $C_M<\infty$ such that
\[
\sup_{b\in B_I(M)}\abs{\hat\Sigma_\gamma(b)-\Sigma_\gamma(b)}
=\sup_{\norm{u}_2\le M}\abs{q_{\gamma,I}(u)}
\le C_M\max_{u\in\mathcal U_M}\abs{q_{\gamma,I}(u)}.
\]

Fix $u\in\mathcal U_M$. Because $b(u)\in B_I(M)$, \eqref{eq:uniform_main_assumption} shows that the residualized finite population $e_{ij}(b(u),\gamma)$ satisfies the hypotheses of
\cref{thm:consistent_variance_estimator}. The proof of that theorem shows that
\[
\frac{\hat\Sigma_\gamma(b(u))}{\Sigma_\gamma(b(u))}\to_p 1.
\]
By \cref{assmp:per_gamma_variance_lb}, $\Sigma_\gamma(b(u)) = \Omega( \Var(\hat \tau_c(\tilde \beta_c)))$, so
\[
\hat\Sigma_\gamma(b(u)) - \Sigma_\gamma(b(u)) = o_p\!\left(\Var(\hat \tau_c(\tilde \beta_c))\right).
\]
Since $\mathcal U_M$ is finite,
\[
\max_{u\in\mathcal U_M}\abs{\hat\Sigma_\gamma(b(u)) - \Sigma_\gamma(b(u))}
= o_p\!\left(\Var(\hat \tau_c(\tilde \beta_c))\right).
\]
Therefore,
\[
\sup_{b\in B_I(M)}\abs{\hat\Sigma_\gamma(b)-\Sigma_\gamma(b)}
=o_p\!\left(\Var(\hat \tau_c(\tilde \beta_c))\right).
\]
Using \cref{assmp:per_gamma_variance_lb} again to control the denominator,
\[
\sup_{b\in B_I(M)}
\left|
\frac{\hat\Sigma_\gamma(b)}{\Sigma_\gamma(b)}-1
\right|
\lesssim
\frac{\sup_{b\in B_I(M)}\abs{\hat\Sigma_\gamma(b)-\Sigma_\gamma(b)}}{\Var(\hat \tau_c(\tilde \beta_c))}
=o_p(1).
\]
In particular, with probability tending to one, $\inf_{b\in B_I(M)} \hat\Sigma_\gamma(b)\ge 0,$ for each $\gamma \in \Gamma$. On that event, the truncation $x^+$ is inactive, and therefore
\[
\sup_{b\in B_I(M)}
\left|
\frac{\left(\hat\Sigma_\gamma(b)^{+}\right)^{1/2}}{\Sigma_\gamma(b)^{1/2}}-1
\right|
=
\sup_{b\in B_I(M)}
\left|
\sqrt{\frac{\hat\Sigma_\gamma(b)}{\Sigma_\gamma(b)}}-1
\right|
=o_p(1),
\]
because $x\mapsto \sqrt{x}$ is continuous at $x=1$. Now define
\[
a_\gamma(b):=|c_\gamma|\Sigma_\gamma(b)^{1/2},
\qquad
\hat a_\gamma(b):=|c_\gamma|\hat\Sigma_\gamma(b)_+^{1/2}
\qquad (\gamma\in\Gamma).
\]
Then
\[
V_c(b)^{1/2}=\sum_{\gamma\in\Gamma} a_\gamma(b),
\qquad
\hat V_c(b)^{1/2}=\sum_{\gamma\in\Gamma} \hat a_\gamma(b).
\]
For every $b\in B_I(M)$ with $V_c(b)>0$,
\begin{align*}
\left|\frac{\hat V_c(b)^{1/2}}{V_c(b)^{1/2}}-1\right|
&=
\left|
\sum_{\gamma\in\Gamma}
\frac{a_\gamma(b)}{\sum_{\gamma'\in\Gamma} a_{\gamma'}(b)}
\left(\frac{\hat a_\gamma(b)}{a_\gamma(b)}-1\right)
\right| \\
&\le
\max_{\gamma\in\Gamma}
\left|\frac{\hat a_\gamma(b)}{a_\gamma(b)}-1\right|.
\end{align*}
Taking the supremum over $b\in B_I(M)$ and using that $\Gamma$ is finite gives
\[
\sup_{b\in B_I(M)}
\left|\frac{\hat V_c(b)^{1/2}}{V_c(b)^{1/2}}-1\right|
=o_p(1).
\]
Finally,
\[
\frac{\hat V_c(b)}{V_c(b)}-1
=
\left(\frac{\hat V_c(b)^{1/2}}{V_c(b)^{1/2}}-1\right)
\left(\frac{\hat V_c(b)^{1/2}}{V_c(b)^{1/2}}+1\right),
\]
and the second factor is $O_p(1)$ uniformly on $B_I(M)$. Hence
\[
\sup_{b\in B_I(M)}
\left|\frac{\hat V_c(b)}{V_c(b)}-1\right|
=o_p(1).
\]
    
\end{proof}

\begin{lemma}\label{lem:quadratic-sigma}
For each $\gamma\in \Gamma$, the maps $b\mapsto \Var_\gamma(\edavg{e}(b,\gamma))$ and $b\mapsto \hat\Sigma_\gamma(b)$ are quadratic polynomials on $\bR^d$.
\end{lemma}

\begin{proof}
By definition.
\end{proof}

\begin{lemma}\label{lem:stencil}
Let $\mathcal Q_d$ denote the vector space of real polynomials on $\bR^d$ of total degree at most $2$.
For every fixed $M>0$, there exist a finite set $\mathcal U_M\subset\{u\in\bR^d:\norm{u}_2\le M\}$
and a constant $C_M<\infty$ such that
\[
\sup_{\norm{u}_2\le M}|q(u)|\le C_M \max_{u\in\mathcal U_M}|q(u)|
\qquad\text{for every }q\in\mathcal Q_d.
\]
\end{lemma}

\begin{proof}
The space $\mathcal Q_d$ has finite dimension $(d+1)(d+2)/2$. Choose points
$u_1,\dots,u_N$ in the closed ball $\{u:\norm{u}_2\le M\}$, with $N=\dim(\mathcal Q_d)$, such that the
evaluation map
\[
T:\mathcal Q_d\to \bR^N,
\qquad
T(q):=(q(u_1),\dots,q(u_N)),
\]
is injective; a generic choice of $N$ points works because a nonzero quadratic polynomial cannot
vanish on an open set. Set $\mathcal U_M:=\{u_1,\dots,u_N\}$. Then
\[
\norm{q}_{\mathcal U_M}:=\max_{u\in\mathcal U_M}|q(u)|
\]
is a norm on $\mathcal Q_d$, as is
\[
\norm{q}_{\infty,M}:=\sup_{\norm{u}_2\le M}|q(u)|.
\]
All norms are equivalent on a finite-dimensional vector space, so there exists $C_M<\infty$ such that
\[
\norm{q}_{\infty,M}\le C_M\norm{q}_{\mathcal U_M}
\]
for every $q\in\mathcal Q_d$.
\end{proof}

\section{Optimal Interacted Regression Adjustments}
\label{sec:interacted_estimators}

In this section, we consider regression adjustments of the form
\[
\hat{\tau}_{c,\text{int}}(\bm{\beta}) := \sum_{\gamma \in \Gamma} c_\gamma \frac{1}{I_\gamma J_\gamma} \sum_{(i,j) \in \gamma} y_{ij} - \beta_\gamma^\top X_{ij}.
\]
We will refer to these as \textit{interacted imputation estimators}. 
For the purposes of not further complicating the paper, we only derive the forms of optimal regression adjustments without further theoretical analysis. 
Our techniques for the non-interacted setting developed in \cref{sec:proofs_for_asymptotic_theory} should apply here. Let us recall the notation
\begin{equation*}
    u_\gamma^{\tB} := \frac{1}{I-1}\sum_{i=1}^I \overline{X}_{i,\bullet}^{(d)}\overline{y}_{i,\bullet}^{(d)}(\gamma), \quad u_\gamma^{\tS} 
    =
    \frac{1}{J-1}\sum_{j=1}^J \overline{X}_{\bullet,j}^{(d)}\overline{y}_{\bullet,j}^{(d)}(\gamma), \quad u_\gamma^{\tBS} 
    =
    \frac{1}{(I-1)(J-1)}\sum_{i=1}^I\sum_{j=1}^J X_{ij}^{(dcr)} y_{ij}^{(dcr)}(\gamma).
\end{equation*}
and where $Z$ is a known linear combination of the matrices
\begin{equation*}
Z^{\tB} := \frac{1}{I}\sum_{i=1}^I \overline{X}_{i,\bullet}^{(d)}\overline{X}_{i,\bullet}^{(d)\top}, \quad Z^{\tS} := \frac{1}{J}\sum_{j=1}^J \overline{X}_{\bullet,j}^{(d)}\overline{X}_{\bullet,j}^{(d)\top}, \quad
Z^{\tBS} := \frac{1}{IJ}\sum_{i=1}^I\sum_{j=1}^J X_{ij}^{(dcr)}X_{ij}^{(dcr)\top}.
\end{equation*}
Expanding,
\begin{equation}
\label{eq:variance_terms_of_residual}
\begin{split}
\nu^\theta_{\gamma}(e_{ij}) & := \beta_\gamma^\top Z^\theta \beta_\gamma - 2u^{\theta\top}_\gamma \beta_\gamma + \nu^\theta_{\gamma} \\
\xi^\theta_{\gamma,\gamma'}(e_{ij}) & := (\beta_\gamma - \beta_\gamma')^\top Z^\theta (\beta_\gamma - \beta_\gamma') - 2(u^{\theta}_\gamma - u^{\theta}_{\gamma'})^\top (\beta_\gamma - \beta_{\gamma'}) + \xi^\theta_{\gamma,\gamma'} \\
& = \beta_\gamma^\top Z^\theta \beta_{\gamma} - 2\beta_{\gamma'}^\top Z^\theta \beta_{\gamma} + \beta_{\gamma'}^\top Z^\theta \beta_{\gamma'} \\ 
& - 2\left(u_\gamma^{\theta \top} \beta_\gamma + u_{\gamma'}^{\theta \top} \beta_{\gamma'} - u_{\gamma}^{\theta \top} \beta_{\gamma'} - u_{\gamma'}^{\theta \top} \beta_{\gamma}   \right) + \xi^\theta_{\gamma,\gamma'}
\end{split}
\end{equation}
Let $\bm{\beta} \in \bR^{4d}$ be the concatenation of the adjustments $\beta_{\tr},\beta_{\ib},\beta_{\is},\beta_{\cc}.$ We can thus see $\Var(\hat \tau_{c,\text{int}}(\bm{\beta}))$ is a quadratic form in $\bm{\beta}$ which can be written as a block equation
\[
\bm{\beta}^\top \bm{Z} \bm{\beta} - 2\bm{u}^\top \bm{\beta} + \Var(\hat \tau_{c}),
\]
with $\bm{Z} \in \mathbb{R}^{4d \times 4d}, \bm{u} \in \mathbb{R}^{4d}.$ Note that $\bm{Z}$ only depends on the covariates while $\bm{u}$ only depends on potential outcome contrasts of the form $u^{\theta}_\gamma - u^{\theta}_{\gamma'}$ which are estimable via the formulas in \cref{sec:methodology}. Thus, the optimal regression-adjusted coefficient is estimable from data. The next several subsections detail the explicit forms of the adjustment for the direct, total, and spillover effects, alongside natural estimators. 

The resulting estimators are generally different from running interacted least squares (taking $\beta_\gamma$ to be the OLS coefficient in each group), which we will call the Lin estimator \cite{lin2013agnostic}. \cref{fig:interacted_total_effect_comparison} gives a simulation setting where the optimal interacted adjustment for the total effect is considerably more efficient than the Lin procedure. In other settings, we find that the optimal interacted procedure and the Lin estimator are comparable; the Lin estimator has slightly better small sample performance, especially when one of the $\tr,\ib,\is,\cc$ blocks is small. We leave a detailed numerical and theoretical study to future work. The tools developed in this paper should be enough to prove the analogous results in \cref{sec:asymptotic_theory}.

\begin{figure}
    \centering
    \includegraphics[width=0.65\linewidth]{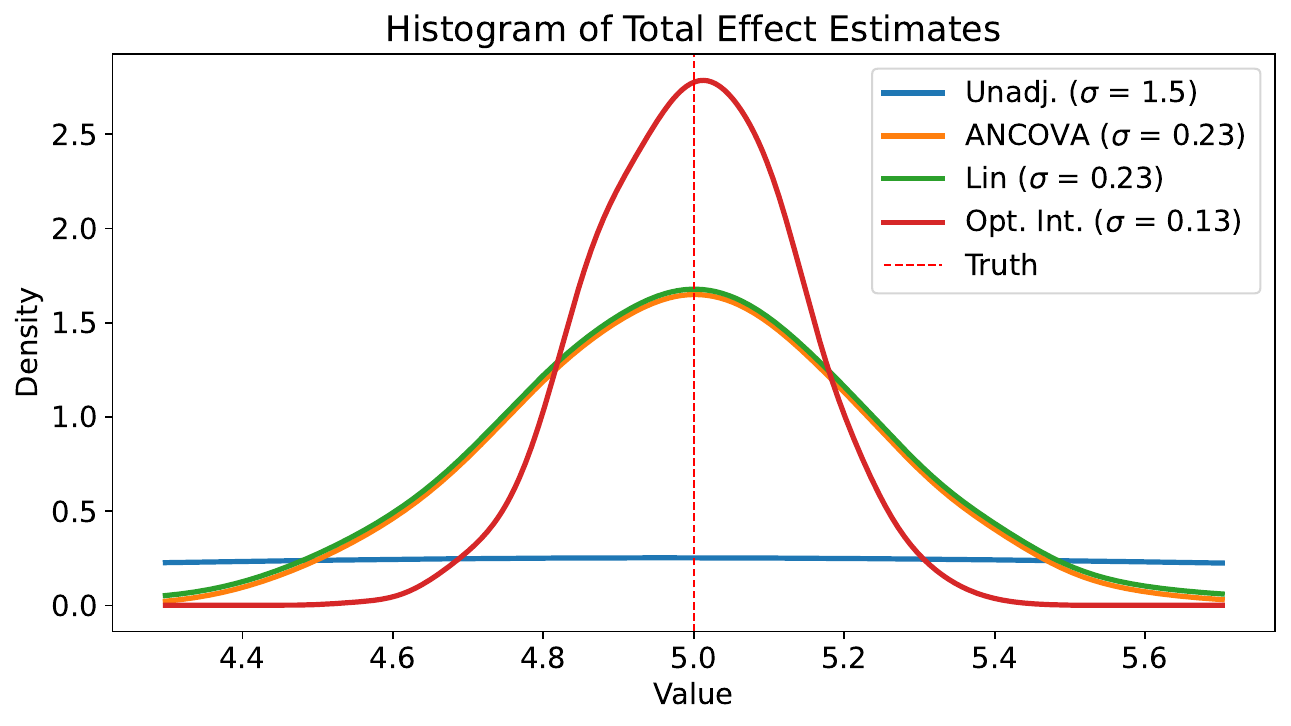}
    \caption{KDE estimate of sampling distributions of the Lin estimator, optimal interacted adjustment, ANCOVA, and unadjusted estimators in a synthetic setting, where the treatment indicators were rerandomized over 2000 Monte Carlo runs. The potential outcomes were taken as $y_{ij}(\cc) = \mu^\tB_i + \mu^\tS_j + \e_{ij}$, where $\mu_i^\tB,\mu^\tS_j, \e_{ij}$ are mutually independent standard normals and $y_{ij}(\tr) = y_{ij}(\cc) + 5.$ $X_{ij} \in \bR$ is taken to be a noisy observation of the control potential outcome $y_{ij}(\cc)$. The relative efficiency of the optimal interacted estimator with respect to the Lin estimator is $\sim 0.3$.}
    \label{fig:interacted_total_effect_comparison}
\end{figure}


\subsection{Interacted Estimators for the Direct Effect}

Recall that the variance expression for the direct effect estimator can be written in \eqref{eq:direct_effect_variance_formula}:

\begin{align*}
& \frac{1}{I_T J_T }\nu^{\tBS}_\tr + \frac{1}{I_T J_C }\nu^{\tBS}_\ib + \frac{1}{I_C J_T }\nu^{\tBS}_\is + \frac{1}{I_C J_C }\nu^{\tBS}_\cc \\   
& + \frac{I_C}{I_{T}I}\xi^{\tB}_{\tr,\ib}
  - \frac{1}{I}\xi^{\tB}_{\tr,\is} 
  + \frac{1}{I}\xi^{\tB}_{\tr,\cc}
 + \frac{1}{I} \xi^{\tB}_{\ib,\is}
 - \frac{1}{I} \xi^{\tB}_{\ib,\cc}
 + \frac{I_T}{I_C}\frac{1}{I}\xi^{\tB}_{\is,\cc} \\
 &  - \frac{1}{J}   \xi^{\tS}_{\tr,\ib}  
+ \frac{J_C}{J_T}\frac{1}{J} \xi^{\tS}_{\tr,\is}
+ \frac{1}{J} \xi^{\tS}_{\tr,\cc} 
+ \frac{1}{J} \xi^{\tS}_{\ib,\is} 
+ \frac{J_T}{J_C}\frac{1}{J} \xi^{\tS}_{\ib,\cc} 
- \frac{1}{J} \xi^{\tS}_{\is,\cc} \\
&  - \frac{I_C}{I_T}\frac{1}{IJ}  \xi^{\tBS}_{\tr,\ib} 
  - \frac{J_C}{J_T}\frac{1}{IJ}  \xi^{\tBS}_{\tr,\is} 
   - \frac{1}{IJ}  \xi^{\tBS}_{\tr,\cc}
  - \frac{1}{IJ}  \xi^{\tBS}_{\ib,\is}
 - \frac{J_T}{J_C}\frac{1}{IJ}  \xi^{\tBS}_{\ib,\cc}
  - \frac{I_T}{I_C}\frac{1}{IJ} \xi^{\tBS}_{\is,\cc}.
\end{align*}

We record the formula for $\bm{Z}_{\dir} := \bm{Z}$ here, as a $4\times 4$ block matrix with block entries $\bm{Z}_{\gamma, \gamma'}$. 
By symmetry, $\bm{Z}_{\gamma,\gamma'} = \bm{Z}_{\gamma',\gamma}$ so we only record the upper diagonal block entries.
\begin{align*}
\bm{Z}_{\tr,\tr} & :=  \left(\frac{1}{I_TJ_T} - \frac{I_C}{I_T IJ} - \frac{J_C}{J_T IJ} - \frac{1}{IJ}\right) Z^\tBS + \frac{I_C}{I_T I} Z^\tB + \frac{J_C}{J_T J} Z^\tS \\
\bm{Z}_{\tr,\ib} & := -\frac{I_C}{I_T I}Z^\tB + \frac{1}{J} Z^\tS + \frac{I_C}{I_T} \frac{1}{IJ}Z^\tBS\\
\bm{Z}_{\tr,\is} & :=  \frac{1}{I}Z^\tB - \frac{J_C}{J_T}\frac{1}{J} Z^\tS + \frac{J_C}{J_T} \frac{1}{IJ}Z^\tBS\\
\bm{Z}_{\tr,\cc} & := -\frac{1}{I}Z^\tB - \frac{1}{J} Z^\tS + \frac{1}{IJ}Z^\tBS\\
\bm{Z}_{\ib,\ib} & := \left(\frac{1}{I_TJ_C} - \frac{I_C}{I_T IJ} - \frac{1}{IJ}  - \frac{J_T}{J_C IJ}\right) Z^\tBS + \frac{I_C}{I_T I} Z^\tB + \frac{J_T}{J_C J} Z^\tS \\
\bm{Z}_{\ib,\is} & :=  -\frac{1}{I}Z^\tB - \frac{1}{J} Z^\tS + \frac{1}{IJ}Z^\tBS\\
\bm{Z}_{\ib,\cc} & :=  \frac{1}{I}Z^\tB - \frac{J_T}{J_C}\frac{1}{J} Z^\tS + \frac{J_T}{J_C} \frac{1}{IJ}Z^\tBS\\
\bm{Z}_{\is,\is} & := \left(\frac{1}{I_CJ_T} - \frac{J_C}{J_T IJ} - \frac{1}{IJ}  - \frac{I_T}{I_C IJ}\right) Z^\tBS + \frac{I_T}{I_C I} Z^\tB + \frac{J_C}{J_T J} Z^\tS \\
\bm{Z}_{\is,\cc} & := -\frac{I_T}{I_C I}Z^\tB + \frac{1}{J} Z^\tS + \frac{I_T}{I_C} \frac{1}{IJ}Z^\tBS\\
\bm{Z}_{\cc,\cc} & := \left(\frac{1}{I_CJ_C} - \frac{1}{IJ} - \frac{J_T}{J_C IJ}  - \frac{I_T}{I_C IJ}\right) Z^\tBS + \frac{I_T}{I_C I} Z^\tB + \frac{J_C}{J_T J} Z^\tS.
\end{align*}

Simplifying the diagonal terms slightly,
\begin{align*}
Z_{\tr,\tr} & := \frac{I_C J_C}{I_T J_T} \frac{1}{IJ} Z^\tBS + \frac{I_C}{I_T I} Z^\tB + \frac{J_C}{J_T J} Z^\tS \\
Z_{\ib,\ib} & := \frac{I_C J_T}{I_T J_C} \frac{1}{IJ} Z^\tBS + \frac{I_C}{I_T I} Z^\tB + \frac{J_T}{J_C J} Z^\tS \\
Z_{\is,\is} & := \frac{I_T J_C}{I_C J_T} \frac{1}{IJ} Z^\tBS + \frac{I_T}{I_C I} Z^\tB + \frac{J_C}{J_T J} Z^\tS \\
Z_{\cc,\cc} & := \frac{I_T J_T}{I_C J_C} \frac{1}{IJ} Z^\tBS + \frac{I_T}{I_C I} Z^\tB + \frac{J_C}{J_T J} Z^\tS.
\end{align*}

Now, we record the block formula for $\bm{u}$ where the block entries $\bm{u}_\gamma$ are the terms multiplied by $\beta_\gamma$.

\begin{align*}
\bm{u}_{\tr} & := \frac{1}{I_TJ_T} u_{\tr}^\tBS - \frac{I_C}{I_T IJ}(u_{\tr}^\tBS - u_{\ib}^\tBS) - \frac{J_C}{J_T IJ}(u_{\tr}^\tBS - u_{\is}^\tBS) - \frac{1}{IJ}(u_{\tr}^\tBS - u_{\cc}^\tBS)\\
& + \frac{I_C}{I_T I} (u_{\tr}^\tB - u_{\ib}^\tB) - \frac{1}{I} (u_{\tr}^\tB - u_{\is}^\tB) + \frac{1}{I} (u_{\tr}^\tB - u_{\cc}^\tB) \\
& - \frac{1}{J} (u_{\tr}^\tS - u_{\ib}^\tS) + \frac{J_C}{J_T J} (u_{\tr}^\tS - u_{\is}^\tS) + \frac{1}{J} (u_{\tr}^\tS - u_{\cc}^\tS)\\
\bm{u}_{\ib} & := \frac{1}{I_TJ_C} u_{\ib}^\tBS - \frac{I_C}{I_T IJ}(u_{\ib}^\tBS - u_{\tr}^\tBS) - \frac{1}{IJ}(u_{\ib}^\tBS - u_{\is}^\tBS) - \frac{J_T}{J_CIJ}(u_{\ib}^\tBS - u_{\cc}^\tBS)\\
& + \frac{I_C}{I_T I} (u_{\ib}^\tB - u_{\tr}^\tB) + \frac{1}{I} (u_{\ib}^\tB - u_{\is}^\tB) - \frac{1}{I} (u_{\ib}^\tB - u_{\cc}^\tB) \\
& - \frac{1}{J} (u_{\ib}^\tS - u_{\tr}^\tS) + \frac{1}{J} (u_{\ib}^\tS - u_{\is}^\tS) + \frac{J_T}{J_CJ} (u_{\ib}^\tS - u_{\cc}^\tS) \\
\bm{u}_{\is} & := \frac{1}{I_CJ_T} u_{\is}^\tBS - \frac{J_C}{J_T IJ}(u_{\is}^\tBS - u_{\tr}^\tBS) - \frac{1}{IJ}(u_{\is}^\tBS - u_{\ib}^\tBS) - \frac{I_T}{I_CIJ}(u_{\is}^\tBS - u_{\cc}^\tBS)\\
& - \frac{1}{I} (u_{\is}^\tB - u_{\tr}^\tB) + \frac{1}{I} (u_{\is}^\tB - u_{\ib}^\tB) + \frac{I_T}{I_CI} (u_{\is}^\tB - u_{\cc}^\tB) \\
& + \frac{J_C}{J_TJ} (u_{\is}^\tS - u_{\tr}^\tS) + \frac{1}{J} (u_{\is}^\tS - u_{\ib}^\tS) - \frac{1}{J} (u_{\is}^\tS - u_{\cc}^\tS) \\
\bm{u}_{\cc} & := \frac{1}{I_CJ_C} u_{\cc}^\tBS - \frac{1}{IJ}(u_{\cc}^\tBS - u_{\tr}^\tBS) - \frac{J_T}{J_CIJ}(u_{\cc}^\tBS - u_{\ib}^\tBS) - \frac{I_T}{I_CIJ}(u_{\cc}^\tBS - u_{\is}^\tBS)\\
& + \frac{1}{I} (u_{\cc}^\tB - u_{\tr}^\tB) - \frac{1}{I} (u_{\cc}^\tB - u_{\ib}^\tB) + \frac{I_T}{I_CI} (u_{\cc}^\tB - u_{\is}^\tB) \\
& + \frac{1}{J} (u_{\cc}^\tS - u_{\tr}^\tS) + \frac{J_T}{J_CJ} (u_{\cc}^\tS - u_{\ib}^\tS) - \frac{1}{J} (u_{\cc}^\tS - u_{\is}^\tS).
\end{align*}

The optimal adjustment is then given by $\tilde{\bm{\beta}}_{\dir} := \bm{Z}_{\dir}^{-1} \bm{u}_{\dir}$. 

\subsection{Interacted Estimators for the Total Effect}

The full formula for the total effect variance is given by \eqref{eq:var_total_effect}:

\begin{align}
\begin{split}
& \frac{1}{I_T}\nu_{\tr}^{\tB} + \frac{1}{J_T}\nu_{\tr}^{\tS} + \frac{1}{IJ}\left( \frac{I_CJ_C}{I_TJ_T} - 1\right)\nu_{\tr}^{\tBS}\\
& + \frac{1}{I_C}\nu_{\cc}^{\tB} + \frac{1}{J_C}\nu_{\cc}^{\tS} + \frac{1}{IJ}\left(\frac{I_TJ_T}{I_CJ_C} - 1\right)\nu_{\cc}^{\tBS}\\
& - \frac{1}{I}\xi_{\tr,\cc}^{\tB} - \frac{1}{J}\xi_{\tr,\cc}^{\tS} + \frac{1}{IJ}
\xi_{\tr,\cc}^{\tBS}.
\end{split}
\end{align}


We can express the optimal adjustment $\bm{\beta} := [\beta_\tr \ \beta_\cc]$ as $\bm{\beta} = \bm{Z}^{-1} \bm{u}$, where $\bm{Z}$ is a $2 \times 2$ block matrix with block entries
\begin{align*}
\bm{Z}_{\tr,\tr} & := \frac{I_C}{I_T I}Z^{\tB} + \frac{J_C}{J_T J} Z^\tS + \frac{I_CJ_C}{I_TJ_T}\frac{1}{IJ} Z^\tBS \\
\bm{Z}_{\tr,\cc} & := \frac{1}{I}Z^\tB + \frac{1}{J}Z^\tS - \frac{1}{IJ}Z^\tBS \\
\bm{Z}_{\cc,\cc} & := \frac{I_T}{I_C I}Z^{\tB} + \frac{J_T}{J_C J} Z^\tS + \frac{I_TJ_T}{I_CJ_C}\frac{1}{IJ} Z^\tBS
\end{align*}
and $\bm{u}$ is given by
\begin{align*}
\bm{u}_{\tr} & := \frac{1}{I_T} u^\tB_\tr -\frac{1}{I} (u^\tB_\tr - u^\tB_\cc) \\
& + \frac{1}{J_T} u^\tS_\tr -\frac{1}{J} (u^\tS_\tr - u^\tS_\cc) \\
& + \frac{1}{IJ}\left( \frac{I_CJ_C}{I_TJ_T} - 1\right)u^\tBS_\tr + \frac{1}{IJ} (u^\tBS_\tr - u^\tBS_\cc) \\
\bm{u}_{\cc} & := \frac{1}{I_C} u^\tB_\cc -\frac{1}{I} (u^\tB_\cc - u^\tB_\tr) \\
& + \frac{1}{J_C} u^\tS_\cc -\frac{1}{J} (u^\tS_\cc - u^\tS_\tr) \\
& + \frac{1}{IJ}\left( \frac{I_TJ_T}{I_CJ_C} - 1\right)u^\tBS_\cc + \frac{1}{IJ} (u^\tBS_\cc - u^\tBS_\tr).
\end{align*}

This simplifies to

\begin{align*}
\bm{u}_{\tr} & := \frac{I_C}{I_T I} u^\tB_\tr + \frac{1}{I} u^\tB_\cc \\
& + \frac{J_C}{J_T J} u^\tS_\tr + \frac{1}{J} u^\tS_\cc \\
& + \frac{1}{IJ}\left( \frac{I_CJ_C}{I_TJ_T}\right)u^\tBS_\tr - \frac{1}{IJ} u^\tBS_\cc \\
\bm{u}_{\cc} & := \frac{I_T}{I_C I} u^\tB_\cc + \frac{1}{I} u^\tB_\tr \\
& + \frac{J_T}{J_C J} u^\tS_\cc + \frac{1}{J} u^\tS_\tr \\
& + \frac{1}{IJ}\left( \frac{I_TJ_T}{I_CJ_C} \right)u^\tBS_\cc - \frac{1}{IJ} u^\tBS_\tr.
\end{align*}

If we replace $y_{ij}(\gamma)$ by the residuals $y_{ij}(\gamma) - X_{ij}^\top \beta_\gamma$, we have
\begin{align*}
\hat{\bm{u}}_{\tr}(e) - \hat{\bm{u}}_{\tr} & := \frac{I_C}{I_T I} \hat Z^\tB_\tr \beta_\tr + \frac{1}{I} \hat Z^\tB_\cc \beta_\cc \\
& + \frac{J_C}{J_T J} \hat{Z}^\tS_\tr \beta_\tr + \frac{1}{J} \hat{Z}^\tS_\cc \beta_\cc \\
& + \frac{1}{IJ}\left( \frac{I_CJ_C}{I_TJ_T}\right)\hat{Z}^\tBS_\tr \beta_\tr - \frac{1}{IJ} \hat{Z}^\tBS_\cc \beta_\cc \\
\hat{\bm{u}}_{\cc}(e) - \hat{\bm{u}}_{\cc} & := \frac{I_T}{I_C I} \hat Z^\tB_\cc \beta_\cc + \frac{1}{I} \hat Z^\tB_\tr \beta_\tr  \\
& + \frac{J_T}{J_C J} \hat Z^\tS_\cc \beta_\cc  + \frac{1}{J} \hat Z^\tS_\tr \beta_\tr  \\
& + \frac{1}{IJ}\left( \frac{I_TJ_T}{I_CJ_C} \right) \hat Z^\tBS_\cc \beta_\cc  - \frac{1}{IJ} \hat Z^\tBS_\tr \beta_\tr.
\end{align*}

We may write this as $\hat{\bm{u}}(e) - \hat{\bm{u}} = \hat{\bm{Z}} \bm{\beta}$ where
\begin{align*}
\hat{\bm{Z}}_{\tr,\tr} & := \frac{I_C}{I_T I}\hat Z_{\tr}^{\tB} + \frac{J_C}{J_T J} \hat Z_{\tr}^\tS + \frac{I_CJ_C}{I_TJ_T}\frac{1}{IJ} \hat Z_{\tr}^\tBS \\
\hat{\bm{Z}}_{\tr,\cc} & := \frac{1}{I}\hat Z_{\cc}^\tB + \frac{1}{J}\hat Z_{\cc}^\tS - \frac{1}{IJ}\hat Z_{\cc}^\tBS \\
\hat{\bm{Z}}_{\cc,\tr} & := \frac{1}{I}\hat Z_{\tr}^\tB + \frac{1}{J}\hat Z_{\tr}^\tS - \frac{1}{IJ}\hat Z_{\tr}^\tBS \\
\hat{\bm{Z}}_{\cc,\cc} & := \frac{I_T}{I_C I}\hat Z_{\cc}^{\tB} + \frac{J_T}{J_C J} \hat Z_{\cc}^\tS + \frac{I_TJ_T}{I_CJ_C}\frac{1}{IJ} \hat Z_{\cc}^\tBS
\end{align*}

It is instructive to compare this to least squares in each group. The OLS objective with constants is to minimize in $\beta_{\tr}$
\begin{align*}
\frac{1}{I_T J_T}\sum_{(i,j) \in \tr} (y_{ij}^{(d)}(\tr) - X_{ij}^{(d)\top} \beta_\tr)^2 & =  \frac{1}{I_T J_T}\sum_{(i,j) \in \tr} (y_{ij}^{(dcr)}(\tr) - X_{ij}^{(dcr)\top} \beta_\tr)^2 \\
& + \frac{1}{I_T} \sum_{i \in \tr} (\bar y_{i,\bullet}^{(d)}(\tr) - \bar X_{i,\bullet}^{(d)\top} \beta_\tr)^2 \\
& + \frac{1}{J_T} \sum_{j \in \tr} (\bar y_{\bullet,j}^{(d)}(\tr) - \bar X_{\bullet,j}^{(d)\top} \beta_\tr)^2,
\end{align*}
which suggests the population objective, up to constant factors,
\begin{equation}
\begin{split}
& \frac{1}{I_T}\nu_{\tr}^{\tB} + \frac{1}{I_T}\nu_{\tr}^{\tS} + \frac{1}{I_T}\nu_{\tr}^{\tBS}\\
& + \frac{1}{I_T}\nu_{\cc}^{\tB} + \frac{1}{I_T}\nu_{\cc}^{\tS} + \frac{1}{I_T}\nu_{\cc}^{\tBS}
\end{split}
\end{equation}
So compared to the optimal adjustment, the OLS upweights the decentered term and underweights the row and column means. Moreover, there is no interaction term. 

\subsection{Interacted Estimators for the Buyer Spillover Effect}

For the buyer spillover estimator $\hat{\tau}_{\text{bs}}(\beta)$ (\eqref{eq:imputation_estimators} with $c = [0,1,0,-1]$), the variance is given by
\begin{align*}
    \Var(\tau_{\text{bs}}(\beta)) & =  \frac{1}{I_T}\nu^{\tB}_\ib  +  \frac{1}{I_C}\nu^{\tB}_\cc + \frac{J_T}{I_TJ_C J}\nu^{\tBS}_\ib + \frac{J_T}{I_CJ_C J}\nu^{\tBS}_\cc \\
    - & \frac{1}{I}\xi^{\tB}_{\ib,\cc} + \frac{J_T}{J_CJ} \xi^{\tS}_{\ib,\cc} - \frac{J_T}{J_C}\frac{1}{IJ}\xi^{\tBS}_{\ib,\cc}.
\end{align*}

We can express the optimal adjustment $\bm{\beta} := [\beta_\ib \ \beta_\cc]$ as $\bm{\beta} = \bm{Z}^{-1} \bm{u}$, where $\bm{Z}$ is a $2 \times 2$ block matrix with block entries
\begin{align*}
\bm{Z}_{\ib,\ib} & := \frac{1}{I_T}Z^{\tB} + \frac{J_T}{I_T J_C J} Z^\tBS \\
\bm{Z}_{\ib,\cc} & := \frac{1}{I}Z^\tB - \frac{J_T}{J_C J} Z^\tS + \frac{J_T}{J_C}\frac{1}{IJ} Z^\tBS \\
\bm{Z}_{\cc,\cc} & := \frac{1}{I_C} Z^{\tB} + \frac{J_T}{I_C J_C J} Z^\tBS.
\end{align*}
and $\bm{u}$ is given by
\begin{align*}
\bm{u}_{\ib} & := \frac{1}{I_T} u^\tB_\ib -\frac{1}{I} (u^\tB_\ib - u^\tB_\cc) \\
& + \frac{J_T}{J_C J} (u^\tS_\ib - u^\tS_\cc) \\
& + \frac{J_T}{I_TJ_CJ}u^\tBS_\ib - \frac{J_T}{J_C IJ} (u^\tBS_\ib - u^\tBS_\cc) \\
\bm{u}_{\cc} & := \frac{1}{I_C} u^\tB_\cc -\frac{1}{I} (u^\tB_\cc - u^\tB_\ib) \\
& + \frac{J_T}{J_C J} (u^\tS_\cc - u^\tS_\ib) \\
& + \frac{J_T}{I_CJ_CJ}u^\tBS_\cc - \frac{J_T}{J_C IJ} (u^\tBS_\cc - u^\tBS_\ib) \\
\end{align*}

\subsection{Interacted Estimators for the Seller Spillover Effect}
In this case, the variance is given by
\begin{align*}
    \Var(\tau_{\text{ss}}(\beta)) & =  \frac{1}{J_T}\nu^{\tS}_\is  +  \frac{1}{J_C}\nu^{\tS}_\cc + \frac{I_T}{J_TI_C I}\nu^{\tBS}_\is + \frac{I_T}{J_CI_C I}\nu^{\tBS}_\cc \\
    + & \frac{I_T}{I_C}\frac{1}{I}\xi^{\tB}_{\is,\cc} - \frac{1}{J} \xi^{\tS}_{\is,\cc} - \frac{I_T}{I_C}\frac{1}{IJ} \xi^{\tBS}_{\is,\cc},
\end{align*}
and the adjustment has the components
\begin{align*}
\bm{Z}_{\is,\is} & := \frac{1}{J_T}Z^{\tS} + \frac{I_T}{J_T I_C I} Z^\tBS \\
\bm{Z}_{\is,\cc} & := -\frac{I_T}{I_CI}Z^\tB + \frac{1}{J} Z^\tS + \frac{I_T}{I_C}\frac{1}{IJ} Z^\tBS \\
\bm{Z}_{\cc,\cc} & := \frac{1}{J_C} Z^{\tS} + \frac{I_T}{J_C I_C I} Z^\tBS.
\end{align*}
and $\bm{u}$ is given by
\begin{align*}
\bm{u}_{\is} & :=  \frac{I_T}{I_C I} (u^\tB_\is - u^\tB_\cc) \\
& + \frac{1}{J_T} u^\tS_\is  - \frac{1}{J} (u^\tS_\is - u^\tS_\cc) \\
& + \frac{I_T}{J_TI_CI}u^\tBS_\is - \frac{I_T}{I_C IJ} (u^\tBS_\is - u^\tBS_\cc) \\
\bm{u}_{\cc} & := \frac{I_T}{I_C I} (u^\tB_\cc - u^\tB_\is) \\
& + \frac{1}{J_C} u^\tS_\cc  - \frac{1}{J} (u^\tS_\cc - u^\tS_\is) \\
& + \frac{I_T}{J_CI_CI}u^\tBS_\cc - \frac{I_T}{I_C IJ} (u^\tBS_\cc - u^\tBS_\is).
\end{align*}

\end{document}